\documentclass{birkjour}
\usepackage{amsmath}
\usepackage{amssymb}
\usepackage{verbatim}
\usepackage{color,graphicx}
\usepackage{amsthm}                
\usepackage{algpseudocode}  
\usepackage{algorithm}
\usepackage{caption}
\usepackage{subcaption}
\usepackage{appendix}
\usepackage[colorlinks]{hyperref}
\usepackage{mathtools, nccmath}
\usepackage{setspace}
\usepackage{dsfont}
\usepackage{bbold}

 \newtheorem{thm}{Theorem}[section]
 \newtheorem{cor}[thm]{Corollary}
 \newtheorem{lem}[thm]{Lemma}
 \newtheorem{prop}[thm]{Proposition}
 \newtheorem{alg}[thm]{Algorithm}
\newtheorem{ass}[thm]{Assumption}

 \theoremstyle{definition}
 \newtheorem{defn}[thm]{Definition}
 \theoremstyle{remark}
 \newtheorem{rem}[thm]{Remark}
 \newtheorem{ex}{Example}
 \numberwithin{equation}{section}

\newcommand{\vecc}{\boldsymbol}
\newcommand{\RR}{\mathbb{R}}      
\newcommand{\CC}{\mathbb{C}}      
\newcommand{\ZZ}{\mathbb{Z}}      

\begin{document}

%
%
%
%
%
%
%
%

\title[Homogenization for Generalized Langevin Equations]{Homogenization for Generalized Langevin Equations with Applications to Anomalous Diffusion}

\author[Soon Hoe Lim]{Soon Hoe Lim}

\address{%
Nordita, KTH Royal Institute of Technology and Stockholm University\\
Roslagstullsbacken 23\\
SE-106 91 Stockholm\\
Sweden}

\email{soon.hoe.lim@su.se}


\author{Jan Wehr}
\address{Department of Mathematics and Program in Applied Mathematics \br
University of Arizona\br
Tucson, AZ 85721-0089\br
USA}
\email{wehr@math.arizona.edu}

\author{Maciej Lewenstein}
\address{ICFO - Institut de Ci\'encies Fot\'oniques\br
The Barcelona Institute of Science and Technology\br 
Av. Carl Friedrich Gauss 3\br
08860 Castelldefels (Barcelona)\br
Spain \\}
\address{ICREA\br
Pg. Lluis Companys 23\br 
08010 Barcelona\br 
Spain}
\email{maciej.lewenstein@icfo.eu}



\submitted{\today}


\begin{abstract}
We study  homogenization  for a class of generalized Langevin equations (GLEs) with state-dependent coefficients and exhibiting multiple time scales. In addition to the small mass limit, we focus on homogenization limits, which involve taking to zero the inertial time scale and, possibly, some of the memory time scales and noise correlation time scales. The latter are meaningful limits for a class of GLEs modeling anomalous diffusion. We find that, in general, the limiting stochastic differential equations (SDEs) for the slow degrees of freedom contain non-trivial drift correction terms and are driven by non-Markov noise processes. These results follow from a general homogenization theorem stated and proven here. We illustrate them using stochastic models of particle diffusion.
\end{abstract}

\maketitle
\tableofcontents


\section{Introduction}
\subsection{Motivation}

Most of the mathematical models of diffusion phenomena use noise which is white (i.e. uncorrelated), or Markovian \cite{nelson1967dynamical}.  The present paper is a step towards removing this limitation.  The diffusion models studied here are driven by noises, belonging to a wide class of non-Markov processes.  
A standard example of Markovian noise is a multidimensional Ornstein-Uhlenbeck process.
An important class of Gaussian stochastic processes is obtained by linear transformations of multidimensional Ornstein-Uhlenbeck processes.  The covariance (equal to correlation in the case of zero mean) of such a process is  a linear combination of exponentials decaying and possibly oscillating on different time scales, and its spectral density (power spectrum) is a ratio of two semi-positive defined polynomials \cite{doob1953stochastic}. In cases when the polynomial in the denominator has degenerate zeros, the covariance contains products of exponentials and polynomials in time.  This is a very general class of processes---every stationary Gaussian process whose covariance is a Bohl function (see Section 2) can be obtained as a linear transformation of an Ornstein-Uhlenbeck process in some (finite) dimension.  In general, these processes are not Markov.

Let us mention here the seminal result by L.A. Khalfin from 1957 \cite{khalfin1958contribution}, who showed, quite generally that in any system with energy spectrum bounded from below (which is a necessary condition for the physical stability), correlations must decay no faster than according to a power law. To this day this result provides inspirations and motivations for further studies in the context of thermalization \cite{tavora2016inevitable}, cooling of atoms in photon reservoirs \cite{lewenstein1993cooling}, decay of metastable states as monitored by luminescence \cite{rothe2006violation}, or quantum anti-Zeno effect (c.f. \cite{peres1980nonexponential,lewenstein2000quantum}), to name a few examples.  Khalfin's result further motivates studying systems with non-Markovian noise, as most natural examples of strongly correlated processes do not satisfy Markov property.

While the noise processes studied here have exponentially decaying covariances, their class is very rich and they may be useful in approximating strongly correlated noises on time intervals, relevant for studied phenomena \cite{siegle2011markovian}.  In addition, as discussed in more detail later, generalization of the method applied here may lead to a representation of a class of noises whose covariances decay as powers (see Remark \ref{rem_inf_dim}). Also, the representation of spectral density of the noise processes as ratio of two polynomials is convenient in applications, in particular for solving the problem of predicting (in the least mean square sense) a colored noise process given observations on a finite segment of the past or on the full past \cite{doob1953stochastic}.



\subsection{Definitions and Models} \label{intro}

We consider the following stochastic model for a particle (for instance, Brownian particle or a tagged tracer particle) interacting with the environment (for instance, a heat bath or a viscous fluid). Let $\vecc{x}_t \in \RR^d$ denote the particle's position, where $t \geq 0$ denotes time and $d$ is a positive integer. The evolution of the particle's velocity, $\vecc{v}_t := \dot{\vecc{x}}_t \in \RR^d$, is described by the following {\it generalized Langevin equation (GLE)}:
\begin{equation}
m d\vecc{v}_t = \vecc{F}_0 \left(t,\vecc{x}_t,\vecc{v}_t,\vecc{\eta}_t \right)dt + \vecc{F}_1\left(t, \{\vecc{x}_s,\vecc{v}_s\}_{s \in [0,t]}, \vecc{\xi}_t \right)dt + \vecc{F}_e(t, \vecc{x}_t)dt. \label{gle2}
\end{equation}
In the above, $m>0$ is the particle's mass, $\vecc{\eta}_t$ is a $k$-dimensional Gaussian white noise satisfying $E[\vecc{\eta}_t] = \vecc{0}$ and $E[\vecc{\eta}_t \vecc{\eta}_s^*] = \delta(t-s)\vecc{I}$, and $\vecc{\xi}_t$ is a colored noise process independent of $\vecc{\eta}_t$. Here and throughout the paper, the superscript $^*$ denotes transposition of matrices or vectors, $\vecc{I}$ denotes identity matrix of appropriate dimension, $E$ denotes expectation, and $\RR^+ := [0,\infty)$.  The initial data are random variables, $\vecc{x}_0 = \vecc{x}$, $\vecc{v}_0 = \vecc{v}$, independent of $\{\vecc{\xi}_t, t \in \RR^+ \}$ and $\{\vecc{\eta}_t, t \in \RR^+ \}$.

The three terms on the right hand side of \eqref{gle2} model forces of different physical natures acting on the particle. 
\begin{itemize}
\item[(i)] $\vecc{F}_e$ is an external force field, which may be conservative (potential) or not.
\item[(ii)] $\vecc{F}_0$ is a Markovian force of the form
\begin{equation}
\vecc{F}_0\left(t, \vecc{x}_t,\vecc{v}_t, \vecc{\eta}_t\right) dt = -\vecc{\gamma}_0(t, \vecc{x}_t)\vecc{v}_t dt + \vecc{\sigma}_0(t, \vecc{x}_t)d\vecc{W}^{(k)}_t,
\end{equation}
containing an instantaneous damping term and a multiplicative white noise term. The damping and noise coefficients, $\vecc{\gamma}_0: \RR^+ \times \RR^d \to \RR^{d \times d}$ and $\vecc{\sigma}_0: \RR^+ \times \RR^d \to \RR^{d \times k}$,  may depend on the particle's position and on time.  $\vecc{W}^{(k)}_t$ denotes a $k$-dimensional Wiener process---the time integral of the white noise $\vecc{\eta}_t$.
 \item[(iii)] $\vecc{F}_1$ is a non-Markovian force of the form
\begin{equation} \label{nonM_force}
\vecc{F}_1\left(t, \{\vecc{x}_s,\vecc{v}_s\}_{s \in [0,t]},\vecc{\xi}_t\right) = - \vecc{g}(t, \vecc{x}_t) \left( \int_{0}^{t} \vecc{\kappa}(t-s) \vecc{h}(s, \vecc{x}_s) \vecc{v}_s ds \right)  + \vecc{\sigma}(t, \vecc{x}_t) \vecc{\xi}_t,
\end{equation}
containing a non-instantaneous damping term, describing the delayed drag effects by the environment on the particle, and a multiplicative colored noise term.
The coefficients, $\vecc{g}: \RR^+ \times \RR^d \to \RR^{d\times q}$, $\vecc{h}: \RR^+ \times \RR^d \to \RR^{q \times d}$  and $\vecc{\sigma}: \RR^+ \times \RR^d  \to \RR^{d \times r} $,  depend in general on the particle's position and on time.  In the above, $q$ and $r$ are positive integers, and the memory function $\vecc{\kappa}: \RR \to  \RR^{q \times q}$ is  a real-valued function that decays sufficiently fast at infinities.  $\vecc{\xi}_t \in \RR^{r}$ is a mean-zero stationary Gaussian vector process, to be defined in detail later. The statistical properties of the process $\vecc{\xi}_t$ are completely determined by its (matrix-valued) {\it covariance function},
\begin{equation}
\vecc{R}(t):= E [\vecc{\xi}_t \vecc{\xi}^{*}_0] = \vecc{R}^{*}(-t) \in \RR^{r \times r},
\end{equation}
or equivalently, by its {\it spectral density}, $\vecc{\mathcal{S}}(\omega)$, i.e. the Fourier transform of $\vecc{R}(t)$ defined as:
\begin{equation} \label{spec_form}
\vecc{\mathcal{S}}(\omega) = \int_{-\infty}^{\infty} \vecc{R}(t) e^{-i\omega t} dt.
\end{equation}
\end{itemize}

For simplicity, we have omitted other forces such as the Basset force \cite{grebenkov2013hydrodynamic} from Eqn. \eqref{gle2}. Note that $\vecc{F}_0$ and $\vecc{F}_1$ describe two types of forces associated with different physical mechanisms. Of particular interest is when the noise term in $\vecc{F}_0$ and $\vecc{F}_1$ models environments of different nature (passive bath and active bath respectively \cite{dabelow2019irreversibility}) that the particle interacts with. 

As the name itself suggests, GLEs are  generalized versions of the Markovian Langevin equations, frequently employed to model physical systems.
 A basic form of the GLEs was first introduced by Mori in \cite{mori1965transport} and subsequently used in numerous statistical physics models \cite{Kubo_fd,toda2012statistical,Zwanzig1973}. The studies of GLEs have attracted increasing interest in recent years. We refer to, for instance, \cite{mckinley2009transient,lysy2016model,siegle2010markovian,hartmann2011balanced,goychuk2012viscoelastic,maes2013langevin,lei2016data,Wei2018,2018superstatistical}  for various applications of GLEs and \cite{Ottobre,mckinley2018anomalous,glatt2018generalized,leimkuhler2018ergodic}  for their asymptotic analysis. 
The main merit of GLEs from modeling point of view is that they take into account the effects of memory and the colored nature of  noise  on the dynamics of the system.

\begin{rem} \label{fdt_rem}
In general, there need not be any relation between $\vecc{\kappa}(t)$ and $\vecc{R}(t)$, or any relation between the damping coefficients and the noise coefficients appearing in the formula for $\vecc{F}_0$ and $\vecc{F}_1$. A particular but important case  that we will revisit often in this paper is the case when a {\it fluctuation-dissipation relation} holds. In this case, $\vecc{\gamma}_0$ is proportional to $\vecc{\sigma}_0 \vecc{\sigma}_0^*$, $\vecc{h} = \vecc{g}^*$, $\vecc{g}$ is proportional to $\vecc{\sigma}$ and  (without loss of generality\footnote{The factor $k_BT$, where $T$ is the absolute temperature and $k_B$ denotes the Boltzmann constant, is here set to $1$.  In general, it can be absorbed into either one of the coefficients $\vecc{g}$, $\vecc{h}$ or $\vecc{\sigma}$.}) $\vecc{R}(t) = \vecc{\kappa}(t)$. Studies of microscopic Hamiltonian models for open classical systems lead to  GLEs of the form \eqref{gle2} satisfying the above fluctuation-dissipation relation (see, for instance, Appendix A of \cite{LimWehr_Homog_NonMarkovian} or \cite{Cui18}). On another note, GLEs of the form \eqref{gle2} are extended versions of the ones studied in our previous work \cite{LimWehr_Homog_NonMarkovian} -- here the GLEs are generalized to include  a Markovian force, in addition to the non-Markovian one, as well as explicit time dependence in the coefficients. 
\end{rem}

As a motivation, we now provide and elaborate on  examples of systems that can be modeled by our GLEs. 

An important type of diffusion, which has been observed in many physical systems, from charge transport in
amorphous materials to intracellular particle motion in cytoplasm of living cells \cite{reverey2015superdiffusion}, is {\it ballistic diffusion}. It is a subclass of anomalous diffusions and is characterized by the property that the particle's long-time mean-squared displacement grows quadratically in time -- in contrast to linear growth in usual diffusion. There are many different theoretical models of anomalous diffusion with diverse properties, coming from different physical assumptions; see \cite{metzler2014anomalous} for a comprehensive survey.  In the following, we provide two GLE models that are employed to study such phenomena.
Their properties will be studied in Section \ref{sect_appl}, as an application of the results proven here.

\begin{ex}{\it Two GLE models for anomalous diffusion of a free Brownian particle in a heat bath.} \label{ex_mot}  
A large class of models for diffusive systems is described by the system of equations (for simplicity, we restrict to one dimension):
\begin{align}
dx_t &= v_t dt,  \label{mot1} \\
m dv_t &=  - \left(\int_0^t \kappa(t-s) v_s ds\right) dt +  \xi_t dt, \label{mot2}
\end{align}
where $x_t, \ v_t \in \RR$ are the position and velocity of the particle, $\kappa(t)$ is called the memory function, and $\xi_t$ is a mean-zero stationary Gaussian process.\\

\noindent Two particular GLE models are described by \eqref{mot1}-\eqref{mot2}, with:
\begin{itemize}
\item[(M1)] memory function of the bi-exponential form:
\begin{equation} \label{kappaeg}
\kappa(t) = \frac{ \Gamma_2^2(\Gamma_2 e^{-\Gamma_2 |t|} - \Gamma_1 e^{-\Gamma_1 |t|})}{2(\Gamma_2^2-\Gamma_1^2)},
\end{equation}
where the parameters satisfy $\Gamma_2 > \Gamma_1 > 0$, and  $\xi_t$ has the covariance function $R(t)= \kappa(t)$ and thus the spectral density,
\begin{equation}
\mathcal{S}(\omega) = \frac{ \Gamma_2^2 \omega^2}{(\omega^2+\Gamma_1^2)(\omega^2+\Gamma_2^2)}.
\end{equation}
This model is similar to the one first introduced and studied in \cite{bao2003ballistic}.
The noise with the above covariance function can be realized by the difference between two Ornstein-Uhlenbeck processes, with different damping rates, driven by the same white noise.  Various properties as well as applications of  GLEs of the form \eqref{mot1}-\eqref{mot2} were studied in \cite{bao2003ballistic,bao2005harmonic,siegle2010markovian}.
\item[(M2)]  memory function of the form:
\begin{equation}
\kappa(t) = \frac{1}{2}(\delta(t)-\Gamma_1 e^{-\Gamma_1 |t|}),
\end{equation}
where  $\Gamma_1 > 0$, and $\xi_t$ has the covariance function $R(t)= \kappa(t)$ and thus the spectral density,
\begin{equation}
\mathcal{S}(\omega) = \frac{ \omega^2}{\omega^2+\Gamma_1^2}.
\end{equation}
This model can be obtained from the one in (M1) by sending $\Gamma_2 \to \infty$ in the formula for $\kappa(t)$ in \eqref{kappaeg}. 

Observe that the spectral densities in both models share the same asymptotic behavior near $\omega = 0$, i.e. $\mathcal{S}(\omega) \sim \omega^2$ as $\omega \to 0$, contributing to the enhanced diffusion (super-diffusion) of the particle with mean-squared displacement growing as $t^2$ as $t \to \infty$ \cite{siegle2010markovian}. See Proposition \ref{asympbeh} for a precise argument. 
\end{itemize}
\end{ex}

Other examples of systems that can be modeled by our GLEs are multiparticle  systems with hydrodynamic interaction \cite{ermak1978brownian}, active matter systems \cite{sevilla2018non}, among others. Although our main results are applicable to these systems, we will not pursue the study of these systems here. 

\subsection{Goals, Organization and Summary of Results of the Paper} \label{goaletc}

\noindent  {\bf Goals of the Paper.}
We aim to derive homogenized models for a general class of GLEs (see Section \ref{sect_gles}), containing the examples (M1) and (M2) as special cases (see Corollary \ref{w1case} and  Corollary \ref{w2case}). This will allow us to gain insights into the stochastic dynamics of such systems, including many systems that exhibit anomalous diffusion (see discussion in the paragraph before Example \ref{ex_mot}) -- this is, in fact, the main motivation of the present paper. To the best of our knowledge, this is the first work that studies homogenization for GLE models describing anomalous diffusion. 

Given a GLE system, it is often desirable to work with  simpler, reduced models that  capture the essential features of its dynamics.  To obtain satisfactory and optimal models, one needs to take into account the trade-off between the simplicity and accuracy of the reduced models sought after. Indeed, one may find that a reduced model, while simplified, fails to give a physically correct model for describing a system of interest \cite{safdari2017aging}.  Two successful reductions were carried out in \cite{hottovy2015smoluchowski} for the case $\vecc{F}_1=\vecc{0}$ and in \cite{LimWehr_Homog_NonMarkovian} for the case $\vecc{F}_0 = \vecc{0}$.

One of our main goals in this paper is to devise and study new homogenization  procedures that yield reduced models retaining essential features of a more general class of models. This program is of importance for identification, parameter inference and uncertainty quantification of stochastic systems \cite{Picci2011_nicereview,hall2016uncertainty,lysy2016model,lei2016data} arising in  the studies of anomalous diffusion \cite{mckinley2009transient,morgado2002relation}, climate modeling \cite{gottwald2015stochastic,majda2001mathematical} and molecular systems \cite{cordoba2012elimination}, among others.  There is  increasing amount of effort striving to implement this or related programs, starting from microscopic models \cite{picci1992stochastic}, using various techniques \cite{givon2004extracting,Pavliotis,bo2016multiple, froyland2016trajectory,hartmann2011balanced}, for different systems of interest in the literature. The derived effective SDE models will be of particular interest for modelers of anomalous diffusion.  \\



\noindent {\bf Organization of the Paper.} 
The paper is organized as follows. We first present the application of the results obtained in the later sections (Section \ref{sect_newsmallmlimit} and Section \ref{sect_newhomogcase}) to study homogenization of generalized versions of the one-dimensional models (M1) and (M2) from Example \ref{ex_mot} in Section \ref{sect_appl}.  Since these results are easier to state and require minimal notation to understand, we have chosen to present them as early as possible to demonstrate the value of our study to application-oriented readers. The later sections study an extended, multi-dimensional version of the GLEs in Section \ref{sect_appl}. In Section \ref{sect_gles} we introduce the GLEs to be studied and revisit them from the perspective of  input-output stochastic dynamical systems exhibiting multiple time scales.   In Section \ref{sect_homogofGLE}, we discuss various ways of homogenizing  GLEs.  Following this discussion, we study the small mass limit of the GLEs in Section \ref{sect_newsmallmlimit}.    We introduce and study novel homogenization procedures for a class of GLEs in Section \ref{sect_newhomogcase}. We state conclusions and make final remarks in Section \ref{sect_conclusions}. Relevant technical details and supplementary materials are provided in the appendix.  In particular, we state a homogenization theorem for a general class of SDEs with state-dependent coefficients in  Appendix \ref{sect_generalhomogthm}. The proof of this theorem is given in Appendix \ref{proof_ch2}.  \\

\noindent {\bf Summary of the Main Results.}
For reader's convenience, below we list (not in exactly the same order as the results appear in the paper) and summarize the  main results obtained in the paper.
\begin{itemize}
\item The first main result is Theorem \ref{newsmallm}. It studies the small mass limit of the GLE described by \eqref{res_gle_smallmass1}-\eqref{res_gle_smallmass}. It states that the position process converges, in a strong pathwise sense, to a component of a higher dimensional process satisfying an It\^o SDE. The SDE contains non-trivial drift correction terms. We stress that, while being a component of a Markov process, the limiting position process itself is not Markov.  This is  in constrast to the nature of limiting processes obtained in earlier works, the difference  which holds interesting implications from a physical point of view (recall the discussion after Eqn. \eqref{spec_form}). Therefore,  Theorem \ref{newsmallm} constitutes a novel result, both mathematically and physically.
\item The second main result is Theorem \ref{compl}. It describes the homogenized behavior of a family of   GLEs (Eqns. \eqref{res_gle_1}-\eqref{res_gle_2}), parametrized by $\epsilon > 0$,  in the limit as $\epsilon \to 0$. This limit is equivalent to the limit in which the inertial time scale, some of the memory time scales and some of the noise correlation time scales in the pre-limit system, tend to zero at the same rate. As in Theorem \ref{newsmallm}, the result here states that the position process converges, in a strong pathwise sense, to a component of a higher dimensional process satisfying an It\^o SDE which contains non-trivial drift correction terms. Again, the limiting position process is non-Markov. However, the structure of the SDE is rather different from the one obtained in Theorem \ref{newsmallm}. As discussed later, this result holds interesting consequences for systems exhibiting anomalous diffusion. 
\item The third and forth main result are Corollary \ref{w1case} and Corollary \ref{w2case}.  These results specialize the earlier ones to one-dimensional GLE models, which are generalizations of (M1) and (M2), and follow from the earlier theorems. They give explicit expressions for the drift correction terms present in the limiting SDEs and therefore may be used directly for modeling and simulation purposes. Furthermore, we show that, in the important case where the fluctuation-dissipation relation (see Remark \ref{fdt_rem}) holds, the two corollaries are intimately connected. Recall that these results are going to be presented first in Section \ref{sect_appl}. 
\item The last main result is Theorem \ref{mainthm}, on homogenization of a family of parametrized SDEs whose coefficients are state-dependent. These SDEs are variants of the ones studied in earlier works \cite{hottovy2015smoluchowski,birrell2017small,2017BirrellLatest}. In comparison with all the earlier studies, the state-dependent coefficients of the pre-limit SDEs \eqref{sde1}-\eqref{sde2} may depend on the parameter $\epsilon > 0$ (to be taken to zero) explicitly. Therefore, this result is new and not simply a minor generalization of earlier results. Moreover, it is important in the context of present paper and is needed here to study various homogenization limits of GLEs, the importance of which is evident in the discussions above, in the main paper.
\end{itemize}

\section{Application to One-Dimensional GLE Models}
\label{sect_appl}

We first study the small mass limit of a one-dimensional GLE, which is a generalized version of the GLE in model (M2) of Example \ref{ex_mot},  modeling super-diffusion of a particle in a heat bath. Our models are generalized in that the coefficients of the GLEs are state-dependent. For simplicity, we are going to omit the explicit time dependence in the damping and noise coefficients---but not in the external force.


For $t \in \RR^+$, $m>0$, let $x_t, v_t \in \RR$ be the solutions to the equations:
\begin{align}
dx_t &= v_t dt, \label{one-dim-pos0} \\
m dv_t &= -g(x_t)\left(\int_0^t \kappa(t-s) h(x_s) v_s ds\right) dt + \sigma(x_t) \xi_t dt + F_e(t,x_t)dt, \label{one-dim-vel0}
\end{align}
where
\begin{equation} \label{w2}
\kappa(t) = \frac{\beta^2}{2} (\delta(t) - \Gamma_1 e^{-\Gamma_1 |t|}),
\end{equation}
where $\Gamma_1 > 0$, and $\xi_t$ is the mean-zero stationary Gaussian process with the covariance function $R(t)=\kappa(t)$ and spectral density,
\begin{equation}
\mathcal{S}(\omega) = \frac{\beta^2  \omega^2}{\omega^2+\Gamma_1^2},
\end{equation}
The initial data $(x,v)$ are random variables independent of $\epsilon$ and have finite moments of all orders.

The following corollary describes the limiting SDE for the particle's position obtained in the small mass limit of \eqref{one-dim-pos0}-\eqref{one-dim-vel0}.

\begin{cor} \label{w1case}
Assume that for every $y \in \RR$, $g(y), g'(y), h(y), h'(y)$, $\sigma(y)$ are bounded continuous functions in $y$, $F_e(t,y)$ is bounded and continuous in $t$ and $y$, and all the listed functions have bounded $y$-derivatives.  Then in the limit $m \to 0$, the particle's position, $x_t \in \RR$, satisfying \eqref{one-dim-pos0}-\eqref{one-dim-vel0},
converges to $X_t$, where $X_t$ solves the following It\^o SDE:
\begin{align}
dX_t &= \frac{2}{\beta^2 g h} F_e(t,X_t) dt - \frac{2}{\beta h} Y_t dt + S_1(X_t) dt + \frac{2 \sigma}{\beta g h} (Z_t dt + dW_t), \label{e1} \\
dY_t &= -\frac{\Gamma_1}{\beta g} F_e(t,X_t) dt + S_2(X_t) dt - \frac{\Gamma_1  \sigma}{g} (dW_t + Z_t dt), \label{e2} \\
dZ_t &= -\Gamma_1 Z_t dt - \Gamma_1 dW_t,  \label{e3}
\end{align}
where \begin{align}
S_1(X) &= \frac{2}{\beta^2} \frac{\partial}{\partial X}\left(\frac{1}{g h} \right) \frac{\sigma^2}{g h},  \ \ \ \ S_2(X) = -\frac{\Gamma_1}{\beta} \frac{\partial}{\partial X}\left(\frac{1}{g} \right) \frac{\sigma^2}{g h}.
\end{align}
Moreover, if in addition $g := \phi \sigma$, where $\phi > 0$, then the number of limiting SDEs reduces from three to two:
\begin{align}
dX_t &= \frac{2}{\beta^2 \phi^2 } \frac{\partial}{\partial X}\left(\frac{1}{\sigma h}\right) \frac{\sigma}{h} dt + \frac{2}{ \phi \sigma h \beta^2} F_e(t,X_t) dt - \frac{2}{\beta \phi h}U_t^{\phi} dt + \frac{2}{\beta \phi h} dW_t, \label{r1} \\
dU_t^\phi &= -\frac{\Gamma_1}{\beta \phi^2} \frac{\partial}{\partial X}\left(\frac{1}{\sigma}\right) \frac{\sigma}{h} dt - \frac{\Gamma_1}{ \beta \sigma} F_e(t,X_t) dt, \label{r2}
\end{align}
where $U_t^\phi = \phi Y_t-Z_t$.

The convergence is in the sense that for every $T>0$, $\sup_{t \in [0,T]} |x_t - X_t| \to 0$ in probability as $m \to 0$.
\end{cor}
\begin{proof}
We apply Theorem \ref{newsmallm} by setting $d=1, d_2 = d_4 = 2$, $\alpha_1 = \alpha_3 = 0$, $\alpha_2 = \alpha_4 = 1$,
$\vecc{\gamma}_0 = \beta^2  g h/2$, $\vecc{\sigma}_0 = \beta \sigma$, $\vecc{h} = h$, $\vecc{g} = g$, $\vecc{\sigma} = \sigma$, $\vecc{C}_2 = \vecc{C}_4 = \beta$, $\vecc{\Gamma}_2 = \Gamma_1$, $\vecc{M}_2 \vecc{C}_2^* = -\Gamma_1 \beta/2$, $\vecc{\Gamma}_4 = \Gamma_1$, $\vecc{\Sigma}_4 = -\Gamma_1$, and  $\vecc{F}_e = F_e$. The assumptions of Theorem \ref{newsmallm} can be verified in a straightforward way and so the results of the corollary follow.
\end{proof}

We next specialize the result of Theorem \ref{compl} to study homogenization of one-dimensional GLEs which are generalizations of the model (M1) in Example \ref{ex_mot}:
for $t \in \RR^+$, $m>0$, let $x_t, v_t \in \RR$ be the solutions to the equations:
\begin{align}
dx_t &= v_t dt, \label{one-dim-pos} \\
m dv_t &= -g(x_t)\left(\int_0^t \kappa(t-s) h(x_s) v_s ds\right) dt + \sigma(x_t) \xi_t dt + F_e(t,x_t)dt, \label{one-dim-vel}
\end{align}
where
\begin{equation} \label{w3}
\kappa(t) = \frac{\beta^2 \Gamma_2^2(\Gamma_2 e^{-\Gamma_2 |t|} - \Gamma_1 e^{-\Gamma_1 |t|})}{2(\Gamma_2^2-\Gamma_1^2)},
\end{equation}
with $\Gamma_2 > \Gamma_1 > 0$, and $\xi_t$ is the mean-zero stationary Gaussian process with the covariance function $R(t)=\kappa(t)$ and spectral density,
\begin{equation}
\mathcal{S}(\omega) = \frac{\beta^2 \Gamma_2^2 \omega^2}{(\omega^2+\Gamma_1^2)(\omega^2+\Gamma_2^2)}.
\end{equation}
The initial data $(x,v)$ are random variables independent of $\epsilon$ and have finite moments of all orders.

For $\epsilon > 0$, we set $m = m_0 \epsilon$ and $\Gamma_2 = \gamma_2/\epsilon$ in  \eqref{one-dim-pos}-\eqref{one-dim-vel}, where $m_0$ and $\gamma_2$ are positive constants. This gives the family of equations:
\begin{align}
dx^\epsilon_t &= v^\epsilon_t dt, \label{res_one-dim-pos} \\
m_0 \epsilon dv^\epsilon_t &= -g(x^\epsilon_t)\left(\int_0^t \kappa^\epsilon(t-s) h(x^\epsilon_s) v^\epsilon_s ds\right) dt + \sigma(x^\epsilon_t) \xi^\epsilon_t dt + F_e(t,x^\epsilon_t)dt, \label{res_one-dim-vel}
\end{align}
where
\begin{equation} \label{w3_rescaled}
\kappa^\epsilon(t) = \frac{\beta^2 \gamma_2^2(\frac{\gamma_2}{\epsilon} e^{-\frac{\gamma_2 }{\epsilon}|t|} - \Gamma_1 e^{-\Gamma_1 |t|})}{2(\gamma_2^2- \epsilon^2 \Gamma_1^2)},
\end{equation}
and $\xi^\epsilon_t$ is the family of mean-zero stationary Gaussian processes with the covariance functions, $R^\epsilon(t) = \kappa^\epsilon(t)$.\\

\noindent {\bf Discussion.}
We discuss the physical meaning behind the above rescaling of  parameters. Recall that in the first case of Example \ref{ex_mot} (i.e. the model (M1)), the mean-square displacement of the particle grows as $t^2$ as $t \to \infty$ and therefore the above model describes a particle exhibiting super-diffusion. As $\epsilon \to 0$, the environment allows for more and more negative correlation  and in the limit the covariance function consists of a delta-type peak at $t=0$ and a negative long tail compensating for the positive peak when integrated (see Figure \ref{fig1} and also page 105 of \cite{toda2012statistical}). Indeed,
\begin{equation}
\kappa^\epsilon(t) \to \kappa(t) := \frac{\beta^2}{2} (\delta(t)-\Gamma_1 e^{-\Gamma_1|t|})
\end{equation}
as $\epsilon \to 0$.
This is the so-called {\it vanishing effective friction case} in \cite{bao2005non}. The noise with the covariance function $\kappa^\epsilon(t)$ is called harmonic velocity noise, whereas the noise with the covariance function $\kappa(t)$ is the derivative of an Ornstein-Uhlenbeck process.

The following corollary provides the homogenized model in the limit $\epsilon \to 0$ of \eqref{res_one-dim-pos}-\eqref{res_one-dim-vel}.

\begin{cor} \label{w2case}
Assume that for every $y \in \RR$, $g(y), g'(y), h(y), h'(y)$, $\sigma(y)$ are bounded continuous functions in $y$, $F_e(t,y)$ is bounded and continuous in $t$ and $y$, and all the listed functions have bounded derivatives in $y$.  Then in the limit $\epsilon \to 0$, the particle's position, $x^\epsilon_t \in \RR$, satisfying \eqref{res_one-dim-pos}-\eqref{res_one-dim-vel},
converges to $X_t$, where $X_t$ solves the following It\^o SDE:
\begin{align}
dX_t &= \frac{2}{\beta^2 g h} F_e(t,X_t) dt -  \frac{2}{\beta h}Y_t dt + S_1(X_t) dt + \frac{2\sigma}{\beta g h } (dW_t + Z_t dt), \\
dY_t &= -\frac{\Gamma_1}{\beta g} F_e(t,X_t) dt + S_2(X_t) dt - \frac{\Gamma_1 \sigma}{g}(dW_t + Z_t dt), \\
dZ_t &= -\Gamma_1 Z_t dt - \Gamma_1 dW_t,
\end{align}
where  $g=g(X_t)$, $h = h(X_t)$, $\sigma=\sigma(X_t)$, $W_t$ is a one-dimensional Wiener process, and
\begin{align}
S_1 &= \frac{2}{\beta^2} \frac{\partial}{\partial X}\left(\frac{1}{gh}\right)\frac{\sigma^2}{gh} - \frac{\partial}{\partial X}\left(\frac{1}{h}\right) \frac{4 \sigma^2}{g(gh \beta^2+4m_0\gamma_2)} \nonumber \\  
&\ \ \ \ + \frac{\partial}{\partial X}\left(\frac{\sigma}{gh}\right)\frac{4 \sigma}{ \beta^2  g h  +4m_0\gamma_2}, \\
S_2 &= -\frac{\Gamma_1}{\beta}\frac{\partial}{\partial X}\left(\frac{1}{g}\right)\frac{\sigma^2}{g h} - \frac{\partial}{\partial X}\left(\frac{\sigma}{g}\right)\frac{2 \Gamma_1 \beta \sigma}{\beta^2 g h +4m_0\gamma_2}.
\end{align}

Moreover, if in addition $g := \phi \sigma$, where $\phi > 0$, then the number of limiting SDEs reduces from three to two:
\begin{align}
dX_t &= \frac{2}{\beta^2 \phi^2 } \frac{\partial}{\partial X}\left(\frac{1}{\sigma h}\right) \frac{\sigma}{h} dt + \frac{2}{ \phi \sigma h \beta^2} F_e(t,X_t) dt - \frac{2}{\beta \phi h}U_t^{\phi} dt + \frac{2}{\beta \phi h} dW_t, \label{r1} \\
dU_t^\phi &= -\frac{\Gamma_1}{\beta \phi^2} \frac{\partial}{\partial X}\left(\frac{1}{\sigma}\right) \frac{\sigma}{h} dt - \frac{\Gamma_1}{ \beta \sigma} F_e(t,X_t) dt, \label{r2}
\end{align}
where $U_t^\phi = \phi Y_t-Z_t$.

The convergence is in the sense that for every $T>0$, $\sup_{t \in [0,T]} |x^\epsilon_t - X_t| \to 0$ in probability as $\epsilon \to 0$.
\end{cor}
\begin{proof}
Let $d=1$, $d_2 = d_4 = 2$ and denote the one-dimensional version of the variables,  coefficients and parameters in Theorem \ref{compl} by non-bold letters (for instance, $x_t$, $B_2$, $\Gamma_{2,2}$ etc.). Furthermore, set $B_2 = B_4 = \beta > 0$, $\gamma_{2,2}=\gamma_{4,2}=\gamma_2 > 0$ and $\Gamma_{2,1}=\Gamma_{4,1}=\Gamma_1$. Then it can be verified that the assumptions of Theorem \ref{compl} hold and the results follow upon solving a Lyapunov equation.
\end{proof}


\begin{rem}  A few remarks on the contents of Corollary \ref{w2case} follow.
\begin{itemize}
\item[(i)] the homogenized position process is non-Markov, driven by a colored noise process which is the derivative of the Ornstein-Uhlenbeck process. This behavior is expected in view of the asymptotic behavior of the rescaled memory function and spectral density as $\epsilon \to 0$.
\item[(ii)] similarly to the small mass limit case considered earlier, the limiting equation for the particle's position not only contains noise-induced drift terms but is also coupled to equations for other slow variables. Moreover, the limiting equations for these other slow variables also contain non-trivial correction terms -- the {\it memory induced drift}.
\end{itemize}
\end{rem}

\noindent {\bf Relation between Corollary \ref{w1case} and Corollary \ref{w2case}.}
The limiting SDE systems in Corollaries \ref{w1case} \& \ref{w2case} are generally different because of the different correction drift terms $S_1$ and $S_2$. In other words, sending $\Gamma_2 \to \infty$ first in \eqref{one-dim-pos}-\eqref{one-dim-vel} and then taking $m \to 0$ of the resulting GLE does not, in general, give the same limiting SDE as taking the joint limit of $m \to 0$ and $\Gamma_2 \to \infty$.  However, if one further assumes that $g$ is proportional to $\sigma$, then the limiting SDE systems coincide.
An important particular case is when $g = h = \sigma$, in which case a fluctuation-dissipation relation holds and the GLE can be derived from a microscopic Hamiltonian model (see Remark \ref{fdt_rem}).  In this case, the homogenized model described in both corollaries reduces to:
\begin{align}
dX_t &= \frac{2}{\beta^2 \sigma^2} F_e(t,X_t) dt - \frac{2}{\beta \sigma} U_t dt + \frac{2}{\beta^2}\frac{\partial}{\partial X_t}\left(\frac{1}{\sigma^2} \right) dt + \frac{2}{\beta \sigma} dW_t, \label{fdt1} \\
dU_t &= -\frac{\Gamma_1}{\beta \sigma} F_e(t,X_t) dt - \frac{\Gamma_1}{\beta} \frac{\partial}{\partial X}\left(\frac{1}{\sigma} \right) dt.  \label{fdt2}
\end{align}

To end this section, we remark that one could in principle repeat the above analysis for the case where the spectral density varies as $\omega^{2l}$, for $l=2,4,\dots$ (i.e. the highly nonlinear case) as well as extending the studies done so far in various other directions. To illustrate how non-trivial the calculations and results could become, we work out another example in Appendix \ref{anothereg}.

\begin{figure}
  \caption{Plot of the memory function $\kappa(t)$ in \eqref{w3} with $\Gamma_1 = 1$, $\beta = 1$ for different values of $\Gamma_2$ (left) and the memory function in \eqref{w5}  with $\Gamma_1 = 1$, $\Gamma_2 = 2$, $\beta = 1$ for different values of $\Gamma_3$ (right)} \label{fig1}
  \centering
  \includegraphics[width=0.48\textwidth]{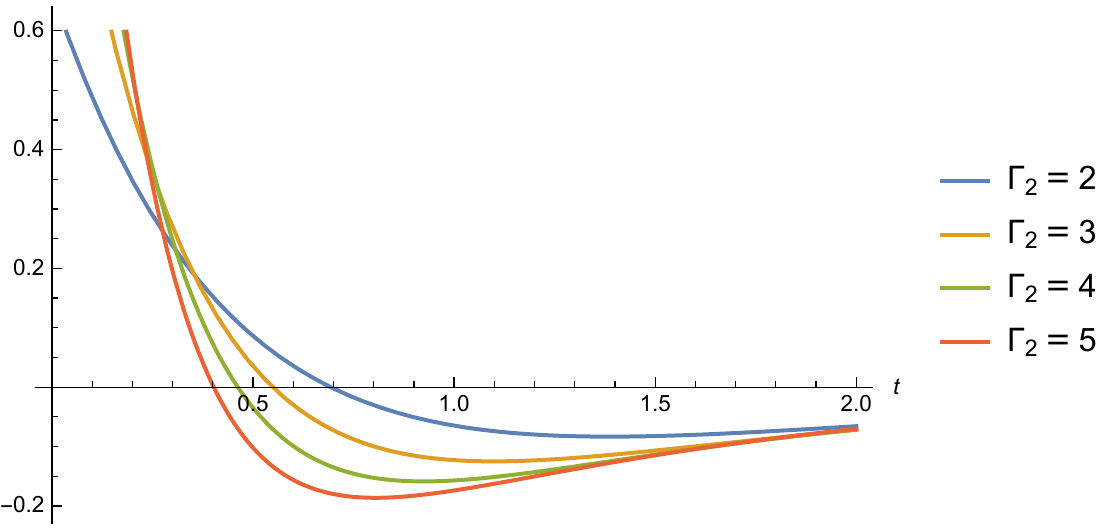}
 \includegraphics[width=0.48\textwidth]{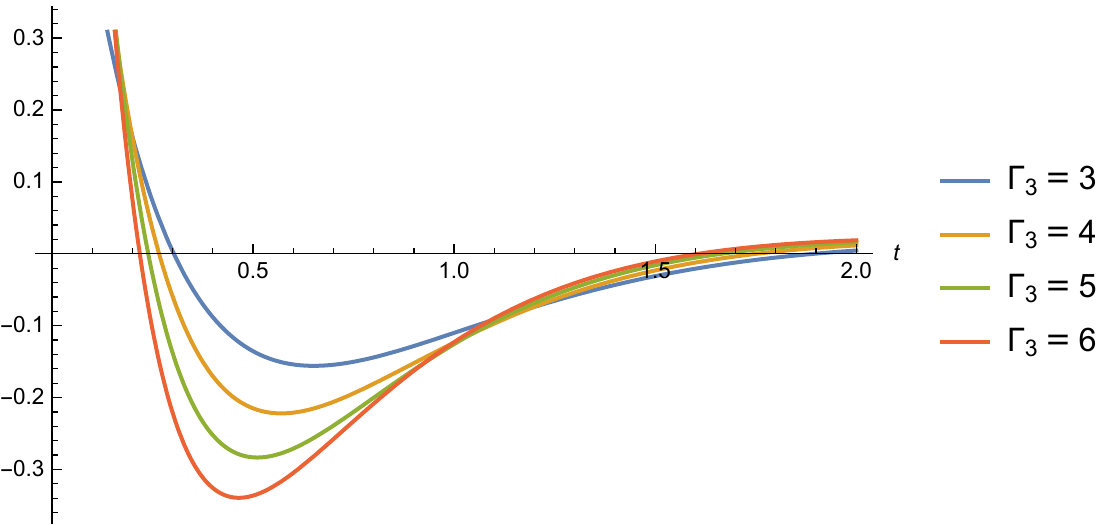}
\end{figure}

\section{GLEs in Finite Dimensions}
\label{sect_gles}
We call a system modeled by GLE of the form \eqref{gle2}  a {\it generalized Langevin system}.  Its dynamics will be referred to as  {\it generalized Langevin dynamics}.

We assume that the memory function $\vecc{\kappa}(t)$ in the GLE \eqref{gle2} is a  {\it Bohl function}, i.e. that each matrix element of $\vecc{\kappa}(t)$ is a finite, real-valued linear combination of exponentials, possibly multiplied by polynomials and/or by trigonometric functions. The noise process, $\{\vecc{\xi}(t), t \in \RR^+ \}$, is a mean-zero, mean-square continuous stationary Gaussian process with Bohl covariance function and, therefore, its spectral density $\vecc{\mathcal{S}}(\omega)$ is a rational function---(see Theorem 2.20 in \cite{trentelman2002control}). In this case, the generalized Langevin dynamics can be realized by an SDE system in a finite-dimensional space  (see next subsection for details).  The case in which an infinite-dimensional space is required is deferrred to a future work (see also Remark \ref{rem_inf_dim} and Section \ref{sect_conclusions}).


We recall a useful fact: given a rational spectral  density $\vecc{\mathcal{S}}(\omega) \in \RR^{r \times r}$, there exists a  rational function
$\vecc{G}(z) \in \CC^{r \times l}$, called a {\it spectral factor},  such that
$\vecc{\mathcal{S}}(\omega) = \vecc{G}(i\omega)\vecc{G}^{*}(-i\omega)$. We emphasize that such factorization is not unique \cite{lindquist2015linear}.

\subsection{Generalized Langevin Systems}
Below we define the memory function and the noise process in the GLE \eqref{gle2} (see Eqn. \eqref{nonM_force}),  and along the way introduce our notation. They are defined in a manner ensuring simplicity as well as providing sufficient parameters for matching the memory function and the correlation function of the noise, thereby preserving the essential statistical properties of the GLE. This provides a systematic framework for our homogenization studies (see the discussion in Section \ref{sect_homogofGLE}). 

For $i=1,2,3,4$, let $\vecc{\Gamma}_i \in \RR^{d_i \times d_i}$, $\vecc{M}_i \in \RR^{d_i \times d_i}$, $\vecc{\Sigma}_i \in \RR^{d_i \times q_i}$  be constant matrices. Also, let $\vecc{C}_i \in \RR^{q \times d_i}$ (for $i=1,2$) and  $\vecc{C}_i \in \RR^{r \times d_i}$ (for $i=3,4$) be constant matrices. Here, the $d_i$ and $q_i$ ($i=1,2,3,4$) are positive integers. Let $\alpha_i \in \{0,1\}$ be a ``switch on or off'' parameter. We define the memory function in terms of the sextuple $(\vecc{\Gamma}_1,\vecc{M}_1,\vecc{C}_1;\vecc{\Gamma}_2,\vecc{M}_2,\vecc{C}_2)$ of matrices:
\begin{equation} \label{memory_realized}
\vecc{\kappa}(t)= \alpha_1 \vecc{\kappa}_1(t) + \alpha_2\vecc{\kappa}_2(t) =  \sum_{i=1}^2 \alpha_i \vecc{C}_i e^{-\vecc{\Gamma_i}|t|}\vecc{M}_i\vecc{C}_i^*,
\end{equation}
The noise process is defined as:
\begin{equation} \label{noise}
\vecc{\xi}_t = \alpha_3 \vecc{C}_3 \vecc{\beta}^3_t + \alpha_4 \vecc{C}_4 \vecc{\beta}^4_t,\end{equation}
where the $\vecc{\beta}^{j}_t \in \RR^{d_j}$ ($j=3,4$) are independent Ornstein-Uhlenbeck type processes, i.e. solutions of the SDEs:
\begin{equation} \label{realize}
d\vecc{\beta}^j_t = -\vecc{\Gamma}_j \vecc{\beta}^j_t dt + \vecc{\Sigma}_j d\vecc{W}^{(q_j)}_t,
\end{equation}
with the initial conditions, $\vecc{\beta}^j_0$,  normally distributed with mean-zero and covariance $\vecc{M}_j$. Here, $\vecc{W}^{(q_j)}_t$ denotes a $q_j$-dimensional Wiener process, independent of $\vecc{\beta}^j_0$. Also, the Wiener processes $\vecc{W}_t^{(q_3)}$ and $\vecc{W}_t^{(q_4)}$ are independent.

For $i=1,2,3,4$, $\vecc{\Gamma}_i$ is {\it positive stable}, i.e. all eigenvalues of $\vecc{\Gamma}_i$ have positive real parts and $\vecc{M}_i = \vecc{M}_i^* > 0$ satisfies the following Lyapunov equation:
\begin{equation}
\vecc{\Gamma}_i \vecc{M}_i+\vecc{M}_i \vecc{\Gamma}_i^*=\vecc{\Sigma}_i \vecc{\Sigma}_i^*.
\end{equation}
The $\vecc{M}_i$ are therefore the steady-state covariances of the systems, i.e. the resulting Ornstein-Uhlenbeck processes are stationary.   In control theory, $\vecc{M}_i$ is also known as the {\it controllability Gramian} for the pair $(\vecc{\Gamma}_i, \vecc{\Sigma}_i)$ \cite{trentelman2002control}.

The covariance matrix, $\vecc{R}(t)$, of the mean-zero Gaussian noise process is expressed by the sextuple $(\vecc{\Gamma}_3,\vecc{M}_3,\vecc{C}_3;
\vecc{\Gamma}_4,\vecc{M}_4,\vecc{C}_4)$ of matrices as follows:
\begin{equation} \label{cov}
\vecc{R}(t)=\alpha_3 \vecc{R}_3(t)+ \alpha_4 \vecc{R}_4(t) =  \sum_{i=3}^4 \alpha_i \vecc{C}_i e^{-\vecc{\Gamma_i}|t|}\vecc{M}_i\vecc{C}_i^*,
\end{equation}
and so the sextuple $(\vecc{\Gamma}_3,\vecc{M}_3,\vecc{C}_3;\vecc{\Gamma}_4,\vecc{M}_4,\vecc{C}_4)$, together with the parameters $\alpha_3, \alpha_4$, completely determine the probability distributions of $\vecc{\xi}_t$. We denote the spectral density of the noise process by $\vecc{\mathcal{S}}(\omega) = \sum_{i=3,4}\alpha_i \vecc{\mathcal{S}}_i(\omega)$, where $\vecc{\mathcal{S}}_i(\omega)$ is the Fourier transform of $\vecc{R}_i(t)$ for $i=3,4$.

We will view the system \eqref{noise}-\eqref{realize} (which is in a statistical steady state) as a representation of the noise process $\vecc{\xi}_t$ and call such a representation a (finite-dimensional) {\it stochastic realization} of $\vecc{\xi}_t$. Similarly, we view \eqref{memory_realized} as a representation of the memory function $\vecc{\kappa}(t)$ and call such a representation a (finite-dimensional, deterministic) {\it memory realization} of $\vecc{\kappa}(t)$. We call the Fourier transform of $\vecc{\kappa}(t)$ and $\vecc{R}(t)$ the {\it spectral density of the memory function} and  {\it spectral density of the noise process} respectively.

An important message from the stochastic realization theory is that the system \eqref{noise}-\eqref{realize} is more than a representation of $\vecc{\xi}_t$ in terms of a white noise, in that it also contains  state variables $\vecc{\beta}^j$ ($j=3,4$) which serve as a ``dynamical memory". In contrast to standard treatments, this dynamical memory comes not from one, but from two independent systems of type \eqref{realize}. This will be used to include two distinct types of dynamical memory that can be switched on or off using the parameters $\alpha_i$ -- see Proposition \ref{asympbeh}. This consideration motivates us to define the memory function (and noise) explicitly using two independent systems, with different constraints on their parameters easier to state than if a single higher-dimensional system were used. 

The sextuples that define the memory function in \eqref{memory_realized} and the noise process in \eqref{noise} are only unique up to the following transformations:  \begin{equation} \label{transf_realize}
(\vecc{\Gamma}'_i=\vecc{T}_i \vecc{\Gamma}_i \vecc{T}^{-1}_i, \vecc{M}_i' = \vecc{T}_i \vecc{M}_i \vecc{T}_i^{*}, \vecc{C}'_i =  \vecc{C}_i \vecc{T}_i^{-1}),
\end{equation}
where $i=1,2,3,4$ and $\vecc{T}_i$ are any invertible matrices of appropriate dimensions \cite{lindquist2015linear}. Different choices of $\vecc{T}_i$ correspond to  different coordinate systems.



\begin{rem} Realization of the memory function and noise process in terms of the matrix sextuples, as defined above, covers all GLEs driven by Gaussian processes that can be realized in a finite dimension (see the propositions and theorems on page 303-308 of \cite{willems1980stochastic}). See also the remarks on the subject in \cite{LimWehr_Homog_NonMarkovian}.
\end{rem}

A summary of the above discussion is included in the following:

\begin{ass} \label{ass_bohl}
The memory function $\vecc{\kappa}(t)$ in the GLE \eqref{gle2} is a real-valued Bohl function defined by \eqref{memory_realized} and the noise process, $\{\vecc{\xi}_t, t \in \RR^+ \}$, is a mean-zero, mean-square continuous, stationary Gaussian process with Bohl covariance function (hence, with rational spectral density), admitting a stochastic realization given by \eqref{noise}-\eqref{realize}. Furthermore, we assume that any spectral factors $\vecc{\Phi}_i(z)$ ($i=1,2,3,4$) of the spectral densities $\vecc{\mathcal{S}}_i(\omega)$ are {\it minimal} (see Chapter 10 in \cite{lindquist2015linear}). 
\end{ass}

We introduce a generalized version of the effective damping constant and effective diffusion constant used in \cite{LimWehr_Homog_NonMarkovian}, which will be useful to study the asymptotic behavior of spectral densities.

\begin{defn} \label{defn_effconstnats}
For $n \in \ZZ$, the {\it $n$th order effective damping constant} is defined as the constant matrix, parametrized by $\alpha_1, \alpha_2 \in \{0,1\}$:
\begin{equation} \label{eff_damping}
\vecc{K}^{(n)}(\alpha_1,\alpha_2) := \alpha_1 \vecc{K}_1^{(n)} + \alpha_2 \vecc{K}_2^{(n)} \in \RR^{q \times q},
\end{equation} where
 $\vecc{K}_i^{(n)} = \vecc{C}_i \vecc{\Gamma}_i^{-n} \vecc{M}_i \vecc{C}_i^*$ (for $i=1,2$).
Likewise, the {\it $n$th order effective diffusion constant},
\begin{equation} \label{eff_diff}
\vecc{L}^{(n)}(\alpha_3,\alpha_4) :=  \alpha_3 \vecc{L}_3^{(n)} + \alpha_4 \vecc{L}_4^{(n)}  \in \RR^{r \times r},
\end{equation}
where $\vecc{L}_j^{(n)} = \vecc{C}_j \vecc{\Gamma}_j^{-n} \vecc{M}_j \vecc{C}_j^*$ (for $j=3,4$).
\end{defn}
Note that the first order effective damping constant $\vecc{K}^{(1)}(\alpha_1,\alpha_2) = \int_0^{\infty} \vecc{\kappa}(t) dt$ and the first order effective diffusion constant $\vecc{L}^{(1)}(\alpha_3,\alpha_4) = \int_0^{\infty} \vecc{R}(t) dt$ are simply the effective damping constant and effective diffusion constant introduced in \cite{LimWehr_Homog_NonMarkovian}. The  memory function and the covariance function of the noise process can be expressed in terms of these constants:
\begin{equation}
\vecc{\kappa}(t) = \sum_{i=1,2} \sum_{n=0}^{\infty} \alpha_i  \frac{(-|t|)^n}{n!} \vecc{K}^{(-n)}_i, \ \ \ \vecc{R}(t) = \sum_{j=3,4} \sum_{n=0}^{\infty} \alpha_j  \frac{(-|t|)^n}{n!} \vecc{L}^{(-n)}_j.
\end{equation}\\

\begin{ass} \label{ass_vanishingornot}
The matrix $\vecc{K}_1^{(1)}$ in the expression for first order effective damping constant is invertible and the matrix $\vecc{K}_2^{(1)}$ equals zero.  Similarly, in the expression for the first order effective diffusion constant
$\vecc{L}_3^{(1)}$, which is invertible, $\vecc{L}_4^{(1)} = \vecc{0}$.
\end{ass}



In order to develop intuition about general GLEs, it will be helpful to study the following exactly solvable special case.

\begin{ex} \label{ass_exactsolve} (An exactly solvable case)
In the GLE \eqref{gle2}, set
$\vecc{F}_e = \vecc{0}$.  Let $\vecc{\gamma}_0(t,\vecc{x}) = \vecc{\gamma}_0$, $\vecc{\sigma}_0(t,\vecc{x}) = \vecc{\sigma}_0$, $\vecc{h}(t,\vecc{x}) = \vecc{h}$, $\vecc{g}(t,\vecc{x}) = \vecc{g}$ and $\vecc{\sigma}(t,\vecc{x}) = \vecc{\sigma}$ be constant matrices. The initial data are the random variables, $\vecc{x}(0) = \vecc{x}$, $\vecc{v}(0) = \vecc{v}$, independent of $\{\vecc{\xi}(t), t \in \RR^+ \}$ and of $\{\vecc{W}^{(k)}(t), t \in \RR^+\}$.  The resulting GLE is:
\begin{equation} \label{gle_es}
m d\vecc{v}(t) = -\vecc{\gamma}_0 \vecc{v}(t) dt  -\vecc{g} \left( \int_0^t \vecc{\kappa}(t-s) \vecc{h} \vecc{v}(s) ds \right) dt + \vecc{\sigma}_0 d\vecc{W}^{(k)}(t) +  \vecc{\sigma} \vecc{\xi}(t) dt. 
\end{equation}
Of particular interest is the GLE \eqref{gle_es} with $\vecc{\gamma}_0 = \vecc{\sigma}_0 \vecc{\sigma}_0^*/2 \geq 0$, $\vecc{g} = \vecc{h}^* = \vecc{\sigma} > 0$, and $\vecc{R}(t) = \vecc{\kappa}(t) = \vecc{\kappa}^*(t)$, so that the fluctuation-dissipation relations hold (see Remark \ref{fdt_rem} and also Remark  \ref{msd_general}). The resulting GLE gives a simple model describing the motion of a free particle, interacting  with a  heat bath.  Note that generally the process $\vecc{v}(t)$ is not assumed to be stationary, in particular $\vecc{v}(0)$ could be an arbitrarily distributed random variable.
\end{ex}

The following proposition gives the asymptotic behavior of the spectral densities (equivalently, covariance functions or memory functions), the regularity\footnote{Sample path continuity does not in general imply mean-square continuity.} (in the mean-square sense) of the noise process, and, in the exactly solvable case of Example \ref{ass_exactsolve}, the long-time mean-squared displacement   of the particle.

\begin{prop} \label{asympbeh}
Suppose that the Assumptions \ref{ass_bohl} and \ref{ass_vanishingornot} are satisfied.  Let $\vecc{x}(t) =  \int_0^t \vecc{v}(s) ds \in \RR^d$, where $\vecc{v}(t)$ solves the GLE \eqref{gle_es}.
\begin{itemize}
\item[(i)] We have $\vecc{\mathcal{S}}_3(\omega) = O(1)$ as $\omega \to 0$.    Also, let $k \geq 3$ be a positive odd integer and assume that $\vecc{L}_4^{(n)} = 0$ for $0 < n < k$, where $n$ is odd, and $\vecc{L}_4^{(k)} \neq 0$.  Then $\vecc{\mathcal{S}}_{4}(\omega) = O(\omega^{k-1})$ as $\omega \to 0$.
If there exists $h > 0$ such that  the noise spectral density, $\vecc{\mathcal{S}}(\omega) = O\left(\frac{1}{\omega^{2h+1}}\right)$ as $\omega \to \infty$, then $\vecc{\xi}_t$ is $n$-times mean-square differentiable\footnote{A process $X(t)$ is mean-square differentiable on a time interval $\mathcal{\tau}$ if for every $t \in \mathcal{\tau}$, $$\left\| \frac{X(t+h)-X(t)}{h}-\frac{dX}{dt}\right\| \to 0,$$ as $h \to 0$.  } for $n < h$.
\item[(ii)] Let $\hat{\vecc{\kappa}}(z)$ denote the Laplace transform of $\vecc{\kappa}(t)$, i.e. $\hat{\vecc{\kappa}}(z) :=\int_0^\infty \vecc{\kappa}(t) e^{-zt} dt$, and $\mathcal{E} = \frac{1}{2} m E[\vecc{v}\vecc{v}^*]$ be the particle's initial average kinetic energy. Assume for simplicity that $\vecc{R}(t) = \vecc{\kappa}(t)$  and  $\vecc{\sigma}\vecc{\kappa}(t) \vecc{\sigma}^* = \vecc{h}^* \vecc{\kappa}^*(t) \vecc{g}^*$.  Then we have the following formula for the particle's mean-squared displacement (MSD):
\begin{align} \label{msd_formula}
E[\vecc{x}(t)\vecc{x}^*(t)] &=  2 \int_0^t \vecc{H}(s) ds + 2m \left(\vecc{H}(t) \mathcal{E} \vecc{H}^*(t) - \int_0^t \vecc{H}(u) \dot{\vecc{H}^*}(u) du \right) \nonumber \\ 
& \ \ \ \ + \int_0^t \vecc{H}(u) (\vecc{\sigma}_0 \vecc{\sigma}_0^* - 2 \vecc{\gamma}^*_0) \vecc{H}^*(u) du, 
\end{align}
where the Laplace transform of $\vecc{H}(t)$ is given by $\hat{\vecc{H}}(z) = z \hat{\vecc{F}}(z)$, with
\begin{equation} 
\hat{\vecc{F}}(z) = (z^2(mz\vecc{I}+\vecc{\gamma}_0+ \vecc{g}\hat{\vecc{\kappa}}(z)\vecc{h}))^{-1}.
\end{equation}\\
\end{itemize}

\noindent For (iii) and (iv) below, we consider the process $\vecc{x}_t$ solving the GLE \eqref{gle_es} with $\vecc{\gamma}_0 = \vecc{\sigma}_0 \vecc{\sigma}_0^*/2 \geq 0$, $\vecc{g} = \vecc{h}^* = \vecc{\sigma} > 0$, and $\vecc{R}(t) = \vecc{\kappa}(t) = \vecc{\kappa}^*(t)$.
\begin{itemize}
\item[(iii)] Let $\alpha_1 = \alpha_3 = 1$ ($\alpha_i$, for $i=2,4$, can be 0 or 1 and $\vecc{F}_0$ can be zero or nonzero). Then $E[\vecc{x}(t)\vecc{x}^*(t)] = O(t)$ as $t \to \infty$, in which case we say that the particle diffuses normally.
\item[(iv)] Let $\alpha_1 = 0$, $\alpha_2 =  1$ and $\vecc{F}_0 = \vecc{0}$ (the vanishing effective damping constant case). Then $E[\vecc{x}(t)\vecc{x}^*(t)] = O(t^{2})$ as $t \to \infty$, in which case we say that the particle  exhibits a ballistic (super-diffusive) behavior.
\end{itemize}
\end{prop}
\begin{proof}
\begin{itemize}
\item[(i)]   For $i=3,4$, it is easy to compute that
\begin{align}
\vecc{\mathcal{S}}_i(\omega) &=  \vecc{C}_i[(i\omega\vecc{I}+\vecc{\Gamma}_i)^{-1} + (-i\omega \vecc{I}+\vecc{\Gamma}_i)^{-1}]\vecc{M}_i \vecc{C}_i^* \\
&= 2\vecc{C}_i [(i\omega\vecc{I}+\vecc{\Gamma}_i)^{-1} \vecc{\Gamma}_i (-i\omega \vecc{I}+\vecc{\Gamma}_i)^{-1}]\vecc{M}_i \vecc{C}_i^*  \\
&= 2\vecc{C}_i\vecc{\Gamma}_i^{-1}(\omega^2 \vecc{\Gamma}_i^{-2} + \vecc{I})^{-1} \vecc{M}_i \vecc{C}_i^*,
\end{align}
and so one has:
\begin{equation}
\vecc{\mathcal{S}}_i(\omega) = 2\vecc{C}_i\vecc{\Gamma}_i^{-1}\vecc{M}_i \vecc{C}_i^* - 2\vecc{C}_i\vecc{\Gamma}_i^{-3}\vecc{M}_i \vecc{C}_i^* \omega^2 + 2\vecc{C}_i\vecc{\Gamma}_i^{-5}\vecc{M}_i \vecc{C}_i^* \omega^4 + \dots,
\end{equation}
as $\omega \to 0$. The first two statements in (i) then follow by Assumption \ref{ass_vanishingornot}. The last statement follows from Lemma 6.11 in \cite{lord2014introduction}.\\
\item[(ii)] Note that $\dot{\vecc{x}}(t) = \vecc{v}(t)$, with $\vecc{x}(0) = \vecc{0}$ and $\vecc{v}(t)$ solving the GLE \eqref{gle_es}, rewritten as:
\begin{equation} \label{ggg}
m \dot{\vecc{v}}(t)=-\vecc{\gamma}_0 \vecc{v}(t)+\vecc{\sigma}_0 \vecc{\eta}(t) -\vecc{g}\int_0^t \vecc{\kappa}(t-s) \vecc{h}\vecc{v}(s) ds + \vecc{\sigma} \vecc{\xi}(t),
\end{equation}
where $\vecc{\eta}(t) dt =  d \vecc{W}^{(k)}(t)$, and $\vecc{v}_0 = \vecc{v}$ is a random variable that is independent of $\{\vecc{\xi}(t), t \in \RR^+\}$ and of $\{\vecc{\eta}(t), t \in \RR^+\}$. These equations can be solved analytically by means of Laplace transform. Applying Laplace transform on the equations for $\vecc{x}_t$ and $\vecc{v}_t$ gives:
\begin{align}
z \hat{\vecc{x}}(z) &= \hat{\vecc{v}}(z), \\
m (z \hat{\vecc{v}}(z) - \vecc{v}(0)) &= -\vecc{g} \hat{\vecc{\kappa}}(z) \vecc{h} \hat{\vecc{v}}(z) - \vecc{\gamma}_0 \hat{\vecc{v}}(z) + \vecc{\sigma}_0 \hat{\vecc{\eta}}(z) + \vecc{\sigma} \hat{\vecc{\xi}}(z),
\end{align}
and thus 
\begin{equation}
\hat{\vecc{x}}(z) = \hat{\vecc{H}}(z) (m \vecc{v}(0) + \vecc{\sigma}_0 \hat{\vecc{\eta}}(z) + \vecc{\sigma} \hat{\vecc{\xi}}(z)),
\end{equation}
where $\hat{\vecc{H}}(z) = (mz^2\vecc{I}+z\vecc{\gamma}_0+ z\vecc{g}\hat{\vecc{\kappa}}(z)\vecc{h})^{-1}$.  Taking the inverse transform gives the following formula for $\vecc{x}(t)$:
\begin{equation}
\vecc{x}(t) = \vecc{H}(t) m \vecc{v} + \int_0^t \vecc{H}(t-s) (\vecc{\sigma}_0 \vecc{\eta}(s) + \vecc{\sigma} \vecc{\xi}(s)) ds, 
\end{equation}
where $\vecc{H}(0) = \vecc{0}$. 

Therefore, using the mutual independence of $\vecc{v}$, $\{\vecc{\xi}(t), t \in \RR^+\}$ and  $\{\vecc{\eta}(t), t \in \RR^+\}$,  the It\^o isometry, and the assumption that $\vecc{R}(t) = \vecc{\kappa}(t)$, we obtain:
\begin{align}
E[\vecc{x}(t) \vecc{x}^T(t)] &= 2m \vecc{H}(t) \mathcal{E} \vecc{H}^*(t) +   \int_0^t \vecc{H}(t-s) \vecc{\sigma}_0 \vecc{\sigma}_0^* \vecc{H}^*(t-s) ds   + \vecc{L}(t), \label{mssd}
\end{align}
where 
\begin{align}
\vecc{L}(t) &= \int_0^t ds \int_0^t du \ \vecc{H}(t-s) \vecc{\sigma} \vecc{\kappa}(|s-u|) \vecc{\sigma}^* \vecc{H}^*(t-u).
\end{align}
To compute the double integral $\vecc{L}(t)$, we first rewrite it as $\vecc{L}(t) = \vecc{L}_1(t) + \vecc{L}_2(t)$, with
\begin{align}
\vecc{L}_1(t) &= \int_0^t ds \ \vecc{H}(t-s) \int_s^t du \ \vecc{\sigma} \vecc{\kappa}(u-s) \vecc{\sigma}^* \vecc{H}^*(t-u), \\
\vecc{L}_2(t) &= \int_0^t ds \ \vecc{H}(t-s) \int_0^s du \ \vecc{\sigma} \vecc{\kappa}(s-u) \vecc{\sigma}^* \vecc{H}^*(t-u).
\end{align}
We then compute:
\begin{align}
\vecc{L}_1(t) &= \int_0^t ds \ \vecc{H}(t-s) \int_s^t d(t-u) \ \vecc{\sigma} \vecc{\kappa}(t-s-(t-u)) \cdot (-1) \vecc{\sigma}^* \vecc{H}^*(t-u), \\
&= \int_0^t ds \ \vecc{H}(t-s) \int_0^{t-s} d\tau \  \vecc{\sigma} \vecc{\kappa}(t-s-\tau) \vecc{\sigma}^* \vecc{H}^*(\tau), \\
&= \int_0^t ds \ \vecc{H}(t-s) (\vecc{\sigma} \vecc{\kappa} \vecc{\sigma}^*  \star \vecc{H}^*)(t-s), \\
&= \int_0^t du \ \vecc{H}(u) (\vecc{\sigma} \vecc{\kappa} \vecc{\sigma}^*  \star \vecc{H}^*)(u),
\end{align}
where $\star$ denotes convolution. Now note that, by the convolution theorem, $(\vecc{\sigma} \vecc{\kappa} \vecc{\sigma}^*  \star \vecc{H}^*)(u)$ is the inverse Laplace transform of $\vecc{\sigma}\hat{\vecc{\kappa}}(z) \vecc{\sigma}^* \hat{\vecc{H}^*}(z)$, which can be written as $\vecc{I}/z-(mz\vecc{I} + \vecc{\gamma}_0^*) \hat{\vecc{H}^*}(z)$ by using the assumption that  $\vecc{\sigma} \vecc{\kappa}(t) \vecc{\sigma}^* = \vecc{h}^* \vecc{\kappa}^*(t) \vecc{g}^*$. Computing the inverse transform gives us:
\begin{equation}
\vecc{L}_1(t) = \int_0^t du \ \vecc{H}(u) (\vecc{I} - m \dot{\vecc{H}^*}(u) - \vecc{\gamma}_0^* \vecc{H}^*(u)). \label{L1t}
\end{equation}
Similarly, we obtain  $\vecc{L}_2(t) = \vecc{L}_1(t)$, and so $\vecc{L}(t) = 2 \vecc{L}_1(t)$. Therefore, combining \eqref{mssd} and \eqref{L1t} gives us the desired formula for MSD.\\
\item[(iii)] $\&$ (iv) 
The assumptions that $\vecc{g} = \vecc{h}^* = \vecc{\sigma}$ and $\vecc{R}(t) = \vecc{\kappa}(t) = \vecc{\kappa}^*(t)$ ensure that we can apply the MSD formula in (ii). The additional assumption that $\vecc{\gamma}_0 = \vecc{\sigma}_0 \vecc{\sigma}_0^*/2$ (fluctuation-dissipation relation of the first kind) implies that $\hat{\vecc{H}}(z) = \hat{\vecc{H}}^*(z)$ and simplifies the formula to:
\begin{align} \label{msd_formula2}
E[\vecc{x}(t)\vecc{x}^*(t)] &=  2 \int_0^t \vecc{H}(s) ds + 2m \left(\vecc{H}(t) \mathcal{E} \vecc{H}(t) - \int_0^t \vecc{H}(u) \dot{\vecc{H}}(u) du \right).
\end{align}
To determine the behavior of $E[\vecc{x}(t)\vecc{x}^*(t)]$ as $t \to \infty$, it suffices to investigate the asymptotic behavior of $\hat{\vecc{H}}(z)$, whose formula is given in (ii), as $z \to 0$. Noting that
\begin{equation}
\hat{\vecc{H}}(z) = \frac{1}{z}\left[mz\vecc{I} + \vecc{\gamma}_0 +  \vecc{g}\sum_{i=1,2}\alpha_i \vecc{C}_i(z\vecc{I}+\vecc{\Gamma}_i)^{-1}\vecc{M}_i\vecc{C}_i^* \vecc{h} \right]^{-1}
\end{equation}
and using Assumption \ref{ass_vanishingornot}, we find that, as $z \to 0$,
\begin{align}
&\hat{\vecc{H}}(z) \sim \frac{1}{z}\bigg[\vecc{\gamma}_0 +
\alpha_1 \vecc{g} \vecc{K}_1^{(1)}\vecc{h} + \left(m\vecc{I}-\sum_{j=1,2}\alpha_j \vecc{g}\vecc{K}_j^{(2)}\vecc{h}\right)z \nonumber \\ 
&\hspace{2cm} + \alpha_2 \vecc{g} \vecc{K}_2^{(3)}\vecc{h} z^2 + \alpha_2 \vecc{g} \vecc{K}_2^{(4)}\vecc{h} z^3 + \dots \bigg]^{-1}.
\end{align}
Therefore, if $\vecc{\gamma}_0  = \vecc{\sigma}_0 \vecc{\sigma}_0^*/2$ is non-zero, then $\hat{\vecc{H}}(z) \sim 1/z$ as $z \to 0$. Otherwise, if in addition $\alpha_1 = 1$, then $\hat{\vecc{H}}(z) \sim 1/z$ as $z \to 0$, whereas if in addition $\alpha_1=0$, $\alpha_2=1$, then $\hat{\vecc{H}}(z) \sim 1/z^2$ as $z \to 0$. The results in (iii) and (iv) then follow by applying the Tauberian theorems \cite{feller-vol-2}, which say, in particular, that if $\hat{\vecc{H}}(z) \sim 1/z^\beta$ as $z \to 0$, then $\vecc{H}(t) \sim t^{\beta-1}$ as $t \to \infty$, for $\beta = 1, 2$ here.
\end{itemize}
\end{proof}

\begin{rem} \label{msd_general}
We emphasize that superdiffusion with $E[\vecc{x}(t) \vecc{x}^*(t)]$  behaving as $t^\alpha$ as $t \to \infty$, where $\alpha > 2$, cannot take place when the velocity process converges to a stationary state. For a system to behave this way, the velocity itself has to grow with time. Moreover, we remark that one could obtain a richer class of asymptotic behaviors for the MSD by relaxing the assumption of fluctuation-dissipation relations. 
\end{rem}

To summarize, (i) says that in the case where $\vecc{F}_0 = \vecc{0}$, $\alpha_1 = \alpha_3 = 0$, the $n$th order effective constants characterize the asymptotic behavior of the spectral densities at low frequencies; (ii) provides a formula for the particle's mean-squared displacement, and (iii)-(iv) classify the types of diffusive behavior of the GLE model, in the exactly solvable case of Example \ref{ass_exactsolve}, satisfying the fluctuation-dissipation relations. We emphasize that in the sequel we go beyond the above exactly solvable case;  in particular the coefficients $\vecc{g}$, $\vecc{h}$, $\vecc{\sigma}$, $\vecc{\gamma}_0$, $\vecc{\sigma}_0$ will depend in general on the particle's position. However, the GLE in the exactly solvable case can be viewed as linear approximation to the general GLE \eqref{gle2} (by expanding these coefficients in a Taylor series about a fixed position $\vecc{x}' \in \RR^d$).

In view of Proposition \ref{asympbeh}, the parameters $\alpha_i \in \{0,1\}$ allow us to control diffusive behavior of the generalized Langevin dynamics.  Our GLE models are very general and need not satisfy a fluctuation-dissipation relation. As we will see, these different behaviors motivate our introduction and study of various homogenization schemes  for the GLE. Depending on the physical systems under consideration, one scheme might be more realistic than the others.  It is one of the goals of this paper to explore homogenization schemes for different GLE classes.

\begin{rem} \label{rem_inf_dim}
In finite dimension, it is not possible to realize generalized Langevin dynamics with a noise and/or memory function whose spectral density varies as $1/\omega^p$, $p \in (0,1)$, near $\omega = 0$ (i.e. the so-called $1/f$-type noise \cite{Kupferman2004}), and, consequently, the noise covariance function and/or memory function decay as a power $1/t^\alpha$, $\alpha \in (0,1)$, as $t \to \infty$.  In this case one can use the formula in (ii) of Proposition \ref{asympbeh} to show, at least for the exactly solvable case in Example \ref{ass_exactsolve} where the fluctuation-dissipation relations hold, that the asymptotic behavior of the particle is sub-diffusive, i.e. $E[\vecc{x}(t) \vecc{x}^*(t)] = O(t^\beta)$, where $\beta \in (0,1)$, as $t \to \infty$ (see also the related works \cite{mckinley2018anomalous,didier2019asymptotic}). Sub-diffusive behavior has been discovered in a wide range of statistical and biological systems \cite{kou2008stochastic}, and, therefore, making the study in this case relevant. One could, following the ideas in \cite{glatt2018generalized, 2018arXiv180409682N}, extend the state space of the GLEs to an infinite-dimensional one, in order to study the  sub-diffusive case. Homogenization studies, where more technicalities are expected to be encountered due to the infinite-dimensional nature of the systems, for this case will be explored in a future work. \\
\end{rem}



\subsection{Generalized Langevin Systems as  Input-Output Stochastic Dynamical Systems with Multiple Time Scales}

In this subsection, we discuss GLEs of the form \eqref{gle2}, under Assumptions \ref{ass_bohl}-\ref{ass_vanishingornot}, from the input-output system-theoretic and multiple time scale points of view.

First, we introduce the notion of stochastic dynamical systems.
\begin{defn} \label{stochdynsystem}
A {\it stochastic dynamical system} is a pair $(\vecc{Z},\vecc{\mathcal{F}})$ of vector-valued stochastic processes satisfying equations of the form:
\begin{align}
d\vecc{Z}(t) &= \vecc{A}(t, \vecc{Z}(t)) dt + \vecc{B}(t, \vecc{Z}(t))\vecc{\eta}(t)dt, \\
\vecc{\mathcal{F}}(t) &= \vecc{C}(t, \vecc{Z}(t)),
\end{align}
where $\vecc{A}$, $\vecc{B}$, $\vecc{C}$ are measurable (jointly in $t$ and $\vecc{Z}$) mappings,  $\vecc{\eta}(t)$ is a random process (the {\it input}). $\vecc{Z}(t)$ is called the  {\it state process} and $\vecc{\mathcal{F}}(t)$ the {\it output process} (observation process). The system is {\it linear} if all the mappings are at most linear in $\vecc{Z}$; otherwise the system is {\it nonlinear}. The system is {\it time-invariant} if all the mappings are independent of $t$. 
\end{defn}



The equation for the particle's position, together with the GLE \eqref{gle2}, can be cast as the system of SDEs for the Markov process  

\noindent $\vecc{z}_t := (\vecc{x}_{t}, \vecc{v}_{t}, \vecc{y}^1_{t}, \vecc{y}^2_t, \vecc{\beta}^3_{t},  \vecc{\beta}^4_{t}) \in \RR^{d}\times \RR^d \times \RR^{d_1} \times \RR^{d_2} \times \RR^{d_3} \times \RR^{d_4}$:
\begin{align}
d\vecc{x}_{t} &= \vecc{v}_{t} dt, \label{sd1} \\
m d\vecc{v}_{t} &= -\vecc{\gamma}_0(t, \vecc{x}_t) \vecc{v}_t dt + \vecc{\sigma}_0(t, \vecc{x}_t) d\vecc{W}_t^{(k)}  - \vecc{g}(t, \vecc{x}_{t}) \sum_{i=1,2} \alpha_i \vecc{C}_i \vecc{y}^i_{t} dt \nonumber \\ &\ \ \ \ + \vecc{\sigma}(t, \vecc{x}_{t}) \sum_{j=3,4} \alpha_j \vecc{C}_j \vecc{\beta}^j_{t} dt  + \vecc{F}_e(t, \vecc{x}_{t})dt, \\
d\vecc{y}^i_{t} &= -\vecc{\Gamma}_i \vecc{y}^i_{t} dt + \vecc{M}_i \vecc{C}_i^* \vecc{h}(t,\vecc{x}_{t}) \vecc{v}_{t} dt, \ \ i=1,2,\\  d\vecc{\beta}^j_{t} &= -\vecc{\Gamma}_j \vecc{\beta}^j_{t} dt +  \vecc{\Sigma}_j d\vecc{W}^{(q_j)}_{t}, \ \ j=3,4, \label{sd6}
\end{align}
where we have defined the auxiliary {\it memory processes}:
\begin{equation}
\vecc{y}^i_{t} := \int_{0}^{t} e^{-\vecc{\Gamma}_i(t-s)} \vecc{M}_i \vecc{C}_i^* \vecc{h}(s,\vecc{x}_{s}) \vecc{v}_{s} ds \in \RR^{d_i}, \ \ i=1,2.
\end{equation}

It is easy to see that the pairs $(\vecc{\beta}^i,\vecc{\xi}^i)$, $i=3,4$, defined in the previous subsection, are linear time-invariant Gaussian stochastic dynamical systems with a white noise input (and therefore the state processes $\vecc{\beta}^j(t)$ are Markov)  in the sense of Definition \ref{stochdynsystem}. Also, the pairs $(\vecc{y}_t^i,\vecc{C}_i\vecc{y}_t^i )$  $(i=1,2$) are linear stochastic dynamical systems driven by the random processes $\vecc{M}_i\vecc{C}_i^*\vecc{h}(t, \vecc{x}_t)\vecc{v}_t$, which depend on the particle's position and velocity variables. The generalized Langevin system can  be viewed as a nonlinear stochastic dynamical system $(\vecc{z},\vecc{\mathcal{F}})$, where the components of $\vecc{z}$ satisfy the SDEs  \eqref{sd1}-\eqref{sd6} and $\vecc{\mathcal{F}}$ is a measurable mapping describing an output process or a quantity of interest, for instance,
\begin{equation}
\vecc{\mathcal{F}} = E \left[ \sup_{t \in [0,T]} |\vecc{x}_t|^p \right]
\end{equation}
for $p>0$ and $T>0$.  In the exactly solvable case of Example \ref{ass_exactsolve}, the generalized Langevin system reduces to a linear time-invariant stochastic dynamical system and can be viewed as a network of input-output systems consisting of components modeling the memory and noise. One of the goals of homogenization of GLEs is to reduce the number of the components needed to describe the effective dynamics in the considered limit.

It is natural question, what class of  GLEs  should be taken as the starting point for homogenization. For feasible treatment, the GLEs should be in some sense minimal. In the network interpretation, the original system should be completely described by  a minimal number of components, with no redundancies.  We will discuss this based on a time scale analysis in the following.

The (discrete) spectrum of the $\vecc{\Gamma}_i$ ($i=1,2$) and  of the $\vecc{\Gamma}_j$ ($j=3,4$) (or equivalently, the  {\it spectrum of the Bohl memory function} $\vecc{\kappa}(t)$ and that of the covariance function $\vecc{R}(t)$--- see Definition 2.5 in \cite{trentelman2002control}) encode information about the memory time scales and noise correlation time scales present in the generalized Langevin system respectively. In realistic experiments, there may be many, possibly infinitely many, time scales (each corresponding to a mode of the environment), but typically they cannot be all observed and/or controlled.  When modeling a system, it is important to focus on those time scales that are controllable and observable.
This motivates the following definition, closely related to the notions of controllable and observable eigenvalues from the systems theory \cite{trentelman2002control}.

\begin{defn} \label{defn_timescales}
Consider a  linear stochastic dynamical system $(\vecc{Z},\vecc{\mathcal{F}})$, as in Definition \ref{stochdynsystem}, where $\vecc{A}\in \RR^{n \times n}$, $\vecc{B} \in \RR^{n \times k}$, $\vecc{C} \in \RR^{m \times n}$ are constant matrices. The time scale, $\tau := 1/\lambda$, where $\lambda$ is an eigenvalue of $\vecc{A}$, in the system, is called {\it $(\vecc{A}, \vecc{B})$-controllable} (or simply controllable) if $rank[\vecc{A}-\lambda \vecc{I} \ \ \vecc{B} ] = n$ and {\it $(\vecc{C}, \vecc{A})$-observable} (or simply observable) if $rank[\vecc{A}-\lambda \vecc{I} \ \ \vecc{C} ]^* = n$.
\end{defn}

The following proposition, which follows from Theorem 3.13 in \cite{trentelman2002control}, states well-known results regarding the above notions.

\begin{prop} Consider the linear dynamical system defined in Definition \ref{defn_timescales}. Then
\begin{itemize}
\item[(i)] the system is controllable (more precisely, $(\vecc{A},\vecc{B})$-controllable, i.e. 

\noindent $[\vecc{B} \ \ \vecc{A}\vecc{B} \ \ \cdots \ \ \vecc{A}^{n-1}\vecc{B}]$ is full rank) if and only if every time scale of the system is controllable.
\item[(ii)]  the system is observable (more precisely, $(\vecc{C},\vecc{A})$-observable, i.e. 

\noindent $[\vecc{C} \ \ \vecc{C}\vecc{A} \ \ \cdots \ \ \vecc{C}\vecc{A}^{n-1}]^{*}$ is full rank) if and only if every time scale of the system is observable.
\end{itemize}
\end{prop}

For $i=1,2,3,4$ we define the time scales, $\tau_{i,k_i} := 1/\lambda_{i,k_i}$,  where $\lambda_{i,k_i}$ ($k_i=1,\dots,d_i$) are  eigenvalues of $\vecc{\Gamma}_i$.
We refer to  the $\tau_{1,k_1}, \tau_{2,k_2}$ as {\it memory time scales} and the $\tau_{3,k_3}, \tau_{4,k_4}$ as {\it noise correlation time scales}.

Our consideration of GLEs will be based on the following assumption.

\begin{ass} \label{minimal}
All the  memory time scales and the  noise correlation time scales in the generalized Langevin systems described by \eqref{gle2} are controllable and observable.
\end{ass}

From the mathematical point of view, our consideration minimizes the dimension of the state space on which the GLE is realized and therefore minimizes the complexity of the model which will be taken as the starting point for our homogenization studies. Indeed, recall that
a stochastic realization is {\it minimal} if the realized process has no other stochastic realization of
smaller dimension. It follows from our assumptions that all the realizations of the memory function and noise process are minimal, since a sufficient  condition for a linear stochastic dynamical system to be minimal is that it is controllable (or reachable in the language of \cite{lindquist2015linear}), observable and the spectral factor of its spectral density is minimal \cite{lindquist2015linear}.

\section{On the Homogenization of Generalized Langevin Dynamics}
\label{sect_homogofGLE}

In this section, we discuss some new directions for homogenization of GLEs. 

In the case of non-vanishing (first order) effective damping constant and effective diffusion constant, homogenization of a  version of the GLE \eqref{gle2} was studied in \cite{LimWehr_Homog_NonMarkovian}, where a limiting SDE for the position process was obtained in the limit, in which all the characteristic time scales of the system (i.e. the inertial time scale, the memory time scale and the noise correlation time scale) tend to zero at the same rate. Extending this result, we are going to focus on the following two cases.
\begin{itemize}
\item[(A)] {\it The case where an instantaneous damping term is present in the GLE, i.e.  $\vecc{F}_0\neq \vecc{0}$, or the non-vanishing effective damping constant case, i.e. $\alpha_1 = 1$.} Together with the conditions in Example \ref{ass_exactsolve}, this gives a model for normally diffusing systems; see Proposition \ref{asympbeh} (iii). One can study the limit in which the inertial time scale and a subset (possibly all or none of) of other characteristic time scales of the system tend to zero; in particular the small mass limit in the case $\vecc{F}_0 \neq \vecc{0}$  of the generalized Langevin dynamics. We remark that the small mass limit is not well-defined in the case $\vecc{F}_0 = \vecc{0}$ and $\alpha_1=\alpha_3=1$ -- this was first observed in \cite{mckinley2009transient}, where it was pointed out that the limit leads to the phenomenon of anomalous gap of the particle's mean-squared displacement (see also \cite{cordoba2012elimination,indei2012treating}). \\
\item[(B)] {\it The vanishing effective damping constant and effective diffusion constant case, i.e. $\vecc{F}_0=\vecc{0}$, $\alpha_1=\alpha_3=0$, $\alpha_2=\alpha_4=1$.} Together with the conditions in Example \ref{ass_exactsolve}, this gives a model for systems with super-diffusive behavior; see Proposition \ref{asympbeh} (iv). One can study the limit in which the inertial time scale, a subset of the memory time scales and a subset of the noise correlation time scales tend to zero at the same rate. Such effective models are physically relevant when they preserve the asymptotic behavior of the spectral densities at low and/or high frequencies in the limit.  Situations are also possible, where some of the eigenmodes of the memory and noise spectrum are damped much stronger than other, for example due to  an injection of monochromatic light  from a laser into the system, which is originally in thermal equilibrium. This justifies studying homogenization limits that selectively target a part of frequencies of memory and noise.
\end{itemize}
We will study homogenization of  the GLE \eqref{gle2} in the limits described in the above scenarios.   In all cases, the inertial time scale is taken to zero -- this gives rise to the singular nature of the limit problems. We remark that one could also consider the more interesting scenarios in which the time scales tend to zero at different rates, but we choose not to pursue this in this already lengthy paper. \\

\noindent {\bf Notation.} Throughout the paper, we denote the variables in the pre-limit equations by small letters (for instance, $\vecc{x}^\epsilon(t)$), and those of the limiting equations by capital letters (for instance, $\vecc{X}(t)$). We use Einstein's summation convention on repeated indices. The Euclidean norm of an arbitrary vector $\vecc{w}$ is denoted by $| \vecc{w} |$ and the (induced operator) norm of a matrix $\vecc{A}$ by $\| \vecc{A} \|$.
For an $\RR^{n_2 \times n_3}$-valued function $\vecc{f}(\vecc{y}):=([f]_{jk}(\vecc{y}))_{j=1,\dots,n_2; k=1,\dots, n_3}$, $\vecc{y} := ([y]_1, \dots, [y]_{n_1}) \in \RR^{n_1}$, we denote by $(\vecc{f})_{\vecc{y}}(\vecc{y})$ the $n_1 n_2 \times n_3$ matrix:
\begin{equation}
(\vecc{f})_{\vecc{y}}(\vecc{y}) = (\vecc{\nabla}_{\vecc{y}}[f]_{jk}(\vecc{y}))_{j=1,\dots, n_2; k=1,\dots,n_3},
\end{equation}
where $\vecc{\nabla}_{\vecc{y}}[f]_{jk}(\vecc{y})$ stands for the gradient vector $\left(\frac{\partial [f]_{jk}(\vecc{y})}{\partial [y]_1}, \dots, \frac{\partial [f]_{jk}(\vecc{y})}{\partial [y]_{n_1}}\right) \in \RR^{n_1}$ for every $j,k$.
We denote by $\vecc{\nabla} \cdot$ the divergence operator which contracts a matrix-valued function to a vector-valued function, i.e.  for the matrix-valued function $\vecc{A}(\vecc{X})$, the $i$th component of its divergence is given by $(\vecc{\nabla} \cdot \vecc{A})^i = \sum_j \frac{\partial A^{ij}}{\partial X^j}$.
Lastly, the symbol $\mathbb{E}$ denotes expectation with respect to the probability measure $\mathbb{P}$.

\section{Small Mass Limit of Generalized Langevin Dynamics}
\label{sect_newsmallmlimit}

Consider  the following family of equations for the processes $(\vecc{x}_t^m, \vecc{v}_t^m) \in \RR^{d \times d}$, $t \in [0,T]$, $m>0$:
\begin{align}
d\vecc{x}_t^m &= \vecc{v}_t^m dt, \label{res_gle_smallmass1} \\
m d\vecc{v}^m_t &= -\vecc{\gamma}_0(t, \vecc{x}_t^m) \vecc{v}_t^m dt - \vecc{g}(t, \vecc{x}_t^m) \left(\int_0^t \vecc{\kappa}(t-s) \vecc{h}(s, \vecc{x}_s^m) \vecc{v}_s^m ds \right) dt  \nonumber \\  
&\ \ \ \ + \vecc{\sigma}_0(t, \vecc{x}_t^m) d\vecc{W}_t^{(k)} + \vecc{\sigma}(t, \vecc{x}_t^m) \vecc{\xi}_t dt + \vecc{F}_e(t, \vecc{x}_t^m) dt,  \label{res_gle_smallmass}
\end{align}
where $\vecc{\kappa}(t)$ and $\vecc{\xi}_t$ are the memory function and noise process defined in \eqref{memory_realized} and \eqref{noise} respectively, with each of the $\alpha_i$ ($i=1,2,3,4$) equal to zero or to one.  The equations \eqref{res_gle_smallmass1}-\eqref{res_gle_smallmass} are equivalent to the following system of SDEs for the Markov process $\vecc{z}^m_t := (\vecc{x}^m_{t}, \vecc{v}^m_{t}, \vecc{y}^{1,m}_{t}, \vecc{y}^{2,m}_t, \vecc{\beta}^{3,m}_{t},  \vecc{\beta}^{4,m}_{t}) \in \RR^{d}\times \RR^d \times \RR^{d_1} \times \RR^{d_2} \times \RR^{d_3} \times \RR^{d_4}$:
\begin{align}
d\vecc{x}^m_{t} &= \vecc{v}^m_{t} dt, \label{res_sd1} \\
m d\vecc{v}^m_{t} &= -\vecc{\gamma}_0(t, \vecc{x}^m_t) \vecc{v}^m_t dt + \vecc{\sigma}_0(t, \vecc{x}^m_t) d\vecc{W}_t^{(k)}  - \vecc{g}(t, \vecc{x}^m_{t}) \sum_{i=1,2} \alpha_i \vecc{C}_i \vecc{y}^{i,m}_{t} dt \nonumber \\ 
&\ \ \ \ + \vecc{\sigma}(t, \vecc{x}^m_{t}) \sum_{j=3,4} \alpha_j \vecc{C}_j \vecc{\beta}^{j,m}_{t} dt + \vecc{F}_e(t, \vecc{x}^m_{t})dt, \\
d\vecc{y}^{i,m}_{t} &= -\vecc{\Gamma}_i \vecc{y}^{i,m}_{t} dt + \vecc{M}_i \vecc{C}_i^* \vecc{h}(t, \vecc{x}^m_{t}) \vecc{v}^m_{t} dt, \ \ i=1,2,\\  d\vecc{\beta}^{j,m}_{t} &= -\vecc{\Gamma}_j \vecc{\beta}^{j,m}_{t} dt +  \vecc{\Sigma}_j d\vecc{W}^{(q_j)}_{t}, \ \ j=3,4, \label{res_sd6}
\end{align}
where we have defined the auxiliary memory processes:
\begin{equation}
\vecc{y}^{i,m}_{t} := \int_{0}^{t} e^{-\vecc{\Gamma}_i(t-s)} \vecc{M}_i \vecc{C}_i^* \vecc{h}(s, \vecc{x}^m_{s}) \vecc{v}^m_{s} ds \in \RR^{d_i}, \ \ i=1,2.
\end{equation}
Note that the processes $\vecc{\beta}_t^{3,m}$ and $\vecc{\beta}_t^{4,m}$ do not actually depend on $m$, but we are adding the superscript $m$ for a more homogeneous notation.

We make the following simplifying assumptions concerning \eqref{res_sd1}-\eqref{res_sd6}.   Let $\vecc{W}^{(q_j)}$ ($j=3,4$) be independent Wiener processes on a filtered probability space $(\Omega, \mathcal{F}, \mathcal{F}_t,\mathbb{P})$ satisfying the usual conditions and let $\mathbb{E}$ denote expectation with respect to $\mathbb{P}$.


\begin{ass} \label{exis_gle}
There are no explosions, i.e. almost surely, for every $m > 0$ there exists global unique solution to the pre-limit SDE  
\eqref{res_sd1}-\eqref{res_sd6} and also to the limiting SDEs \eqref{sm1}-\eqref{sm3} on the time interval $[0,T]$.
\end{ass}

\begin{ass} \label{bounded}
For $t \in \RR^+$,  $\vecc{y} \in \RR^{d}$, the functions  $\vecc{F}_e(t, \vecc{y})$, $\vecc{\sigma}_0(t,\vecc{y})$ and $\vecc{\sigma}(t,\vecc{y})$ are continuous and bounded (in $t$ and $\vecc{y}$) as well as Lipschitz in $\vecc{y}$, whereas the functions $\vecc{\gamma}_0(t, \vecc{y})$, $\vecc{g}(t, \vecc{y})$, $\vecc{h}(t, \vecc{y})$, $(\vecc{\gamma}_0)_{\vecc{y}}(t, \vecc{y})$, $(\vecc{g})_{\vecc{y}}(t, \vecc{y})$ and $(\vecc{h})_{\vecc{y}}(t, \vecc{y})$ are continuously differentiable and Lipschitz in $\vecc{y}$ as well as bounded (in $t$ and $\vecc{y}$).
Moreover, the  functions $(\vecc{\gamma}_0)_{\vecc{y}\vecc{y}}(t, \vecc{y})$, $(\vecc{g})_{\vecc{y}\vecc{y}}(t, \vecc{y})$ and $(\vecc{h})_{\vecc{y}\vecc{y}}(t, \vecc{y})$  are bounded for every $t \in \RR^+$, $\vecc{y} \in \RR^{d}$. 
\end{ass}


\begin{ass} \label{initialdata}
The initial data $\vecc{x}, \vecc{v} \in \RR^d$ are $\mathcal{F}_0$-measurable random variables independent of the $\sigma$-algebra generated by the Wiener processes $\vecc{W}^{(q_j)}$ ($j=3,4$). They are independent of $m$ and have finite moments of all orders.
\end{ass}


The following theorem describes the homogenized behavior of the particle's position modeled by the family of the equations \eqref{res_gle_smallmass1}-\eqref{res_gle_smallmass}---or, equivalently, by the SDE systems \eqref{res_sd1}-\eqref{res_sd6}---in the limit as the particle's mass tends to zero.

\begin{thm} \label{newsmallm} 
Let $\vecc{z}_t^m := (\vecc{x}_t^m, \vecc{v}_t^m, \vecc{y}_t^{1,m}, \vecc{y}_t^{2,m}, \vecc{\beta}_t^{3,m}, \vecc{\beta}_t^{4,m}) $ be a family of processes solving the SDE system \eqref{res_sd1}-\eqref{res_sd6}.
Suppose that Assumptions \ref{ass_bohl}-\ref{minimal} and  Assumptions \ref{exis_gle}-\ref{initialdata} hold. In addition, suppose that for every $m > 0$, $\vecc{x} \in \RR^d$, the family of matrices $\vecc{\gamma}_0(t, \vecc{x})$ is positive stable, uniformly in $t$ and $\vecc{x}$.  Then as $m \to 0$, the position process $\vecc{x}^m_t$ converges to $\vecc{X}_t$, where $\vecc{X}_t$ is the first component of the process $(\vecc{X}_t, \vecc{Y}_t^1, \vecc{Y}_t^2, \vecc{\beta}_t^3, \vecc{\beta}_t^4)$ satisfying the It\^o SDE system:
\begin{align}
d\vecc{X}_t &= \vecc{\gamma}_0^{-1}(t, \vecc{X}_t)\bigg[ - \vecc{g}(t, \vecc{X}_t) \sum_{i=1}^2 \alpha_i \vecc{C}_i \vecc{Y}_t^i + \vecc{\sigma}(t, \vecc{X}_t) \sum_{j=3}^4 \alpha_j \vecc{C}_j \vecc{\beta}_t^j \nonumber \\
&\ \ \ \ + \vecc{F}_e(t, \vecc{X}_t) \bigg] dt  + \vecc{\gamma}_0^{-1}(t, \vecc{X}_t)\vecc{\sigma}_0(t, \vecc{X}_t)d\vecc{W}_t^{(k)} +\vecc{S}^{(0)}(t, \vecc{X}_t)dt, \label{sm1} \\
d\vecc{Y}_t^k &= -\vecc{\Gamma}_k \vecc{Y}_t^k dt + \vecc{M}_k \vecc{C}_k^* \vecc{h}(t, \vecc{X}_t)\vecc{\gamma}_0^{-1}(t, \vecc{X}_t) \bigg[ - \vecc{g}(t, \vecc{X}_t) \sum_{i=1}^2 \alpha_i \vecc{C}_i \vecc{Y}_t^i  \nonumber \\ 
&\ \ \ \ +  \vecc{\sigma}(t, \vecc{X}_t) \sum_{j=3}^4 \alpha_j \vecc{C}_j \vecc{\beta}_t^j + \vecc{F}_e(t,\vecc{X}_t) \bigg] dt + \vecc{S}^{(k)}(t, \vecc{X}_t) dt \nonumber \\
&\ \ \ \ + \vecc{M}_k \vecc{C}_k^* \vecc{h}(t, \vecc{X}_t)\vecc{\gamma}_0^{-1}(t, \vecc{X}_t)\vecc{\sigma}_0(t, \vecc{X}_t)d\vecc{W}_t^{(k)}, \ \ \text{for } k=1,2, \label{sm2} \\
d\vecc{\beta}_t^l &= -\vecc{\Gamma}_l \vecc{\beta}_t^l dt + \vecc{\Sigma}_l d\vecc{W}_t^{(q_l)}, \ \ \text{ for } l=3,4, \label{sm3}
\end{align}
where the $i$th component of the  $\vecc{S}^{(k)}$ ($k=0,1,2$) is given by:
\begin{align}
S_i^{(0)}(t, \vecc{X}) &= \frac{\partial}{\partial X_l}\left((\vecc{\gamma}_0^{-1})_{ij}(t, \vecc{X}) \right) J_{lj}, \ \ j,l=1,\dots,d, \end{align}
and for $k=1,2$,
\begin{align}
S_i^{(k)}(t, \vecc{X}) &= \frac{\partial}{\partial X_l}\left((\vecc{M}_k \vecc{C}_k^* \vecc{h}(t, \vecc{X}) \vecc{\gamma}_0^{-1}(t, \vecc{X}))_{ij} \right) J_{lj}, \ \ j,l=1,\dots,d,
\end{align}
with $\vecc{J} \in \RR^{d \times d}$ solving the Lyapunov equation, $\vecc{\gamma}_0 \vecc{J} + \vecc{J} \vecc{\gamma}_0^* = \vecc{\sigma}_0 \vecc{\sigma}_0^*$.
The convergence is obtained in the following sense:  for all finite $T>0$, $\sup_{t \in [0,T]} |\vecc{x}^m_t - \vecc{X}_t|  \to 0$ in probability, as $m \to 0$.
\end{thm}

\begin{proof}
We prove the theorem by applying Theorem \ref{mainthm}.  
Using the notation in the statement of Theorem \ref{mainthm}, let $\epsilon = m$, $n_1 =d+d_1+d_2+d_3+d_4$, $n_2 = d$, $k_1 = q_3 + q_4$, $k_2 = k$, $\vecc{x}^\epsilon(t) = (\vecc{x}_t^m, \vecc{y}_t^{1,m}, \vecc{y}_t^{2,m}, \vecc{\beta}_t^{3,m}, \vecc{\beta}_t^{4,m})$, $\vecc{v}^\epsilon(t) = \vecc{v}_t^m$,
\begin{align}
\vecc{a}_1 &= [\vecc{I} \ \  \vecc{M}_1 \vecc{C}_1^* \vecc{h}(t, \vecc{x}_t^m) \  \ \vecc{M}_2 \vecc{C}_2^* \vecc{h}(t, \vecc{x}_t^m) \ \  \vecc{0} \  \ \vecc{0}], \\
\vecc{a}_2 &= -\vecc{\gamma}_0(t, \vecc{x}_t^m), \\
\vecc{b}_1 &= -(\vecc{0},\vecc{\Gamma}_1 \vecc{y}_t^{1,m}, \vecc{\Gamma}_2 \vecc{y}_t^{2,m}, \vecc{\Gamma}_3 \vecc{\beta}_t^{3,m}, \vecc{\Gamma}_4 \vecc{\beta}_t^{4,m}), \\
\vecc{b}_2 &= \vecc{F}_e(t, \vecc{x}_t^m) - \vecc{g}(t, \vecc{x}_t^m) \sum_{i=1,2} \alpha_i \vecc{C}_i \vecc{y}_t^{i,m} + \vecc{\sigma}(t, \vecc{x}_t^m) \sum_{j=3,4} \alpha_j \vecc{C}_j \vecc{\beta}_t^{j,m}, \\
\vecc{\sigma}_1 &= \begin{bmatrix}
\vecc{0} & \vecc{0} \\
\vecc{0} & \vecc{0} \\
\vecc{0} & \vecc{0} \\
\vecc{\Sigma}_3 & \vecc{0} \\
\vecc{0} & \vecc{\Sigma}_4
\end{bmatrix},\\
\vecc{\sigma}_2 &= \vecc{\sigma}_0(t, \vecc{x}_t^m),
\end{align}
$\vecc{W}^{(k_1)}(t) = (\vecc{W}_t^{(q_3)}, \vecc{W}_t^{(q_4)})$ and $\vecc{W}^{(k_2)}(t) = \vecc{W}_t^{(k)}$.
The initial conditions are $\vecc{x}(0) = (\vecc{x}, \vecc{0}, \vecc{0}, \vecc{\beta}_0^3, \vecc{\beta}_0^4)$ and $\vecc{v}(0) = \vecc{v}$, where $\vecc{\beta}_0^j$ $(j=3,4$) are normally distributed with mean-zero and covariance $\vecc{M}_j$.  They are independent of $m$.

Observe that in the above formula, $\vecc{a}_i$, $\vecc{b}_i$, $\vecc{\sigma}_i$ ($i=1,2$) do not depend explicitly on $\epsilon = m$, so by the convention adopted earlier, we denote them $\vecc{A}_i$, $\vecc{B}_i$, $\vecc{\Sigma}_i$ respectively, and we put $a_i = b_i = c_i = d_i = \infty$, where $a_i, b_i, c_i, d_i$ are the rates in Assumption \ref{a5_ch2}.

Next, we verify the assumptions of Theorem \ref{mainthm}.
Assumption \ref{aexis} clearly follows from the Assumption \ref{exis_gle}.
Since the family of matrices $\vecc{\gamma}_0(t, \vecc{x})$ is positive stable (uniformly in $t$ and $\vecc{x}$), Assumption \ref{a0_ch2} is satisfied. It is straightforward to see that our assumptions on the coefficients of the GLE imply Assumption \ref{a1_ch2}. As $\vecc{x}(0)$ and $\vecc{v}(0)$ are random variables independent of $m$, Assumption \ref{a2_ch2} holds by our assumptions on the initial conditions $\vecc{x}_0$, $\vecc{v}_0$ and $\vecc{\beta}^j_0$ ($j=3,4$). Finally, as noted earlier, Assumption \ref{a5_ch2} holds with $a_i = b_i = c_i = d_i = \infty$.
The assumptions of the Theorem \ref{mainthm} are thus satisfied.   Applying  it, we obtain the limiting SDE system \eqref{sm1}-\eqref{sm3}.

\end{proof}

We remark that the  limiting SDE is unique up to transformation in \eqref{transf_realize}, as pointed out already in \cite{LimWehr_Homog_NonMarkovian}.

\begin{rem}
In the special case when $\alpha_i=0$ for $i=1,2,3,4$ and the coefficients do not depend on $t$ explicitly, Theorem \ref{newsmallm} reduces to the result obtained in \cite{hottovy2015smoluchowski}. In general, by comparing the result with the one obtained in  \cite{hottovy2015smoluchowski}, we see that perturbing the original Markovian system by adding a memory and colored noise changes the behavior of the homogenized system obtained in the small mass limit. In particular,
\begin{itemize}
\item[(i)] the limiting equation for the particle's position not only contains a correction drift term ($\vecc{S}^{(0)}$) -- the {\it noise-induced drift}, but is also coupled to equations for other slow variables;
\item[(ii)] in the case when $\alpha_1$ and/or $\alpha_2$  equal $1$, the limiting equation for the (slow) auxiliary memory variables contains correction drift terms ($\vecc{S}^{(1)}$ and/or $\vecc{S}^{(2)}$) -- which could be called the {\it memory-induced drifts}. Interestingly, the memory-induced drifts disappear when $\vecc{h}$ is proportional to $\vecc{\gamma}_0$, a phenomenon that can be attributed to the interaction between the forces $\vecc{F}_0$ and $\vecc{F}_1$. 
\end{itemize}
Note that the highly coupled structure of the limiting SDEs is due to the fact that only one time scale (inertial time scale) was taken to zero in the limit. We expect the structure to simplify when all time scales present in the problem are taken to zero at the same rate.
\end{rem}

\section{Homogenization for the Case of Vanishing Effective Damping Constant and Effective Diffusion Constant}
\label{sect_newhomogcase}

In this section we consider the GLE \eqref{gle2}, with $\vecc{F}_0 = \vecc{0}$,  $\alpha_1=\alpha_3=0$, and $\alpha_2 = \alpha_4 = 1$.  We explore a class of homogenization schemes, aiming to:
\begin{itemize}
\item[(P1)] reduce the complexity of the generalized Langevin dynamics in a way that the homogenized dynamics can be realized on a state space with minimal dimension and are described by minimal number of effective parameters;
\item[(P2)]  retain non-trivial effects of the memory and the colored noise in the homogenized dynamics by matching the asymptotic behavior of the spectral density of the noise process and memory function in the original and the effective model.
\end{itemize}

\begin{rem} \label{hmm}
Generally, the larger the number of time scales (the eigenvalues of the $\vecc{\Gamma}_i$) present in the system, the higher the dimension of the state space needed to realize the generalized Langevin system. On the other hand, in addition to $\vecc{\Gamma}_i$, information on $\vecc{C}_i$ and $\vecc{M}_i$ is needed to determine the asymptotic behavior of the spectral densities (see Proposition \ref{asympbeh}(i)). In other words, although  analysis based solely on  time scales consideration may reduce the dimension of the model, it does not in general allow one to achieve the model matching in (P2).
It is desirable to have homogenization schemes that achieve both goals of dimension reduction (P1) and matching of models (P2).  Such a scheme is considered below.
\end{rem}

The idea is to consider the limit when the inertial time scale, a proper subset of the memory time scales and a proper subset of the noise correlation time scales tend to zero at the same rate. The case of sending all the characteristic time scales to zero is excluded here as it is uninteresting when the effective damping and diffusion vanish in the limit.

Recall that the notions of controllability and observability are invariant under the trivial equivalence relation of type \eqref{transf_realize}. Therefore, one can, without loss of generality, assume that the $\vecc{\Gamma}_i$ $(i=1,2,3,4)$ are already in the Jordan normal form and work in Jordan basis. Such form will reveal the slow-fast time scale structure of the system and so give us a rubric to develop homogenization schemes.

\begin{ass} \label{jordan} Let $i=2,4$.
All the $\vecc{\Gamma}_i$ are of the following Jordan normal form:
\begin{equation}
\vecc{\Gamma}_i = diag(\vecc{\Gamma}_{i,1},\cdots,\vecc{\Gamma}_{i,N_i}),
\end{equation}
where  $N_i < d_i$,  $\vecc{\Gamma}_{i,k} \in \RR^{\nu(\lambda_{i,k}) \times \nu(\lambda_{i,k})}$ ($k=1, \dots, N_i$)  is the Jordan block associated with the (controllable and observable) eigenvalue $\lambda_{i,k}$ (or time scale $\tau_{i,k}=1/\lambda_{i,k}$) and corresponds to the invariant subspace $\mathcal{X}_{i,k} = Ker(\lambda_{i,k}\vecc{I}-\vecc{\Gamma}_{i,k})^{\nu(\lambda_{i,k})}$, where $\nu(\lambda_{i,k})$ is the index of $\lambda_{i,k}$, i.e.  the size of the largest Jordan block corresponding to the eigenvalue $\lambda_{i,k}$. Let $1 \leq M_i < N_i$ and the eigenvalues be ordered as $0 < \lambda_{i,1}  \leq  \dots \leq \lambda_{i,M_i} < \lambda_{i,M_{i}+1}  \leq \dots \leq \lambda_{i,N_i}$, so that we have the invariant subspace decomposition, $\RR^{d_i} = \bigoplus_{j=1}^{N_i} \mathcal{X}_{i,j}$, with $d_i = \sum_{k=1}^{N_i} \nu(\lambda_{i,k})$.
 \end{ass}

Let $0 < l_i < d_i$. The following procedure studies generalized Langevin dynamics whose spectral densities of the memory and the noise process have the asymptotic behavior, $\vecc{\mathcal{S}}_i(\omega) \sim \omega^{2l_i}$ for small $\omega$, and $\vecc{\mathcal{S}}_i(\omega) \sim 1/\omega^{2d_i}$ for large $\omega$, for $i=2,4$.  We construct a homogenized version of the model in such a way that its  memory and noise processes have spectral densities whose asymptotic behavior at low $\omega$ matches that of the original model (to achieve (P2)), while that at high $\omega$ it varies as $1/\omega^{2l_i}$ (to achieve (P1)).

\begin{alg} \label{alg} {\it Procedure to study a class of homogenization problems.}
\begin{itemize}
\item[(1)] Let $\alpha_1=\alpha_3 =0$, $\alpha_2=\alpha_4 =1$ and $\vecc{F}_0 = \vecc{0}$ in the GLE \eqref{gle2}.
Suppose that Assumption \ref{jordan} holds and there exists $M_i$ such that $l_i = \sum_{k=1}^{M_i} \nu(\lambda_{i,k})$. Take this $M_i$.
\item[(2)] For $i=2,4$, set $m=m' \epsilon$ and $\lambda_{i,k} = \lambda'_{i,k}/\epsilon$, for $k=M_i+1,\dots,N_i$ (i.e. we scale the $(d_2-l_2)$ smallest memory time scales and the $(d_4-l_4)$ smallest noise correlation time scales with $\epsilon$), where $m'$ and the $\lambda'_{i,k}$ are positive constants.
\item[(3)] Select the $\vecc{C}_i$, $\vecc{M}_i$, $\vecc{\Sigma}_i$ such that the $\vecc{C}_i$ are constant matrices independent of the $\lambda_{i,k}$ ($k=1,\dots,N_i$), $\vecc{C}_i \vecc{\Gamma}_i^{-n_i} \vecc{M}_i \vecc{C}_i^* = \vecc{0}$ for $0 < n_i < 2l_i$, $\vecc{C}_i \vecc{\Gamma}_i^{-(2l_i+1)} \vecc{M}_i \vecc{C}_i^* \neq \vecc{0}$, and upon a suitable rescaling involving the mass, memory time scales and noise correlation time scales the resulting family of  GLEs can be cast in the form of the SDEs \eqref{sde1}-\eqref{sde2}. Note that the matrix entries of the $\vecc{M}_i$ and/or $\vecc{\Sigma}_i$ necessarily depend on the $\lambda_{i,k}$ due to the Lyapunov equations that relate them to the $\vecc{\Gamma}_i$.
\item[(4)] Apply Theorem \ref{mainthm} to study the limit $\epsilon \to 0$ and obtain the homogenized model, under appropriate assumptions on the coefficients and parameters in the GLEs.
\end{itemize}
\end{alg}

We remark that while one has the above procedure to study  homogenization schemes that achieve (P1) and (P2), the derivations and formulas for the limiting equations could become tedious and complicated  as the $l_i$ and $d_i$ become large. To illustrate this, we consider a simple yet still sufficiently general instance of Algorithm \ref{alg} in the following.

\begin{ass}  \label{special}
The spectral densities, $\vecc{\mathcal{S}}_i(\omega) = \vecc{\Phi}_i(i\omega)\vecc{\Phi}_i^*(-i\omega)$ ($i=2,4$), with the (minimal) spectral factor:
\begin{equation}
\vecc{\Phi}_i(z) = \vecc{Q}^{-1}_i(z) \vecc{P}_i(z), \end{equation}
where the $\vecc{P}_i(z) \in \RR^{p_i \times m_i}$ are matrix-valued monomials with degree $l_i$ : \begin{equation}
\vecc{P}_i(z) = \vecc{B}_{l_i}z^{l_i}
\end{equation}
and the $\vecc{Q}_i(z) \in \RR^{p_i \times p_i}$ are matrix-valued polynomials with degree $d_i$, i.e.
\begin{equation} \vecc{Q}_i(z) = \prod_{k=1}^{d_i} (z\vecc{I}+\vecc{\Gamma}_{i,k}). 
\end{equation}
Here $p_2=q$, $p_4=r$, the $m_i$ $(i=2,4$) are positive integers, the $\vecc{B}_{l_i} \in \RR^{p_i \times m_i}$ are constant matrices, $\vecc{\Gamma}_{i,k}\in \RR^{p_i \times p_i}$ are diagonal matrices with positive entries, and  $\vecc{I}$ denotes identity matrix of appropriate dimension.
\end{ass}

Under Assumption \ref{special}, the spectral densities have the following asymptotic behavior: $\vecc{\mathcal{S}}_i(\omega) \sim \omega^{2l_i}$ for small $\omega$, and $\vecc{\mathcal{S}}_i(\omega) \sim 1/\omega^{2d_i}$ for large $\omega$.
One can then implement Algorithm \ref{alg} explicitly to study homogenization for a sufficiently large class of GLEs, where the rescaled spectral densities tend to the ones with the asymptotic behavior mentioned in the paragraph just before Algorithm \ref{alg} in the limit. We discuss one such implementation in Appendix \ref{implem_alg}. Since the calculations become more complicated as $l_i$ and $d_i$ become large, we will only study simpler cases and illustrate how things could get complicated in the following.

We assume $d_2$ and $d_4$ are even integers and consider in detail  the case when
$l_2 = l_4 = l = 1$, $d_2 = d_4 = h = 2$,
\begin{align}
\vecc{\Gamma}_{2,1} &= diag(\lambda_{2,1}, \dots, \lambda_{2,d_2/2}), \ \ \  \vecc{\Gamma}_{2,2}= diag(\lambda_{2,d_2/2+1},\dots,\lambda_{2,d_2}), \\
\vecc{\Gamma}_{4,1}&=diag(\lambda_{4,1},\dots,\lambda_{4,d_4/2}),   \ \ \  \vecc{\Gamma}_{4,2}= diag(\lambda_{4,d_4/2+1},\dots,\lambda_{4,d_4}),
\end{align}
with $\lambda_{2,d_2} \geq \dots \geq \lambda_{2,d_2/2+1}>\lambda_{2,d_2/2}\geq \dots \geq \lambda_{2,1}>0$ and $\lambda_{4,d_4} \geq \dots \geq \lambda_{4,d_4/2+1}>\lambda_{4,d_4/2}\geq \dots \geq \lambda_{4,1}>0$ in Assumption \ref{special}, so that for $i=2,4$,
\begin{equation} \label{gammai}
\vecc{\Gamma}_i = diag(\vecc{\Gamma}_{i,1},\vecc{\Gamma}_{i,2}) \in \RR^{d_i \times d_i}.
\end{equation}
We consider:
\begin{align}
\vecc{C}_i &= [\vecc{B}_i \ \ \vecc{B}_i] \in \RR^{p_i \times d_i}, \label{Ci} \\
\vecc{\Sigma}_i &= \left[ -\vecc{\Gamma}_{i,1}\vecc{\Gamma}_{i,2}(\vecc{\Gamma}_{i,2}-\vecc{\Gamma}_{i,1})^{-1} \  \ \ \vecc{\Gamma}_{i,2}^2(\vecc{\Gamma}_{i,2}-\vecc{\Gamma}_{i,1})^{-1} \right]^*
\in \RR^{d_i \times d_i/2}, \label{Sigmai} \\
\text{ so that } \nonumber  \\
\vecc{M}_i &=  \left[ \begin{array}{cc}  \label{Mi}
\vecc{M}_i^{11} & \vecc{M}_i^{12} \\
\vecc{M}_i^{21} & \vecc{M}_i^{22}
\end{array} \right] \in \RR^{d_i \times d_i}, 
\end{align}
where 
\begin{align}
\vecc{M}_i^{11} &=  \frac{1}{2}\vecc{\Gamma}_{i,1} \vecc{\Gamma}_{i,2}^2 (\vecc{\Gamma}_{i,1}-\vecc{\Gamma}_{i,2})^{-2}, \\
\vecc{M}_i^{12} &= \vecc{M}_i^{21} =  -\vecc{\Gamma}_{i,1} \vecc{\Gamma}^3_{i,2} (\vecc{\Gamma}_{i,1}+\vecc{\Gamma}_{i,2})^{-1}  (\vecc{\Gamma}_{i,1}-\vecc{\Gamma}_{i,2})^{-2}, \\
\vecc{M}_i^{22} &=  \frac{1}{2}\vecc{\Gamma}^3_{i,2} (\vecc{\Gamma}_{i,1}-\vecc{\Gamma}_{i,2})^{-2},
\end{align} 
$p_2 = q$ and $p_4 = r$ as in Assumption \ref{special}. 
One can verify that this is indeed the vanishing effective damping constant and effective diffusion constant case (i.e. $\vecc{C}_i \vecc{\Gamma}_i^{-1} \vecc{M}_i \vecc{C}_i^* = \vecc{0}$ for $i=2,4$). Also, for $i=2,4$, the memory kernel, $\vecc{\kappa}_2(t)$ and covariance function, $\vecc{R}_4(t)$, are of the following bi-exponential form:
\begin{equation}  \label{78}
 \vecc{C}_ie^{-\vecc{\Gamma}_i|t|}\vecc{M}_i \vecc{C}_i^* =  \frac{1}{2}\vecc{B}_i  \vecc{\Gamma}_{i,2}^2(\vecc{\Gamma}_{i,2}^2-\vecc{\Gamma}_{i,1}^2)^{-1} \left( \vecc{\Gamma}_{i,2} e^{-\vecc{\Gamma}_{i,2} |t|} - \vecc{\Gamma}_{i,1} e^{-\vecc{\Gamma}_{i,1} |t|} \right) \vecc{B}_i^*
\end{equation}
and their Fourier transforms are:
\begin{equation} \label{79}
\vecc{\mathcal{S}}_i(\omega)=\vecc{B}_i \vecc{\Gamma}_{i,2}^2 \vecc{B}_i^* \omega^2 ((\omega^2\vecc{I}+\vecc{\Gamma}_{i,1}^2)(\omega^2\vecc{I}+\vecc{\Gamma}_{i,2}^2))^{-1},
\end{equation} which vary as $\omega^2$ near $\omega = 0$.
Note that in the above the $\vecc{B}_i$ do not necessarily commute with the $\vecc{\Gamma}_{i,j}$.


Following step (2) of Algorithm \ref{alg}, we set $m = m_0 \epsilon$,  $\vecc{\Gamma}_{i,2} = \vecc{\gamma}_{i,2}/\epsilon$ for $i=2,4$, where $m_0 > 0$ is a constant and the $\vecc{\gamma}_{i,2}$ are diagonal matrices with positive eigenvalues, in \eqref{78}-\eqref{79}. We consider the family of GLEs (parametrized by $\epsilon > 0$):
\begin{align}
m_0 \epsilon d\vecc{v}_t^\epsilon &= -\vecc{g}(t, \vecc{x}_t^\epsilon) \left(\int_0^t \vecc{\kappa}_2^\epsilon(t-s) \vecc{h}(s, \vecc{x}_s^\epsilon) \vecc{v}_s^\epsilon ds \right) dt + \vecc{\sigma}(t, \vecc{x}_t^\epsilon)  \vecc{C}_4 \vecc{\beta}_t^{4,\epsilon} dt \nonumber \\ 
&\ \ \ \ + \vecc{F}_e(t, \vecc{x}_t^\epsilon) dt,  \label{res_gle_1}  \\
 \epsilon d\vecc{\beta}_t^{4,\epsilon} &= -\vecc{\Gamma}_4 \vecc{\beta}_t^{4,\epsilon} dt + \vecc{\Sigma}_4 d\vecc{W}_t^{(q_4)},  \label{res_gle_2}
\end{align}
where
\begin{equation} \label{res_mkernel}
\vecc{\kappa}_2^\epsilon(t) =  \frac{1}{2}\vecc{B}_2  \vecc{B}_2^*  \vecc{\gamma}_{2,2}^2(\vecc{\gamma}_{2,2}^2 - \epsilon^2 \vecc{\Gamma}_{2,1}^2)^{-1} \left( \frac{\vecc{\gamma}_{2,2}}{\epsilon} e^{-\frac{\vecc{\gamma}_{2,2}}{\epsilon} |t|} - \vecc{\Gamma}_{2,1} e^{-\vecc{\Gamma}_{2,1} |t|} \right)
\end{equation}
and the covariance function of the noise process $\vecc{\xi}_t^\epsilon = \vecc{C}_4 \vecc{\beta}_t^{4,\epsilon}$ is given by 
\begin{equation} \label{res_noiseee}
\vecc{R}_4^\epsilon(t)  =  \frac{1}{2}\vecc{B}_4  \vecc{B}_4^*  \vecc{\gamma}_{4,2}^2(\vecc{\gamma}_{4,2}^2 - \epsilon^2 \vecc{\Gamma}_{4,1}^2)^{-1} \left( \frac{\vecc{\gamma}_{4,2}}{\epsilon} e^{-\frac{\vecc{\gamma}_{4,2}}{\epsilon} |t|} - \vecc{\Gamma}_{4,1} e^{-\vecc{\Gamma}_{4,1} |t|} \right).
\end{equation}

Note that $\vecc{\kappa}_2^\epsilon(t)$ and $\vecc{R}_4^\epsilon(t)$ converge (in the sense of distribution), as $\epsilon \to 0$, to
\begin{equation}
 \frac{1}{2} \vecc{B}_i \vecc{B}_i^* (\delta(t)\vecc{I}-\vecc{\Gamma}_{i,1} e^{-\vecc{\Gamma}_{i,1} |t|}), \end{equation}
with $i=2$ and $i=4$ respectively.  The corresponding spectral densities are
\begin{equation} \label{limitspec}
\vecc{\mathcal{S}}_i(\omega) = \vecc{B}_i\vecc{B}_i^*\omega^2 (\omega^2 \vecc{I}+\vecc{\Gamma}_{i,1}^2)^{-1},
\end{equation}
with $i=2$ and $i=4$ respectively.

Together with the equation for the particle's position, the equations \eqref{res_gle_1}-\eqref{res_gle_2} form the SDE system:
\begin{align}
d\vecc{x}^\epsilon_t &= \vecc{v}^\epsilon_t dt, \label{res_s1} \\
\epsilon m_0 d\vecc{v}^\epsilon_t &= -\vecc{g}(t, \vecc{x}^\epsilon_t)\vecc{B}_2(\vecc{y}_t^{2,1,\epsilon}+\vecc{y}_t^{2,2,\epsilon}) dt + \vecc{\sigma}(t, \vecc{x}^\epsilon_t) \vecc{B}_4 (\vecc{\beta}_t^{4,1,\epsilon}+\vecc{\beta}_t^{4,2,\epsilon})dt \nonumber \\
&\ \ \ \  + \vecc{F}_e(t, \vecc{x}^\epsilon_t) dt, \\
d\vecc{y}_t^{2,1,\epsilon} &= -\vecc{\Gamma}_{2,1}\vecc{y}_t^{2,1,\epsilon} dt + \mathcal{M}_1^\epsilon \vecc{h}(t, \vecc{x}^\epsilon_t)\vecc{v}^\epsilon_t dt,  \\
\epsilon d\vecc{y}_t^{2,2,\epsilon} &= -\vecc{\gamma}_{2,2}\vecc{y}_t^{2,2,\epsilon}dt+\mathcal{M}_2^\epsilon\vecc{h}(t, \vecc{x}^\epsilon_t) \vecc{v}^\epsilon_t dt, \\
d\vecc{\beta}_t^{4,1,\epsilon} &= -\vecc{\Gamma}_{4,1} \vecc{\beta}_t^{4,1,\epsilon} dt + \vecc{\sigma}_1^\epsilon d\vecc{W}_t^{(q_4/2)}, \\
\epsilon d\vecc{\beta}_t^{4,2,\epsilon} &= -\vecc{\gamma}_{4,2} \vecc{\beta}_t^{4,2,\epsilon} dt + \vecc{\sigma}_2^\epsilon d\vecc{W}_t^{(q_4/2)}, \label{res_s6}
\end{align}
where
\begin{align}
\mathcal{M}_1^\epsilon &= \bigg( (2(\epsilon \vecc{\Gamma}_{2,1}-\vecc{\gamma}_{2,2})^2)^{-1} \vecc{\Gamma}_{2,1}\vecc{\gamma}_{2,2}^2  \nonumber \\ 
&\hspace{1cm} -((\epsilon \vecc{\Gamma}_{2,1}-\vecc{\gamma}_{2,2})^2(\epsilon \vecc{\Gamma}_{2,1} + \vecc{\gamma}_{2,2} ) )^{-1} \vecc{\Gamma}_{2,1} \vecc{\gamma}_{2,2}^3 \bigg) \vecc{B}_2^* , \\
\mathcal{M}_2^\epsilon &= \bigg((2(\epsilon \vecc{\Gamma}_{2,1}-\vecc{\gamma}_{2,2})^2)^{-1} \vecc{\gamma}_{2,2}^3 \nonumber \\ 
&\hspace{1cm} - \epsilon ((\epsilon \vecc{\Gamma}_{2,1}-\vecc{\gamma}_{2,2})^2(\epsilon \vecc{\Gamma}_{2,1} + \vecc{\gamma}_{2,2} ) )^{-1} \vecc{\Gamma}_{2,1} \vecc{\gamma}_{2,2}^3 \bigg)\vecc{B}_2^*, \\
\vecc{\sigma}_1^\epsilon &= -(\vecc{\gamma}_{4,2}-\vecc{\Gamma}_{4,1}\epsilon)^{-1} \vecc{\Gamma}_{4,1} \vecc{\gamma}_{4,2}, \\
\vecc{\sigma}_2^\epsilon &= (\vecc{\gamma}_{4,2}-\vecc{\Gamma}_{4,1} \epsilon)^{-1} \vecc{\gamma}_{4,2}^2.
\end{align}


In the following, we take $\epsilon \in  \mathcal{E}$ to be small.  We make the following assumptions, similar to those made in Theorem \ref{newsmallm}.

\begin{ass} \label{exis_gle2}
There are no explosions, i.e. almost surely, for every $\epsilon \in  \mathcal{E}$, there exist unique solutions on the time interval $[0,T]$  to the pre-limit SDEs
\eqref{res_s1}-\eqref{res_s6} and to the limiting SDEs \eqref{lim_2}.
\end{ass}


\begin{ass} \label{initialdata2}
The initial data $\vecc{x}, \vecc{v} \in \RR^d$ are $\mathcal{F}_0$-measurable random variables independent of the $\sigma$-algebra generated by the Wiener processes  $\vecc{W}^{(q_j)}$ ($j=3,4$). They are independent of $\epsilon$ and have finite moments of all orders.
\end{ass}


The following theorem describes the homogenized dynamics of the family of the GLEs \eqref{res_gle_1}-\eqref{res_gle_2} (or equivalently, of the SDEs \eqref{res_s1}-\eqref{res_s6})  in the limit $\epsilon \to 0$, i.e.  when the inertial time scale, one half of the memory time scales and one half of the noise correlation time scales in the original generalized Langevin system tend to zero at the same rate.

\begin{thm} \label{compl}
Consider the family of the GLEs \eqref{res_gle_1}-\eqref{res_gle_2} (or equivalently, of the SDEs \eqref{res_s1}-\eqref{res_s6}).  Suppose that Assumption \ref{bounded} and Assumptions \ref{special}-\ref{initialdata2} hold, with the $\vecc{C}_i$, $\vecc{\Sigma}_i$, $\vecc{M}_i$ and $\vecc{\Gamma}_i$ ($i=2,4$)  given in \eqref{gammai}-\eqref{Sigmai}.


Assume that for every $t \in \RR^+$, $\epsilon > 0$, $\vecc{x} \in \RR^d$,
\begin{equation}
\vecc{I} + \vecc{g}(t, \vecc{x}) \tilde{\vecc{\kappa}}_\epsilon(\lambda) \vecc{h}(t, \vecc{x})/\lambda m_0 \ \text{ and } \  \vecc{I} + \vecc{g}(t, \vecc{x}) \tilde{\vecc{\kappa}}(\lambda) \vecc{h}(t, \vecc{x})/\lambda m_0
\end{equation}
are invertible for all $\lambda$ in the right half plane $\{\lambda \in \CC: Re(\lambda) > 0\}$, where
\begin{equation}
\tilde{\vecc{\kappa}}_\epsilon(z) = \vecc{B}_2(z \vecc{I} + \vecc{\gamma}_{2,2})^{-1} \mathcal{M}_2^\epsilon  \ \text{ and } \
\tilde{\vecc{\kappa}}(z) = \frac{1}{2} \vecc{B}_2 (z\vecc{I} + \vecc{\gamma}_{2,2})^{-1} \vecc{\gamma}_{2,2} \vecc{B}_2^*.
\end{equation}
Also, assume that $\vecc{\nu}(t, \vecc{x}) := \frac{1}{2} \vecc{g}(t, \vecc{x}) \vecc{B}_2 \vecc{B}_2^* \vecc{h}(t, \vecc{x}) $ is invertible for every $t \in \RR^+$,  $\vecc{x} \in \RR^d$.

Then the particle's position, $\vecc{x}^\epsilon_t \in \RR^d$, solving the  family of GLEs, converges as $\epsilon \to 0$, to $\vecc{X}_t \in \RR^d$, where $\vecc{X}_t$ is the first component of the process $\vecc{\theta}_t := (\vecc{X}_t, \vecc{Y}_t, \vecc{Z}_t) \in \RR^{d+d_2/2+d_4/2}$, satisfying the It\^o SDE:
\begin{align} \label{lim_2}
d\vecc{\theta}_t &= \vecc{P}(t,\vecc{\theta}_t) dt + \vecc{Q}(t, \vecc{\theta}_t) dt + \vecc{R}(t, \vecc{\theta}_t) d\vecc{W}_t^{(d_4/2)},
\end{align}
where
\begin{equation} \label{theta}
\vecc{P}(t,\vecc{\theta}) =
\begin{bmatrix}
\vecc{\nu}^{-1}(\vecc{F}_e-\vecc{g}\vecc{B}_2\vecc{Y}_t+\vecc{\sigma}\vecc{B}_4\vecc{Z}_t) \\
-\frac{1}{2} \vecc{\Gamma}_{2,1} \vecc{B}_2^* \vecc{h}  \vecc{\nu}^{-1}(\vecc{F}_e-\vecc{g}\vecc{B}_2\vecc{Y}_t+\vecc{\sigma}\vecc{B}_4\vecc{Z}_t)  - \vecc{\Gamma}_{2,1} \vecc{Y}_t  \\
-\vecc{\Gamma}_{4,1} \vecc{Z}_t
\end{bmatrix},
\end{equation}
\begin{equation}
\vecc{R}(t, \theta) =
\begin{bmatrix}
\vecc{\nu}^{-1} \vecc{\sigma} \vecc{B}_4  \\
-\frac{1}{2} \vecc{\Gamma}_{2,1} \vecc{B}_2^* \vecc{h} \vecc{\nu}^{-1} \vecc{\sigma} \vecc{B}_4  \\
-\vecc{\Gamma}_{4,1}
\end{bmatrix},
\end{equation}
and the $i$th component of $\vecc{Q}$, $i=1,\dots,d+d_2/2+d_4/2$, is given by:
\begin{equation}
Q_i = \frac{\partial}{\partial X_l}\left[ H_{i,j}(t, \vecc{X}) \right] J_{j,l}, \ \ l=1,\dots,d; \ j=1,\dots,d+d_2/2+d_4/2,
\end{equation}
with $\vecc{H}(t, \vecc{X}) = \vecc{T}(t, \vecc{X})\vecc{U}^{-1}(t, \vecc{X}) \in \RR^{(d+d_2/2+d_4/2) \times (d+d_2/2+d_4/2)}$
and 

\noindent $\vecc{J} \in \RR^{(d+d_2/2+d_4/2) \times (d+d_2/2+d_4/2)}$ is the solution to the Lyapunov equation $\vecc{U}\vecc{J}+\vecc{J}\vecc{U}^* = diag(\vecc{0},\vecc{0},\vecc{\gamma}_{4,2}^2)$, where
\begin{equation} \label{TU}
\vecc{T} =
\begin{bmatrix}
\vecc{I} & \vecc{0} & \vecc{0} \\
-\frac{1}{2}\vecc{\Gamma}_{2,1} \vecc{B}_2^* \vecc{h} & \vecc{0} & \vecc{0} \\
\vecc{0} & \vecc{0} & \vecc{0}
\end{bmatrix}, \ \ \
\vecc{U} =
\begin{bmatrix}
\vecc{0} & \vecc{g} \vecc{B}_2/m_0 & -\vecc{\sigma} \vecc{B}_4/m_0 \\
-\frac{1}{2}\vecc{\gamma}_{2,2}\vecc{B}_2^* \vecc{h} & \vecc{\gamma}_{2,2} & \vecc{0} \\
\vecc{0} & \vecc{0} & \vecc{\gamma}_{4,2}
\end{bmatrix}.
\end{equation}
The convergence holds in the same sense as in Theorem \ref{newsmallm}, i.e. for all finite $T>0$,
$\sup_{t \in [0,T]} |\vecc{x}^\epsilon_t - \vecc{X}_t|  \to 0$
in probability, as $\epsilon \to 0$.
\end{thm}

\begin{proof}
We  apply Theorem \ref{mainthm} to the SDEs \eqref{res_s1}-\eqref{res_s6}. To this end, we set, in Theorem \ref{mainthm}, $n_1 = n_2 = d+d_2/2+d_4/2$, $k_1 = k_2 = d_4/2$ and
\begin{equation}
\vecc{x}^\epsilon(t) = (\vecc{x}^\epsilon_t, \vecc{y}_t^{2,1,\epsilon},\vecc{\beta}_t^{4,1,\epsilon}), \ 
\vecc{v}^\epsilon(t) = (\vecc{v}^\epsilon_t, \vecc{y}_t^{2,2,\epsilon}, \vecc{\beta}_t^{4,2,\epsilon}) \in \RR^{d+d_2/2+d_4/2}, \end{equation}
\begin{align}
\vecc{a}_1(t, \vecc{x}^\epsilon(t),\epsilon) &=
\begin{bmatrix}
\vecc{I} & \vecc{0} & \vecc{0} \\
\mathcal{M}_1^\epsilon \vecc{h}(t, \vecc{x}^\epsilon_t) & \vecc{0} & \vecc{0} \\
\vecc{0} & \vecc{0} & \vecc{0}
\end{bmatrix} \in \RR^{(d+d_2/2+d_4/2) \times (d+d_2/2+d_4/2)}, \\
\vecc{a}_2(t, \vecc{x}^\epsilon(t),\epsilon) &=
\begin{bmatrix}
\vecc{0} & -\vecc{g}(t, \vecc{x}^\epsilon_t) \vecc{B}_2/m_0 & \vecc{\sigma}(t, \vecc{x}^\epsilon_t)\vecc{B}_4/m_0 \\
\mathcal{M}_2^\epsilon \vecc{h}(t, \vecc{x}^\epsilon_t) & -\vecc{\gamma}_{2,2} & \vecc{0} \\
\vecc{0} & \vecc{0} & -\vecc{\gamma}_{4,2}  \\
\end{bmatrix} \\ 
&\ \ \ \ \in  \RR^{(d+d_2/2+d_4/2) \times (d+d_2/2+d_4/2)}, \nonumber \\
\vecc{b}_1(t, \vecc{x}^\epsilon(t),\epsilon) &= (\vecc{0}, -\vecc{\Gamma}_{2,1} \vecc{y}_t^{2,1,\epsilon}, -\vecc{\Gamma}_{4,1} \vecc{\beta}_t^{4,1,\epsilon}) \in \RR^{d+d_2/2+d_4/2},\\
\vecc{b}_2(t,\vecc{x}^\epsilon(t),\epsilon) &= ((-\vecc{g}(t, \vecc{x}^\epsilon_t)\vecc{B}_2\vecc{y}_t^{2,1,\epsilon} + \vecc{\sigma}(t, \vecc{x}^\epsilon_t) \vecc{B}_4 \vecc{\beta}_t^{4,1,\epsilon}+\vecc{F}_e(t,\vecc{x}^\epsilon_t))/m_0, \nonumber \\
&\ \ \ \ \ \  \vecc{0},\vecc{0}) \in \RR^{d+d_2/2+d_4/2},  \\
\vecc{\sigma}_1(t,\vecc{x}^\epsilon(t),\epsilon) &= [\vecc{0} \ \ \vecc{0} \ \  \vecc{\sigma}_1^\epsilon ]^* \in \RR^{(d+d_2/2+d_4/2)\times d_4/2}, \\
\vecc{\sigma}_2(t, \vecc{x}^\epsilon(t),\epsilon) &= [\vecc{0} \ \ \vecc{0} \ \  \vecc{\sigma}_2^\epsilon ]^* \in \RR^{(d+d_2/2+d_4/2)\times d_4/2}.
\end{align}

The initial conditions are $\vecc{x}^\epsilon(0) = (\vecc{x}, \vecc{0}, \vecc{\beta}_0^{4,1,\epsilon})$ and $\vecc{v}^\epsilon(0) = (\vecc{v}, \vecc{0}, \vecc{\beta}_0^{4,2,\epsilon})$;  both depend on $\epsilon$.

We now verify each of the assumptions of Theorem \ref{mainthm}. Assumption \ref{aexis} clearly holds by our assumptions on the GLE. The assumptions on the coefficients in the SDEs   follow easily from the Assumptions \ref{bounded}-\ref{initialdata} and therefore Assumption \ref{a1_ch2} holds.  

Next, note that $\vecc{\beta}_0^{4,\epsilon} = (\vecc{\beta}_0^{4,1,\epsilon}, \vecc{\beta}_0^{4,2,\epsilon})$ is a random variable normally distributed with mean-zero and covariance:
\begin{equation}
\vecc{M}_4^\epsilon = \begin{bmatrix}
\mathbb{E}[|\vecc{\beta}_0^{4,1,\epsilon}|^2] & \mathbb{E}[ \vecc{\beta}_0^{4,1,\epsilon} (\vecc{\beta}_0^{4,2,\epsilon})^*] \\
 \mathbb{E}[\vecc{\beta}_0^{4,2,\epsilon} (\vecc{\beta}_0^{4,1,\epsilon})^*] & \mathbb{E}[|\vecc{\beta}_0^{4,2,\epsilon}|^2]
\end{bmatrix},
\end{equation}
where
\begin{align}
\mathbb{E}[|\vecc{\beta}_0^{4,1,\epsilon}|^2] &= \frac{1}{2}\vecc{\Gamma}_{4,1} \vecc{\gamma}_{4,2}^2 (\epsilon\vecc{\Gamma}_{4,1}-\vecc{\gamma}_{4,2})^{-2}   = O(1),  \\
\mathbb{E}[\vecc{\beta}_0^{4,1,\epsilon} (\vecc{\beta}_0^{4,2,\epsilon})^*] &= \mathbb{E}[ \vecc{\beta}_0^{4,2,\epsilon} (\vecc{\beta}_0^{4,1,\epsilon})^*]   \nonumber \\
&=  -\vecc{\Gamma}_{4,1} \vecc{\gamma}^3_{4,2} (\epsilon \vecc{\Gamma}_{4,1}+\vecc{\gamma}_{4,2})^{-1}  (\epsilon \vecc{\Gamma}_{4,1}-\vecc{\gamma}_{4,2})^{-2}  = O(1), \\
\mathbb{E}[|\vecc{\beta}_0^{4,2,\epsilon}|^2] &= \frac{1}{2\epsilon}\vecc{\gamma}^3_{4,2} (\epsilon \vecc{\Gamma}_{4,1}-\vecc{\gamma}_{4,2})^{-2}  =   O\left(\frac{1}{\epsilon}\right)
\end{align}
as $\epsilon \to 0$. Using the bound $\mathbb{E}[ |\vecc{z}|^p ]\leq C_p (\mathbb{E}[|\vecc{z}|^2])^{p/2}$, where $\vecc{z}$ is a mean-zero Gaussian random variable, $C_p>0$ is a constant and $p>0$, it is straightforward to see that Assumption \ref{a2_ch2} is satisfied.

 Note that $\vecc{B}_i = \vecc{b}_i$ (for $i=1,2$) by our convention, as the $\vecc{b}_i$ do not depend explicitly on $\epsilon$.
The uniform convergence of $\vecc{a}_i(t, \vecc{x},\epsilon)$, $(\vecc{a}_i)_{\vecc{x}}(t, \vecc{x},\epsilon)$ and $\vecc{\sigma}_i(t, \vecc{x},\epsilon)$ (in $\vecc{x}$) to $\vecc{A}_i(t, \vecc{x})$, $(\vecc{A}_i)_{\vecc{x}}(t, \vecc{x})$ and $\vecc{\Sigma}_i(t, \vecc{x})$ respectively in the limit $\epsilon \to 0$ can be shown easily and, in fact,  we see that $\vecc{A}_1 = \vecc{T}$, $\vecc{A}_2 = -\vecc{U}$, where $\vecc{T}$ and $\vecc{U}$ are given in the  theorem,
\begin{align}
\vecc{\Sigma}_1 &= [\vecc{0} \ \ \vecc{0} \ \ -\vecc{\Gamma}_{4,1}  ]^*, \\
\vecc{\Sigma}_2 &= [\vecc{0} \ \ \vecc{0} \ \  \vecc{\gamma}_{4,2} ]^*,
\end{align}
and $a_1 = a_2 = c_1 = c_2 = d_1 = d_2 = 1$, $b_1=b_2 = \infty$, where the $a_i$, $b_i$, $c_i$ and $d_i$ are from Assumption \ref{a5_ch2} of Theorem \ref{mainthm}.  Therefore, the first part of Assumption \ref{a5_ch2} is satisfied.

 It remains to verify the (uniform) Hurwitz stability of $\vecc{a}_2$ and $\vecc{A}_2$ (i.e. Assumption \ref{a0_ch2} and the last part of Assumption \ref{a5_ch2}). This can be done using the methods of the proof of Theorem 2  in  \cite{LimWehr_Homog_NonMarkovian}  and we omit the details here. The results then follow by applying Theorem \ref{mainthm} and \eqref{lim_2}-\eqref{TU} follow from  matrix algebraic calculations.
\end{proof}

It is clear from Theorem \ref{compl} that the homogenized position process is a component of the (slow) Markov process $\vecc{\theta}_t$.  In general, it is not a Markov process itself.  Also, the components of $\vecc{\theta}_t$ are coupled in a non-trivial way.  We emphasize that one could use Theorem \ref{mainthm} to study cases in which the different time scales are taken to zero in a different manner.

The limiting SDE for  the position process may simplify under additional assumptions.  In particular, in the one-dimensional case, i.e. with $d=1$ (or when all the matrix-valued coefficients and the parameters are diagonal in the multi-dimensional case), the formula for the limiting SDEs becomes more explicit. This special case has been studied in an earlier section in the context of the models  (M1) and (M2) from Example \ref{ex_mot}.

\section{Conclusions and Final Remarks}
\label{sect_conclusions}
We have explored various homogenization schemes for a wide class of generalized Langevin equations. The relevance of the studied limit problems in the context of usual and anomalous diffusion of a particle in a heat bath. Our explorations here open up a wide range of possibilities and provide insights in the model reduction of and effective drifts in generalized  Langevin systems.

The following summarizes the main conclusions of the paper:

\begin{itemize}
\item[(i)] (stochastic modeling point of view) Homogenization schemes producing effective SDEs, driven by white noise, should be the exception rather than the rule.  This is particularly important if one seeks to reduce the original model, retain its non-trivial features; 
\item[(ii)] (complexity reduction point of view) There is a trade-off in simplifying  GLE models with state-dependent coefficients: the greater the level of model reduction, the more complicated the correction drift terms, entering the homogenized model;
\item[(iii)] (statistical physics point of view) Homogenized equation obtained could be  further simplified, i.e. number of effective equations could be reduced and the drift terms become simplified, when certain special conditions  such as a fluctuation-dissipation theorem holds. 
\end{itemize}

We conclude this paper by mentioning a very interesting future direction. As mentioned in Remark \ref{rem_inf_dim}, one could extend the current GLE studies to the infinite-dimensional setting so that a larger class of memory functions and covariance functions can be covered.  To this end, one can define the noise process as an appropriate linear functional of a Hilbert space valued process solving a stochastic evolution equation \cite{da2014stochastic,da1996ergodicity}.  This way, one can approach a class of GLEs, driven by noises having a completely monotone covariance function.  
This large class of functions contains covariances with power decay and thus the method outlined above can be viewed as an extension of those considered in \cite{glatt2018generalized, 2018arXiv180409682N}, where the memory function and covariance of the driving noise  are represented as suitable  infinite series with a power-law tail
(these works are,  to our knowledge, among the few works that study rigorously GLEs with a power-law memory).
This approach to systems driven by strongly correlated noise, which is our future project, is expected to involve substantial technical difficulties.  More importantly, one can expect that power decay of correlations leads to new phenomena, altering the nature of noise-induced drift.   

\appendix

\section{Homogenization for a Class of SDEs with State-Dependent Coefficients}
\label{sect_generalhomogthm}



In this section, we study homogenization for a general class of perturbed SDEs with state-dependent coefficients. Homogenization  of differential equations has been extensively studied, from the seminal works of Kurtz \cite{kurtz1973limit}, Papanicolaou \cite{papanicolaou1976some} and Khasminksy \cite{has1966stochastic} to the more recent works \cite{g2005analysis,Pavliotis,hottovy2015smoluchowski,herzog2016small,birrell2017,birrell2017small,chevyrev2016multiscale}.
Here we are going to present yet another variant of homogenization result that will be needed for studying homogenization for our GLEs (see the last paragraph in Section \ref{goaletc} for comments on novelty of this result). 


Let $n_1$, $n_2$, $k_1$, $k_2$ be positive integers. Let $\epsilon \in (0,\epsilon_0] =: \mathcal{E}$ be a small parameter and $\vecc{x}^{\epsilon}(t) \in \RR^{n_1}$, $\vecc{v}^{\epsilon}(t) \in \RR^{n_2}$ for $t \in [0,T]$, where $\epsilon_0>0$ and $T>0$ are finite constants. Let $\vecc{W}^{(k_1)}$ and $\vecc{W}^{(k_2)}$ denote independent Wiener processes, which are $\RR^{k_1}$-valued and $\RR^{k_2}$-valued respectively, on a filtered probability space $(\Omega, \mathcal{F}, \mathcal{F}_t, \mathbb{P})$ satisfying the usual conditions \cite{karatzas2012Brownian}.

With respect to the standard bases of  $\RR^{n_1}$ and $\RR^{n_2}$ respectively, we write:
\begin{align}
\vecc{x}^{\epsilon}(t) &= ([x^{\epsilon}]_1(t),[x^{\epsilon}]_2(t),\dots, [x^{\epsilon}]_{n_1}(t)),  \\
\vecc{v}^{\epsilon}(t) &= ([v^{\epsilon}]_1(t),[v^{\epsilon}]_2(t),\dots, [v^{\epsilon}]_{n_2}(t)).
\end{align}

We consider the following family of perturbed SDE systems\footnote{Note that here the variables $\vecc{x}^\epsilon(t)$ and $\vecc{v}^\epsilon(t)$ are general and they do not necessarily represent position and velocity variables of a physical system.} for 

\noindent $(\vecc{x}^\epsilon(t), \vecc{v}^\epsilon(t)) \in \RR^{n_1+n_2}$:
\begin{align}
d\vecc{x}^{\epsilon}(t) &= \vecc{a}_{1}(t,\vecc{x}^{\epsilon}(t), \epsilon) \vecc{v}^{\epsilon}(t) dt + \vecc{b}_{1}(t,\vecc{x}^{\epsilon}(t),\epsilon) dt + \vecc{\sigma}_{1}(t,\vecc{x}^{\epsilon}(t),\epsilon) d\vecc{W}^{(k_1)}(t), \label{sde1} \\
\epsilon d\vecc{v}^{\epsilon}(t) &= \vecc{a}_{2}(t,\vecc{x}^{\epsilon}(t),\epsilon) \vecc{v}^{\epsilon}(t) dt + \vecc{b}_{2}(t,\vecc{x}^{\epsilon}(t),\epsilon) dt + \vecc{\sigma}_2(t,\vecc{x}^{\epsilon}(t), \epsilon) d\vecc{W}^{(k_2)}(t), \label{sde2}
\end{align}
with the initial conditions, $\vecc{x}^{\epsilon}(0) = \vecc{x}^\epsilon$ and $\vecc{v}^{\epsilon}(0) = \vecc{v}^\epsilon$, where $\vecc{x}^\epsilon$ and $\vecc{v}^\epsilon$ are random variables that possibly depend on $\epsilon$. In the SDEs \eqref{sde1}-\eqref{sde2}, the coefficients $\vecc{a}_1: \RR^+ \times \RR^{n_1} \times \mathcal{E} \to  \RR^{n_1 \times n_2}$, $\vecc{a}_2 : \RR^+ \times \RR^{n_1} \times \mathcal{E} \to   \RR^{n_2 \times n_2}$, $\vecc{\sigma}_2 : \RR^+ \times \RR^{n_1} \times \mathcal{E} \to \RR^{n_2 \times k_2}$ are non-zero matrix-valued functions, whereas $\vecc{b}_1 : \RR^+ \times  \RR^{n_1} \times  \mathcal{E} \to \RR^{n_1}$, $\vecc{b}_2 : \RR^+ \times \RR^{n_1} \times  \mathcal{E} \to \RR^{n_2}$,
$\vecc{\sigma}_1 : \RR^+ \times \RR^{n_1} \times  \mathcal{E} \to \RR^{n_1 \times k_1}$ are matrix-valued or vector-valued functions, which may depend on $\vecc{x}^{\epsilon}$,  as well as on $t$ and $\epsilon$ explicitly, as indicated by the parenthesis $(t, \vecc{x}^{\epsilon}(t), \epsilon)$.  In the case where the coefficients do not depend on  $\epsilon$ explicitly, we will denote them by the corresponding capital letters (for instance, if  $\vecc{a}_i(t,\vecc{x},\epsilon)=\vecc{a}_i(t,\vecc{x})$, then $\vecc{a}_i(t,\vecc{x}) := \vecc{A}_i(t,\vecc{x})$ etc.).


We are interested in the limit as $\epsilon \to 0$ of the SDEs \eqref{sde1}-\eqref{sde2}, in particular the limiting behavior of the  process $\vecc{x}^{\epsilon}(t)$, under  appropriate assumptions\footnote{We forewarn the readers that our assumptions can be relaxed in various directions (see later remarks) but we will not pursue these generalizations here.}  on the coefficients. In this section, we present a homogenization theorem that studies this limit and delay its proof and applications to later sections.\\




We make the following assumptions concerning the SDEs \eqref{sde1}-\eqref{sde2} and \eqref{mainlimitingeqn}.

\begin{ass} \label{aexis} The global solutions, defined on $[0,T]$, to the pre-limit SDEs \eqref{sde1}-\eqref{sde2} and to the limiting SDE \eqref{mainlimitingeqn} a.s. exist and are unique for all $\epsilon \in \mathcal{E} $ (i.e. there are no explosions).
\end{ass}

\begin{ass} \label{a0_ch2} The matrix-valued functions 
$$\{ -\vecc{a}_2(t,\vecc{y}, \epsilon); t \in [0,T], \vecc{y} \in \RR^{n_1}, \epsilon \in \mathcal{E} \}$$ are {\it uniformly positive stable}, i.e. all real parts of the eigenvalues of $-\vecc{a}_2(t, \vecc{y},\epsilon)$ are bounded from below, uniformly in $t$, $\vecc{y}$ and $\epsilon$, by a positive constant (or, equivalently, the matrix-valued functions $\{\vecc{a}_2(t, \vecc{y},\epsilon);  t \in [0,T], \vecc{y} \in \RR^{n_1}, \epsilon \in \mathcal{E} \}$ are {\it uniformly Hurwitz stable}). They are $O(1)$ as $\epsilon \to 0$ (see Assumption \ref{a5_ch2}). 
\end{ass}







\begin{ass} \label{a1_ch2} For $t \in [0,T]$, $\vecc{y} \in \RR^{n_1}$, $\epsilon \in  \mathcal{E}$, and $i=1,2$, the functions $\vecc{b}_i(t,\vecc{y},\epsilon)$ and
$\vecc{\sigma}_i(t,\vecc{y},\epsilon)$ are continuous and bounded in $t$ and $\vecc{y}$, and Lipschitz in $\vecc{y}$, whereas the functions $\vecc{a}_i(t,\vecc{y},\epsilon)$ and $(\vecc{a}_i)_{\vecc{y}}(t,\vecc{y},\epsilon)$ are continuous in $t$, continuously differentiable in $\vecc{y}$, bounded in $t$ and $\vecc{y}$, and Lipschitz in $\vecc{y}$.
Moreover, the  functions $(\vecc{a}_i)_{\vecc{y} \vecc{y}}(t,\vecc{y},\epsilon)$ ($i=1,2$) are bounded for every $t \in [0,T]$, $\vecc{y} \in \RR^{n_1}$ and $\epsilon \in  \mathcal{E}$.


We assume that the (global) Lipschitz constants are bounded by $L(\epsilon)$, where $L(\epsilon)=O(1)$  as $\epsilon \to 0$, i.e. for every $t \in [0,T]$, $\vecc{x}$, $\vecc{y} \in \RR^{n_1}$,
\begin{align}
&\max\bigg\{\|\vecc{a}_i(t, \vecc{x},\epsilon)-\vecc{a}_i(t,\vecc{y},\epsilon)\|,\|(\vecc{a}_i)_{\vecc{x}}(t,\vecc{x},\epsilon)-(\vecc{a}_i)_{\vecc{x}}(t,\vecc{y},\epsilon)\|,  \nonumber \\
&\hspace{1.1cm} |\vecc{b}_i(t,\vecc{x},\epsilon)-\vecc{b}_i(t,\vecc{y},\epsilon)|,
\|\vecc{\sigma}_i(t,\vecc{x},\epsilon)-\vecc{\sigma}_i(t,\vecc{y},\epsilon)\|; \  i=1,2\bigg\} \nonumber \\ 
&\leq L(\epsilon)|\vecc{x}-\vecc{y}|.
\end{align}
\end{ass}


\begin{ass} \label{a2_ch2} The initial condition $\vecc{x}^\epsilon_0 = \vecc{x}^\epsilon \in \RR^{n_1}$ is an $\mathcal{F}_0$-measurable random variable that may depend on $\epsilon$, and we assume that $\mathbb{E}[|\vecc{x}^\epsilon|^p] = O(1)$ as $\epsilon \to 0$ for all $p>0$. Also, $\vecc{x}^\epsilon$ converges, in the limit as $\epsilon \to 0$, to a random variable $\vecc{x}$ as follows:
$\mathbb{E}\left[|\vecc{x}^\epsilon - \vecc{x}|^p \right]  = O(\epsilon^{p r_0})$, where $r_0 > 1/2$ is a constant, as $\epsilon \to 0$.       The initial condition $\vecc{v}^\epsilon_0 = \vecc{v}^\epsilon \in \RR^{n_2}$ is an $\mathcal{F}_0$-measurable random variable that may depend on $\epsilon$, and we assume that for every $p>0$, $\mathbb{E}[ |\epsilon \vecc{v}^\epsilon|^p] = O(\epsilon^\alpha)$ as $\epsilon \to 0$, for some $\alpha \geq p/2$.



\end{ass}

\begin{ass} \label{a5_ch2} For $i=1,2$, $t \in [0,T]$, and every $\vecc{x} \in \RR^{n_1}$, each of the matrix or vector entries of the (non-zero) functions $\vecc{a}_i(t,\vecc{x},\epsilon)$, $(\vecc{a}_i)_{\vecc{x}}(t,\vecc{x},\epsilon)$, $\vecc{b}_i(t,\vecc{x},\epsilon)$ and
$\vecc{\sigma}_i(t,\vecc{x},\epsilon)$, converges, uniformly in $\vecc{x}$, to a unique non-zero element, in the limit as $\epsilon \to 0$. Their limits are denoted by
$\vecc{A}_i(t,\vecc{x})$,
$(\vecc{A}_i)_{\vecc{x}}(t,\vecc{x})$,
$\vecc{B}_i(t,\vecc{x})$ and
$\vecc{\Sigma}_i(t,\vecc{x})$ respectively.
Their rate of convergence is assumed to satisfy the following power law bounds: for every $t \in [0,T]$, $\vecc{x} \in \RR^{n_1}$ and $i=1,2$,
\begin{align}
\|\vecc{a}_i(t,\vecc{x},\epsilon)-\vecc{A}_i(t,\vecc{x}) \| &\leq \alpha_i(\epsilon), \\
|\vecc{b}_i(t,\vecc{x},\epsilon)-\vecc{B}_i(t,\vecc{x}) | &\leq \beta_i(\epsilon), \\ \|\vecc{\sigma}_i(t,\vecc{x},\epsilon)-\vecc{\Sigma}_i(t,\vecc{x}) \| &\leq \gamma_i(\epsilon),\\
\| (\vecc{a}_i)_{\vecc{x}}(t,\vecc{x},\epsilon)-(\vecc{A}_i )_{\vecc{x}}(t,\vecc{x}) \| &\leq \theta_i(\epsilon)
\end{align}
where $\alpha_i(\epsilon) =  O(\epsilon^{a_i})$, $\beta_i(\epsilon) = O(\epsilon^{b_i})$, $\gamma_i(\epsilon) = O(\epsilon^{c_i})$ and $\theta_i(\epsilon) =  O(\epsilon^{d_i})$, as $\epsilon \to 0$, for some positive exponents $a_i$, $b_i$, $c_i$ and $d_i$. Moreover, we assume that $\vecc{A}_2(t,\vecc{x})$ is Hurwitz stable for every $t$ and $\vecc{x}$. \\

\noindent {\bf Convention.}  In the case where the coefficients do not show explicit dependence on $\epsilon$ or the case when any of the coefficients $\vecc{b}_1$, $\vecc{b}_2$ and $\vecc{\sigma}_1$ is zero, we set the exponent, describing the  corresponding rate of convergence, to infinity. For instance, if $\vecc{a}_i(t,\vecc{x},\epsilon) = \vecc{A}_i(t,\vecc{x})$, we set $a_i  = \infty$. Meanwhile, if $\vecc{\sigma}_1 = \vecc{0}$, we set $c_1 = \infty$, etc..
\end{ass}




We now state our homogenization theorem.

\begin{thm} \label{mainthm}
Suppose that the family of SDE systems $\eqref{sde1}$-$\eqref{sde2}$ satisfies Assumption \ref{aexis}-\ref{a5_ch2}. Let $(\vecc{x}^{\epsilon}(t), \vecc{v}^{\epsilon}(t)) \in \RR^{n_1} \times \RR^{n_2}$ be their solutions, with the initial conditions $(\vecc{x}^\epsilon, \vecc{v}^\epsilon)$. Let $\vecc{X}(t) \in \RR^{n_1}$ be the solution to the following It\^o SDE with the initial position $\vecc{X}(0) = \vecc{x}$:
\begin{align}
d\vecc{X}(t) &= [\vecc{B}_1(t,\vecc{X}(t))-\vecc{A}_1(t,\vecc{X}(t))\vecc{A}_2^{-1}(t,\vecc{X}(t))\vecc{B}_2(t,\vecc{X}(t))] dt \nonumber \\
&\ \ \ \  + \vecc{S}(t,\vecc{X}(t)) dt  + \vecc{\Sigma}_1(t,\vecc{X}(t)) d\vecc{W}^{(k_1)}(t) \nonumber \\
&\ \ \ \  - \vecc{A}_1(t,\vecc{X}(t)) \vecc{A}_2^{-1}(t,\vecc{X}(t))\vecc{\Sigma}_2(t,\vecc{X}(t)) d\vecc{W}^{(k_2)}(t), \label{mainlimitingeqn}
\end{align}
where  $\vecc{S}(t,\vecc{X}(t))$ is the {\it noise-induced drift vector} whose $i$th component is given by
\begin{equation}
[S]_{i}(t,\vecc{X}) = -\frac{\partial}{\partial X_{l}} \bigg([A_1 A_2^{-1}]_{i,j}(t,\vecc{X}) \bigg) \cdot [A_1]_{l,k}(t,\vecc{X}) \cdot [J]_{j,k}(t,\vecc{X}), 
\end{equation}
where $i,l=1,\dots,n_1, \  j,k=1,\dots,n_2$, or in index-free notation, 
\begin{equation} \label{indexfree}
\vecc{S} = \vecc{A}_1 \vecc{A}_2^{-1} \vecc{\nabla}\cdot (\vecc{J}\vecc{A}_1^*) -\vecc{\nabla} \cdot (\vecc{A}_1 \vecc{A}_2^{-1} \vecc{J} \vecc{A}_1^*) , 
\end{equation}
and $\vecc{J} \in \RR^{n_2 \times n_2}$ is the unique solution to the Lyapunov equation:
\begin{equation} \label{lyp}
\vecc{J} \vecc{A}_2^{*} + \vecc{A}_2 \vecc{J} = -\vecc{\Sigma}_2  \vecc{\Sigma}_2^{*}.
\end{equation}
Then the process $\vecc{x}^{\epsilon}(t)$ converges, as $\epsilon \to 0$, to the solution $\vecc{X}(t)$, of the It\^o SDE \eqref{mainlimitingeqn}, in the following sense: for all finite $T > 0$, $p > 0$, there exists a positive random variable $\epsilon_1$ such that
\begin{equation} \label{mainconv}
\mathbb{E}\left[\sup_{t \in [0,T]} |\vecc{x}^{\epsilon}(t) - \vecc{X}(t)|^p; \epsilon \leq \epsilon_1 \right] = O(\epsilon^{r}), \end{equation}
in the limit as $\epsilon \to 0$,
with $r>0$ is the rate determined to be:
\begin{equation} \label{rate_mainresult}
    r=
    \begin{cases}
      \beta \ \text{ for all } 0 < \beta < \frac{p}{2}, & \text{ if}\  a_i, b_i, c_i, d_i \geq \frac{1}{2} \text{ for } i=1,2, \\
      p \cdot \min(a_i, b_i, c_i, d_i; i=1,2) , & \text{ otherwise},
    \end{cases}
  \end{equation}
where the $a_i$, $b_i$, $c_i$, $d_i$ ($i=1,2$) are the positive constants from Assumption \ref{a5_ch2}. In particular, for all finite $T>0$,
\begin{equation}
 \sup_{t \in [0,T]} |\vecc{x}^\epsilon(t) - \vecc{X}(t)| \to 0, \end{equation}
in probability, in the limit as $\epsilon \to 0$.
\end{thm}

\begin{rem} \label{warn}
With more work and additional assumptions, one could prove the statements in Assumption \ref{aexis} from Assumption \ref{a0_ch2}-\ref{a5_ch2}. However, we choose to incorporate such existence and uniquess results into our assumptions and work with the assumptions as stated above. Moreover, as we have forewarned the readers, our assumptions can be relaxed in various directions at the cost of more technicalities. For instance, the boundedness assumption on the coefficients of the SDEs may be removed to obtain still a pathwise convergence result by adapting the techniques in \cite{herzog2016small} -- see also analogous remarks in Remark 5 in \cite{LimWehr_Homog_NonMarkovian}. However, we choose not to pursue the above technical details in this already lengthy paper.  
\end{rem}

\section{Proof of Theorem \ref{mainthm}} \label{proof_ch2}

Proof of Theorem \ref{mainthm} uses techniques developed in earlier works \cite{hottovy2015smoluchowski, 2017BirrellLatest,LimWehr_Homog_NonMarkovian}, but here  one needs to additionally take into account the $\epsilon$-dependence of the coefficients in the SDEs \eqref{sde1}-\eqref{sde2}. As a preparation for the proof, we  need a few lemmas and propositions.

We start from an elementary calculus result.

\begin{lem} \label{lipzlemma}
For $i=1,\dots,N$, let $\vecc{f}_i(\vecc{y},\epsilon): \RR^n \times (0,\infty) \to \RR^{m_i \times n}$ be bounded and globally Lipschitz in $\vecc{y}$  for every $\epsilon > 0$, with a Lipschitz constant that is bounded as $\epsilon \to 0$, i.e. for every $\vecc{y}, \vecc{z} \in \RR^{n}$, there exists a constant $M_i(\epsilon)>0$ such that \begin{equation}
\|\vecc{f}_i(\vecc{y},\epsilon)-\vecc{f}_i(\vecc{z},\epsilon)\| \leq M_i(\epsilon)|\vecc{y}-\vecc{z}|,
\end{equation}
where $M_i(\epsilon)=O(1)$ as $\epsilon \to 0$.
\begin{itemize}
\item[(i)] Suppose that for each $i$ and $\vecc{y} \in \RR^{n}$, there exists a unique bounded $\vecc{F}_i(\vecc{y}):\RR^n \to \RR^{m_i \times n}$ and a constant $C_i>0$ such that  $\|\vecc{f}_i(\vecc{y},\epsilon) - \vecc{F}_i(\vecc{y})\| \leq C_i \epsilon^{r_i}$, for some positive constant $r_i$, as $\epsilon \to 0$ (i.e. the left-hand side is of order $O(\epsilon^{r_i})$ as $\epsilon \to 0$). Then there exists  constants $D$, $K_1, \dots, K_N >0$, such that
\begin{align} \label{boundlip}
\bigg\|\prod_{i=1}^{N} \vecc{f}_i(\vecc{y},\epsilon)-\prod_{i=1}^N \vecc{F}_i(\vecc{y})\bigg\| &\leq K_1 \epsilon^{r_1} + \dots + K_N \epsilon^{r_N} \leq D \epsilon^{\min(r_1, \dots, r_N)} \\ 
&= O(\epsilon^{\min(r_1, \dots, r_N)}),
\end{align}
as $\epsilon \to 0$. If, in addition, $n=m_1$, $\vecc{f}_1(\vecc{y},\epsilon)$ and $\vecc{F}_1(\vecc{y})$ are invertible for every $\vecc{y} \in \RR^n$ and $\epsilon > 0$, then $\|\vecc{f}_1^{-1}(\vecc{y},\epsilon)-\vecc{F}_1^{-1}(\vecc{y})\| = O(\epsilon^{r_1})$ as $\epsilon \to 0$.
\item[(ii)] Let $c_i \in \RR$, $i=1,\dots,N$. For every $\epsilon > 0$ and $\vecc{y} \in \RR^n$, $\sum_{i=1}^{N} c_i \vecc{f}_i(\vecc{y},\epsilon)$ and $\prod_{i=1}^N c_i \vecc{f}_i(\vecc{y},\epsilon)$ are globally Lipschitz with a Lipschitz constant that is $O(1)$  as $\epsilon \to 0$. Moreover, if $m_1=n$ and for every $\epsilon>0$, $\vecc{y} \in \RR^n$,  $\vecc{f}_1(\vecc{y},\epsilon)$ is invertible, then for every $\epsilon > 0$, $\vecc{y} \in \RR^n$, $\vecc{f}^{-1}_1(\vecc{y},\epsilon)$ is globally Lipschitz in $\vecc{y}$ with a Lipschitz constant that is $O(1)$ as $\epsilon \to 0$.
\end{itemize}
\end{lem}
\begin{proof}
\begin{itemize}
\item[(i)] We prove this inductively. The base case of $N=1$ clearly holds with $D = C_1$. Let $k \in \{1,\dots,N-1\}$. Assume that \eqref{boundlip} holds with $N:=k$ and $D := D_k$. Then
\begin{align}
&\bigg\|\prod_{i=1}^{k+1} \vecc{f}_i(\vecc{y},\epsilon)-\prod_{i=1}^{k+1} \vecc{F}_i(\vecc{y})\bigg\|  \nonumber \\ 
&= \bigg\|\vecc{f}_{k+1}(\vecc{y},\epsilon)\cdot\prod_{i=1}^{k} \vecc{f}_i(\vecc{y},\epsilon)-\vecc{F}_{k+1}(\vecc{y})\cdot \prod_{i=1}^{k} \vecc{F}_i(\vecc{y})\bigg\|  \\
&\leq \|\vecc{f}_{k+1}(\vecc{y},\epsilon)\| \cdot \left\| \prod_{i=1}^{k} \vecc{f}_i(\vecc{y},\epsilon)-\prod_{i=1}^{k} \vecc{F}_i(\vecc{y})\right\| \nonumber \\
&\ \ \ \ \  + \|\vecc{f}_{k+1}(\vecc{y},\epsilon)- \vecc{F}_{k+1}(\vecc{y})\| \cdot \left\| \prod_{i=1}^{k} \vecc{F}_i(\vecc{y})\right\|   \\
&\leq C ( D_k \epsilon^{\min(r_1,\dots,r_k)} + C_{k+1}\epsilon^{r_{k+1}} )  \\
&\leq C \max\{D_k, C_{k+1} \}(\epsilon^{\min(r_1,\dots,r_k)}  + \epsilon^{r_{k+1}}) \leq D_{k+1} \epsilon^{\min(r_1,\dots,r_{k+1})},\end{align}
as $\epsilon \to 0$, where $C$, $D_{k+1}$ are  positive constants and we have used  the inductive hypothesis and assumptions of the lemma in the last two lines above. The last statement follows from:
\begin{align}
\|\vecc{f}_1^{-1}(\vecc{y},\epsilon)-\vecc{F}_1^{-1}(\vecc{y})\| &= \|\vecc{f}_1^{-1}(\vecc{y},\epsilon) (\vecc{F}_1(\vecc{y})-\vecc{f}_1(\vecc{y},\epsilon)) \vecc{F}_1^{-1}(\vecc{y}) \|  \\
&\leq \|\vecc{f}_1^{-1}(\vecc{y},\epsilon)\|\cdot \| \vecc{F}_1(\vecc{y})-\vecc{f}_1(\vecc{y},\epsilon) \| \cdot \|\vecc{F}_1^{-1}(\vecc{y}) \|   \\
&\leq C \epsilon^{r_1},
\end{align} as $\epsilon \to 0$, where $C$ is a positive constant.
\item[(ii)] The statements can be proven using the same techniques used for (i) and so we omit the proof.
\end{itemize}
\end{proof}

Let $\vecc{x}^\epsilon(t) \in \RR^{n_1}$, $\vecc{v}^\epsilon(t) \in \RR^{n_2}$ and $T>0$. For $t \in [0,T]$, let $\vecc{p}^{\epsilon}(t) := \epsilon \vecc{v}^{\epsilon}(t)$ denote a solution of the SDE:
\begin{align}
d\vecc{p}^{\epsilon}(t) &= \frac{\vecc{a}_{2}(t, \vecc{x}^{\epsilon}(t),\epsilon)}{\epsilon} \vecc{p}^{\epsilon}(t) dt + \vecc{b}_{2}(t,\vecc{x}^{\epsilon}(t),\epsilon) dt + \vecc{\sigma}_2(t,\vecc{x}^{\epsilon}(t),\epsilon) d\vecc{W}^{(k_2)}(t). \label{sdeforp}
\end{align}

We provide estimates for the moments concerning the process $\vecc{p}^\epsilon(t)$, under appropriate assumptions on the coefficients and the initial conditions, in the limit as $\epsilon \to 0$. 


We need the following lemma, adapted from Proposition A.2.3 of \cite{kabanov2013two}, to obtain an exponential bound on certain fundamental matrix solution.

\begin{lem} \label{expb} Fix a filtered probability space $(\Omega, \mathcal{F}, \mathcal{F}_t, \mathbb{P})$. For each  $\epsilon > 0$, let $\vecc{B}^\epsilon: [0,T] \times \Omega \to \RR^{n \times n}$ be a bounded (uniformly in $\epsilon$, $\omega \in \Omega$ and $t \in [0,T]$), pathwise continuous  process. Assume that the real parts of all eigenvalues of $\vecc{B}$ are bounded from above by $-2\kappa$, uniformly in $\epsilon$, $\omega \in \Omega$ and $t \in [0,T]$, where $\kappa$ is a positive constant. Let $\vecc{\Phi}^\epsilon(t,s,\omega)$ be the fundamental matrix that solves the initial value problem (IVP):
\begin{equation} \label{ivp}
\frac{\partial \vecc{\Phi}^\epsilon(t,s,\omega)}{\partial t} = \frac{\vecc{B}^\epsilon(t,\omega)}{\epsilon} \vecc{\Phi}^\epsilon(t,s,\omega), \ \ \vecc{\Phi}^\epsilon(s,s,\omega)=\vecc{I}, \ \ 0 \leq s \leq t \leq T.
\end{equation}
Then there exists a constant $C > 0$ and an (in general random\footnote{See also Remark 14 in \cite{LimWehr_Homog_NonMarkovian}.}) $\epsilon_{1}=\epsilon_1(\omega)$ such that
\begin{equation} \label{888}
\|\vecc{\Phi}^\epsilon(t,s,\omega)\| \leq C e^{-\kappa(t-s)/\epsilon} \end{equation} for all $\epsilon \leq \epsilon_{1}$ and for all $s,t \in [0,T]$.
\end{lem}
\begin{proof}
Let $u \in [s,t]$. We rewrite for $\omega \in \Omega$, $s, t \in [0,T]$:
\begin{equation}
\frac{\partial \vecc{\Phi}^\epsilon(t,s,\omega)}{\partial t} =  \frac{\vecc{B}^\epsilon(u,\omega)}{\epsilon} \vecc{\Phi}^\epsilon(t,s,\omega) + \frac{\vecc{B}^\epsilon(t,\omega)-\vecc{B}^\epsilon(u,\omega)}{\epsilon} \vecc{\Phi}^\epsilon(t,s,\omega),
\end{equation}
and represent the solution to the IVP as:
\begin{equation}
\vecc{\Phi}^\epsilon(t,s,\omega) = e^{(t-s)\frac{\vecc{B}^\epsilon(u,\omega)}{\epsilon}} + \frac{1}{\epsilon} \int_s^t e^{(t-r) \frac{\vecc{B}^\epsilon(u,\omega)}{\epsilon}}  (\vecc{B}^\epsilon(r,\omega)-\vecc{B}^\epsilon(u,\omega)) \vecc{\Phi}^\epsilon(r,s,\omega) dr.
\end{equation}

Denote $\vecc{W}^\epsilon(t,s,\omega) := e^{\kappa(t-s)/\epsilon} \vecc{\Phi}^\epsilon(t,s,\omega)$.
Setting $u = t$ in the above representation and multiplying both sides by $e^{\kappa(t-s)/\epsilon}$, we obtain:
\begin{align}
&\vecc{W}^\epsilon(t,s,\omega)  \nonumber \\
&= e^{\kappa(t-s)/\epsilon} e^{(t-s)\vecc{B}^\epsilon(t,\omega)/\epsilon} + \frac{1}{\epsilon} \int_s^t e^{\kappa(t-s)/\epsilon} e^{(t-r) \vecc{B}^\epsilon(t,\omega)/\epsilon}  (\vecc{B}^\epsilon(r,\omega)-\vecc{B}^\epsilon(t,\omega))   \nonumber \\ &\ \hspace{5cm} \cdot  \vecc{\Phi}^\epsilon(r,s,\omega) dr \\
&= e^{\kappa(t-s)/\epsilon} e^{(t-s)\vecc{B}^\epsilon(t,\omega)/\epsilon} + \frac{1}{\epsilon} \int_s^t e^{\kappa(t-s)/\epsilon} e^{(t-r) \vecc{B}^\epsilon(t,\omega)/\epsilon} e^{-\kappa(r-s)/\epsilon} \nonumber \\
&\ \hspace{5cm} \cdot  (\vecc{B}^\epsilon(r,\omega)-\vecc{B}^\epsilon(t,\omega)) \vecc{W}^\epsilon(r,s,\omega) dr.
\end{align}
Since $\vecc{B}^\epsilon$ is bounded (uniformly in $\omega$, $t$ and $\epsilon$), by assumption on the spectrum of $\vecc{B}^\epsilon$, there exists a constant $C > 0$, such that for all $s,t \in [0,T]$ we have
\begin{equation}\|e^{s \vecc{B}^\epsilon(t,\omega)/\epsilon}\| \leq C e^{-2\kappa s/\epsilon}
\end{equation}

Using this, we obtain:
\begin{align}
&\| \vecc{W}^\epsilon(t,s,\omega)\| \nonumber \\
&\leq C e^{-\kappa(t-s)/\epsilon}  \nonumber \\ 
&\ \hspace{0.02cm} + \frac{C}{\epsilon} \int_s^t e^{-2 \kappa(t-r)/\epsilon} e^{-\kappa(r-s)/\epsilon} e^{\kappa(t-s)/\epsilon} \| \vecc{W}^\epsilon(r,s,\omega)\| \cdot \|\vecc{B}^\epsilon(r,\omega)-\vecc{B}^\epsilon(t,\omega)\| dr.
\end{align}
This leads to the estimate:
\begin{align}
&\sup_{s,t \in [0,T]} \| \vecc{W}^\epsilon(t,s,\omega)\| \leq C +   \sup_{r,s \in [0,T]}\| \vecc{W}^\epsilon(r,s,\omega)\| \cdot A_\epsilon(\omega),
\end{align}
where 
\begin{equation}
A_\epsilon(\omega) = \frac{C}{\epsilon} \sup_{t \in [0,T]} \int_0^t e^{- \frac{\kappa(t-r)}{\epsilon}}  \left\|\vecc{B}^\epsilon(r,\omega)-\vecc{B}^\epsilon(t,\omega)\right\| dr.
\end{equation}
For a fixed $\omega \in \Omega$, $A_\epsilon(\omega)$ can be made arbitrary small as $\epsilon \to 0$. 
Therefore, there exists an $\epsilon_1 = \epsilon_1(\omega) > 0$ (generally dependent on $\omega$) such that
\begin{equation}
\sup_{s,t \in [0,T]} \| \vecc{W}^\epsilon(t,s,\omega)\| \leq C + \frac{1}{2} \sup_{s,t \in [0,T]} \| \vecc{W}^\epsilon(t,s,\omega)\|
\end{equation}
for all $\epsilon \leq \epsilon_1$. This implies that $\sup_{s,t \in [0,T]} \| \vecc{W}^\epsilon(t,s,\omega)\| \leq 2C$,  which is the claimed bound.
\end{proof}


We now prove a lemma that gives a bound on a class of stochastic integrals.  It is modification of Lemma 5.1 in \cite{birrell2017small}. In both cases, the main idea is to rewrite some of the stochastic integrals in terms of ordinary ones.

\begin{lem} \label{sib} Let $\vecc{H}_{t} := \vecc{H}_{0} + \vecc{M}_{t} + \vecc{A}_{t}$ be the Doob-Meyer decomposition of a continuous $\RR^{k}$-valued semimartingale on $(\Omega, \mathcal{F}, \mathcal{F}_{t}, P)$ with a local martingale $\vecc{M}_{t}$ and a process of locally bounded variation $\vecc{A}_{t}$. Let $\vecc{V} \in L_{loc}^{1}(A) \cap L_{loc}^2(M)$ be $\RR^{n \times k}$-valued and let $\vecc{B}^\epsilon(t)$ be an adapted process whose values are $n \times n$ matrices, satisfying the assumptions of Lemma \ref{expb}. Let $\vecc{\Phi}^\epsilon(t) := \vecc{\Phi}^\epsilon(t,0)$ be the adapted $C^{1}$ process that pathwise solves the IVP \eqref{ivp}. Then for every $T \geq \delta > 0$ and for every $\epsilon \leq \epsilon_{1}$, we have the $\mathbb{P}$-a.s bound:
\begin{align}
&\sup_{t \in [0,T]} \left|\vecc{\Phi}^\epsilon(t) \int_{0}^{t} (\vecc{\Phi}^\epsilon)^{-1}(s) \vecc{V}_{s} d\vecc{H}_{s} \right| \nonumber \\  
&\leq C\left(1+\frac{4}{\kappa} \sup_{s \in [0,T]} \|\vecc{B}^\epsilon(s)\| \right) \bigg(e^{-\kappa \delta/\epsilon} \sup_{t \in [0,T]} \left| \int_{0}^{t} \vecc{V}_{r} d\vecc{H}_{r} \right|  \nonumber \\
&\ \ \ \ \  \ \ + \max_{k=0,1,\dots,N-1} \sup_{t \in [k\delta, (k+2)\delta]} \left| \int_{k\delta}^{t} \vecc{V}_{r} d\vecc{H}_{r} \right| \bigg) \label{8888},
\end{align}
where $N = \max\{k \in \ZZ: k \delta < T\}$, $\epsilon_{1}$, $\kappa$ and $C$ are  from Lemma \ref{expb}, and $l_{2}$-norm is used on every $\RR^{k}$.
\end{lem}
\begin{proof}
The proof is  identical to that of Lemma 5.1 in \cite{birrell2017small} up to line (5.10), with the constant $\alpha$ there replaced by $\kappa$, etc. We let $\epsilon \leq \epsilon_{1}$ and replace the bound in line (5.11) there by the following bound, which follows from the semigroup property of the fundamental matrix process and Lemma \ref{expb}: \begin{equation}\|\vecc{\Phi}^\epsilon(t) (\vecc{\Phi}^\epsilon)^{-1}(s)\|= \|\vecc{\Phi}^\epsilon(t,0) \vecc{\Phi}^\epsilon(0,s)\| = \|\vecc{\Phi}^\epsilon(t,s)\| \leq C e^{-\kappa(t-s)/\epsilon}.\end{equation} Then we proceed as in the proof of Lemma 5.1 in \cite{birrell2017small} to get the desired bound.
\end{proof}


In particular, \eqref{888} and \eqref{8888} hold for $\vecc{B}^\epsilon = \vecc{a}_2(t, \vecc{x}^\epsilon(t), \epsilon)$.

\begin{prop} \label{mom_bound}
Suppose that Assumptions \ref{aexis}-\ref{a5_ch2} hold. For all $p \geq 1$, $T>0$, $0<\beta<p/2$, there exists a positive random variable $\epsilon_1$ such that:
\begin{equation}
\mathbb{E}\left[\sup_{t \in [0,T]}|\vecc{p}^\epsilon(t)|^p; \epsilon \leq \epsilon_1 \right] = O(\epsilon^\beta),
\end{equation}
as $\epsilon \to 0$, where $\vecc{p}^\epsilon(t)$ solves  the SDE \eqref{sdeforp}.  Therefore, for any $p \geq 1$, $T>0$, $\beta > 0$, we have
\begin{equation}
\mathbb{E}\left[ \sup_{t \in [0,T]} \|\epsilon \vecc{v}^\epsilon(t)\vecc{v}^\epsilon(t)^*\|_{F}^p; \epsilon \leq \epsilon_1 \right] = O(\epsilon^{-\beta}),
\end{equation}
as $\epsilon \to 0$, where $\|\cdot\|_F$ denotes the Frobenius norm.
\end{prop}

\begin{proof}
Let $\vecc{\Phi}_\epsilon(t)$ be the matrix-valued process solving the IVP: 
\begin{equation}
\frac{\partial \vecc{\Phi}_\epsilon(t)}{\partial t} = \frac{\vecc{a}_2(t, \vecc{x}^\epsilon(t), \epsilon)}{\epsilon} \vecc{\Phi}_\epsilon(t), \ \ \vecc{\Phi}_\epsilon(0) = \vecc{I}.
\end{equation}

Then,
\begin{align}
\vecc{p}^\epsilon(t) &= \vecc{\Phi}_\epsilon(t)\epsilon\vecc{v}^\epsilon + \vecc{\Phi}_\epsilon(t) \int_0^t \vecc{\Phi}^{-1}_\epsilon(s) \vecc{b}_2(s,\vecc{x}^\epsilon(s),\epsilon)ds \nonumber \\  
&\ \ \ \ +  \vecc{\Phi}_\epsilon(t) \int_0^t \vecc{\Phi}^{-1}_\epsilon(s) \vecc{\sigma}_2(s,\vecc{x}^\epsilon(s),\epsilon)d\vecc{W}^{(k_2)}(s) \\\
&= \vecc{\Phi}_\epsilon(t)\epsilon\vecc{v}^\epsilon + \vecc{\Phi}_\epsilon(t) \int_0^t \vecc{\Phi}^{-1}_\epsilon(s) \vecc{B}_2(s,\vecc{x}^\epsilon(s))ds \nonumber \\
&\ \ \ \ +\vecc{\Phi}_\epsilon(t) \int_0^t \vecc{\Phi}^{-1}_\epsilon(s) \left[ \vecc{b}_2(s,\vecc{x}^\epsilon(s),\epsilon) -\vecc{B}_2(s,\vecc{x}^\epsilon(s)) \right]ds \nonumber \\
&\ \ \ \  +   \vecc{\Phi}_\epsilon(t) \int_0^t \vecc{\Phi}^{-1}_\epsilon(s)  \vecc{\sigma}_2(s,\vecc{x}^\epsilon(s),\epsilon) d\vecc{W}^{(k_2)}(s).
\end{align}

Therefore, for $T>0$ and $p \geq 1$, using the bound
\begin{equation} \label{algebra} 
\left|\sum_{i=1}^N a_i \right|^p \leq N^{p-1} \sum_{i=1}^N |a_i|^p
\end{equation}
for $p \geq 1$ (here the $a_i \in \RR$ and $N$ is a positive integer), taking supremum on both sides, and applying Lemma \ref{expb} (with $\vecc{B}^\epsilon = \vecc{a}_2(t, \vecc{x}^\epsilon(t), \epsilon)$), we estimate:
\begin{align}
&\sup_{t \in [0,T]}|\vecc{p}^\epsilon(t)|^p \nonumber \\
&\leq 4^{p-1} \sup_{t \in [0,T]} \bigg[C^p e^{-\frac{\kappa p}{\epsilon}t} \epsilon^p |\vecc{v}^\epsilon|^p + C^p \left( \int_0^t e^{-\frac{\kappa}{\epsilon}(t-s)} |\vecc{B}_2(s,\vecc{x}^\epsilon(s))| ds \right)^p \nonumber \\
&\ \ \ \ \ \ +  C^p \left( \int_0^t e^{-\frac{\kappa}{\epsilon}(t-s)} \bigg|[\vecc{b}_2(s,\vecc{x}^\epsilon(s),\epsilon) - \vecc{B}_2(s,\vecc{x}^\epsilon(s))] \bigg| ds \right)^p \nonumber \\
&\ \ \ \ \  \ + \bigg| \vecc{\Phi}_\epsilon(t) \int_0^t \vecc{\Phi}_\epsilon^{-1}(s)  \vecc{\sigma}_2(s,\vecc{x}^\epsilon(s),\epsilon)    d\vecc{W}^{(k_2)}(s) \bigg|^p  \bigg] \\
&\leq 4^{p-1} \bigg(C^p \epsilon^p |\vecc{v}^\epsilon|^p + \frac{C^p \epsilon^p}{\kappa^p} \bigg(\sup_{s \in [0,T]}|\vecc{B}_2 (s,\vecc{x}^\epsilon(s))|^p \nonumber \\
&\ \ \ \ \ \ \ + \sup_{s \in [0,T]}|\vecc{b}_2(s,\vecc{x}^\epsilon(s),\epsilon)-\vecc{B}_2(s,\vecc{x}^\epsilon(s))|^p\bigg) \nonumber \\
&\ \ \ \ \ \ \ + \sup_{t \in [0,T]} \bigg| \vecc{\Phi}_\epsilon(t) \int_0^t \vecc{\Phi}_\epsilon^{-1}(s) \vecc{\sigma}_2(s,\vecc{x}^\epsilon(s),\epsilon)  d\vecc{W}^{(k_2)}(s) \bigg|^p  \bigg), \label{lastline!}
\end{align}
for $\epsilon \leq \epsilon_1$, where $C>0$, $\kappa >0$, and $\epsilon_1>0$ is the random variable whose existence was proven in Lemma \ref{expb}.

Note that $\sup_{s \in [0,T]}|\vecc{B}_2 (s,\vecc{x}^\epsilon(s))|^p   < \infty$   and Assumption \ref{a5_ch2} implies that
\begin{equation}
\sup_{s \in [0,T]}|\vecc{b}_2(s,\vecc{x}^\epsilon(s),\epsilon)-\vecc{B}_2(s,\vecc{x}^\epsilon(s))|^p \leq |\beta_2(\epsilon)|^p,
\end{equation}
where $\beta_2(\epsilon) \leq K \epsilon^{b_2}$. 

Denote $\mathbb{E}_1[\cdot] = \mathbb{E}[\cdot; \epsilon \leq \epsilon_1]$, i.e. the expectation is taken on $\{ \omega : \epsilon  \leq \epsilon_1(\omega)\}$. We are going to estimate $\mathbb{E}_1\left[\sup_{t \in [0,T]} |\vecc{p}^\epsilon(t)|^p \right]$. 

By Assumption \ref{a2_ch2}, we have $\mathbb{E}_1 [\sup_{t \in [0,T]} |\epsilon\vecc{v}^\epsilon|^p] = O(\epsilon^{\alpha})$ as $\epsilon \to 0$, for some $\alpha \geq p/2$. Therefore, combining the above estimates, we obtain:
\begin{align}
&\mathbb{E}_1 \left[ \sup_{t \in [0,T]} |\vecc{p}^\epsilon(t)|^p\right]  \nonumber \\
&\leq C_1(p)(\epsilon^\alpha + \epsilon^{b_2 p} + \epsilon^p) \nonumber  \\ 
&\ \ \  + C_2(p) \mathbb{E}_1 \left[ \sup_{t \in [0,T]} \bigg| \vecc{\Phi}_\epsilon(t) \int_0^t \vecc{\Phi}_\epsilon^{-1}(s) \vecc{\sigma}_2(s,\vecc{x}^\epsilon(s),\epsilon)  d\vecc{W}^{(k_2)}(s) \bigg|^p \right],
\end{align}
where $C_1(p), C_2(p) > 0$ are constants.  

Next, the idea is to use Lemma \ref{sib}  and the Burkholder-Davis-Gundy inequality (see Theorem 3.28 in \cite{karatzas2012Brownian}) to estimate the last term on the right hand side above. This is analogous to the technique used in the proof of Proposition 5.1 in \cite{birrell2017small}.

Let $\delta$ be a constant such that $0<\delta < T$. Applying Lemma \ref{sib}, we estimate, using \eqref{algebra}:

\begin{align}
&\mathbb{E}_1\left[ \sup_{t \in [0,T]} \left| \vecc{\Phi}_{\epsilon}(t) \int_0^t \vecc{\Phi}_{\epsilon}^{-1}(s) \vecc{\sigma}_2(s,\vecc{x}^\epsilon(s),\epsilon) d\vecc{W}^{(k_2)}(s) \right|^p \right] \nonumber \\  
&\leq 2^{p-1} C^p \mathbb{E}_1 \left[ \left( 1 + \frac{4}{\kappa} \sup_{s \in [0,T]} \| \vecc{a}_2(s, \vecc{x}^\epsilon(s),\epsilon)\| \right)^p \cdot \Pi \right],  \\
&\leq 2^{p-1} C^p  \left( 1 + \frac{4}{\kappa} \| \vecc{a}_2(t,\vecc{x}^\epsilon(t),\epsilon)\|_{\infty} \right)^p \cdot \mathbb{E}_1 [\Pi],
\end{align}
where  $\| \vecc{a}_2(t,\vecc{x}^\epsilon(t),\epsilon)\|_{\infty} := \sup_{ t \in [0,T], \vecc{y} \in \RR^{n_1},\epsilon \in \mathcal{E}} \|\vecc{a}_2(t,\vecc{y},\epsilon)\|$ and
\begin{align}
\Pi &=  e^{-p \delta \kappa/\epsilon}  \sup_{t \in [0,T]} \bigg| \int_0^t \vecc{\sigma}_2(s,\vecc{x}^{\epsilon}(s),\epsilon) d\vecc{W}^{(k_2)}(s) \bigg|^p \nonumber \\
&\ \ \ + \max_{k=0,\dots,N-1} \sup_{t \in [k \delta, (k+2)\delta]}  \bigg| \int_{k \delta}^t  \vecc{\sigma}_2(s,\vecc{x}^{\epsilon}(s),\epsilon) d\vecc{W}^{(k_2)}(s) \bigg|^p.
\end{align}

We estimate:
\begin{align}
\mathbb{E}_1 [\Pi]  &=  e^{-p \delta \kappa/\epsilon} \mathbb{E}_1 \left[ \sup_{t \in [0,T]} \bigg| \int_0^t \vecc{\sigma}_2(s,\vecc{x}^{\epsilon}(s),\epsilon) d\vecc{W}^{(k_2)}(s) \bigg|^p\right] \nonumber \\ 
&\ \ \ \ + \mathbb{E}_1 \left[ \max_{k=0,\dots,N-1} \sup_{t \in [k \delta, (k+2)\delta]}  \bigg| \int_{k \delta}^t  \vecc{\sigma}_2(s,\vecc{x}^{\epsilon}(s),\epsilon) d\vecc{W}^{(k_2)}(s) \bigg|^p \right] \\
&\leq  e^{-p \delta \kappa/\epsilon} \mathbb{E}_1 \left[ \sup_{t \in [0,T]} \bigg| \int_0^t \vecc{\sigma}_2(s,\vecc{x}^{\epsilon}(s),\epsilon) d\vecc{W}^{(k_2)}(s) \bigg|^p  \right] \nonumber \\ 
&\ \ \ \ +  \mathbb{E}_1 \left[ \left( \sum_{k=0}^{N-1}  \sup_{t \in [k \delta, (k+2)\delta]} \left(\int_{k \delta}^t  \vecc{\sigma}_2(s,\vecc{x}^{\epsilon}(s),\epsilon) d\vecc{W}^{(k_2)}(s) \right)^{pq} \right)^{1/q} \right] \\
&\leq  e^{-p \delta \kappa/\epsilon} \mathbb{E}_1 \left[ \sup_{t \in [0,T]} \bigg| \int_0^t \vecc{\sigma}_2(s,\vecc{x}^{\epsilon}(s),\epsilon) d\vecc{W}^{(k_2)}(s) \bigg|^p \right] \nonumber \\
&\ \ \ \  +  \left( \sum_{k=0}^{N-1} \mathbb{E}_1 \left[ \sup_{t \in [k \delta, (k+2)\delta]} \left(\int_{k \delta}^t  \vecc{\sigma}_2(s, \vecc{x}^{\epsilon}(s),\epsilon) d\vecc{W}^{(k_2)}(s) \right)^{pq} \right)^{1/q} \right],
\end{align}
with $N := \max\{k \in \ZZ : k \delta < T\}$, where we have used the fact that the $l^\infty$-norm on $\RR^{N}$ is bounded by the $l^q$ norm for every $q \geq 1$ and then applied H\"older's inequality to get the last two lines above.

Now, letting $\delta = \epsilon^{1-h}$ for $0 < h < 1$, and using the Burkholder-Davis-Gundy inequality,
\begin{align}
\mathbb{E}_1[\Pi] &\leq C_{p,q} \bigg[ e^{-p \kappa/\epsilon^h} \mathbb{E}_1  \bigg[ \bigg( \int_0^T \|\vecc{\sigma}_2(s,\vecc{x}^{\epsilon}(s),\epsilon)\|_{F}^2 ds \bigg)^{\frac{pq}{2}} \bigg]^{1/q} \nonumber  \\ 
&\ \ \ \ +  \left( \sum_{k=0}^{N-1} \mathbb{E}_1 \left(\int_{k \delta}^{(k+2)\delta}  \|\vecc{\sigma}_2(s,\vecc{x}^{\epsilon}(s),\epsilon) \|_F^2 ds \right)^{\frac{pq}{2}} \right)^{1/q} \bigg]  \\
&\leq C_{p,q} \|\vecc{\sigma}_2(s,\vecc{x}^\epsilon(s),\epsilon)\|^p_{F,\infty} (e^{-p \kappa/\epsilon^h} T^{p/2} + 2^{p/2} (N \delta^{\frac{pq}{2}})^{1/q}),
\end{align}
where $C_{p,q}$ is some constant and 
\begin{equation} 
\|\vecc{\sigma}_2(s,\vecc{x}^\epsilon(s),\epsilon)\|_{F,\infty} := \sup_{t \in [0,T], \vecc{y} \in \RR^{n_{1}}, \epsilon \in \mathcal{E}} \|\vecc{\sigma}_2(t, \vecc{y},\epsilon)\|_F < \infty. 
\end{equation}


Since $N \delta < T$, we have $N \delta^{pq/2} < T \delta^{pq/2-1} = T \epsilon^{(1-h)(pq/2 - 1)}$. Therefore, $\mathbb{E}_1 [\Pi ] = O(\epsilon^{(1-h)(p/2-1/q)})$.
For all $0 < \beta < p/2$, one can choose $0 < h < 1$ and $q > 1$ such that $(1-h)(p/2-1/q) = \beta$. 

Therefore, we have
\begin{equation}
\mathbb{E}_1 \left[  \sup_{t \in [0,T]} \left| \vecc{\Phi}_{\epsilon}(t) \int_0^t \vecc{\Phi}_{\epsilon}^{-1}(s) \vecc{\sigma}_2(s,\vecc{x}^\epsilon(s),\epsilon) d\vecc{W}^{(k_2)}(s) \right|^p \right]= O(\epsilon^\beta)
\end{equation}
 as $\epsilon \to 0$, for all $0 < \beta < p/2$.

Combining all the estimates obtained, one has:
\begin{equation}
\mathbb{E}_1\left[\sup_{t \in [0,T]} |\vecc{p}^\epsilon(t)|^p \right] \leq C_1 \epsilon^\alpha  + C_2 \epsilon^{p} + C_3 \epsilon^{p b_2} + C_4 \epsilon^\beta
\end{equation}
where the $C_i$ are positive constants, $\alpha \geq p/2$ is some constant, and $b_2 > 0$ is the constant from Assumption \ref{a5_ch2}. The statement of the proposition follows.

\end{proof}

We also need the following estimate on a class of integrals with respect to products of the coordinates of the  process $\vecc{p}^{\epsilon}(t)$.

\begin{prop} \label{bound_on_integ_wrt_p_square}
Suppose that Assumptions \ref{aexis}-\ref{a5_ch2} hold and $\epsilon \in \mathcal{E}$. 
Let $h^\epsilon: \RR^+ \times  \RR^{n_1}  \to \RR$ be a family of functions, continuously differentiable in $\vecc{y} \in \RR^{n_1}$ and bounded (in $s \in \RR^+$ and $\vecc{y} \in \RR^{n_1}$), with bounded first derivatives $\vecc{\nabla}_{\vecc{y}} h^\epsilon(\vecc{y})$ for $\vecc{y} \in \RR^{n_1}$. Assume that $h^\epsilon$ and $\vecc{\nabla}_{\vecc{y}} h^\epsilon(\vecc{y})$ are $O(1)$ as $\epsilon \to 0$. Moreover, assume that $\frac{\partial}{\partial s}h^\epsilon$ is bounded (in all variables) and is $O(1)$ as $\epsilon \to 0$.

Then for any $p \geq 1$, $T>0$, $0< \beta < p/2$, $i,j = 1, \dots,n_2$,  in the limit as $\epsilon \to 0$ we have
\begin{equation}
\mathbb{E}\left[\sup_{t \in [0,T]} \left| \int_{0}^{t} h^\epsilon(s,\vecc{x}^{\epsilon}(s)) d( [\vecc{p}^{\epsilon}]_{i}(s)\cdot [\vecc{p}^{\epsilon}]_{j}(s)) \right|^{p}; \epsilon \leq \epsilon_1 \right] = O(\epsilon^\beta),\end{equation}
where $\vecc{x}^\epsilon(t)$ and $\vecc{p}^\epsilon(t)$ solve the SDEs \eqref{sde1}-\eqref{sde2} and the SDE \eqref{sdeforp} respectively, and $\epsilon_1$ is from Proposition \ref{mom_bound}. \end{prop}

\begin{proof} 
Let $\epsilon \in \mathcal{E}$, $t \in [0,T]$, and $i, j =1,\dots, n_2$.  An integration by parts gives:
\begin{align}
&\int_{0}^{t} h^\epsilon(s,\vecc{x}^{\epsilon}(s)) d( [\vecc{p}^{\epsilon}]_{i}(s)\cdot [\vecc{p}^{\epsilon}]_{j}(s))  \nonumber \\ 
&= h^\epsilon(t, \vecc{x}^\epsilon(t)) [\vecc{p}^{\epsilon}]_{i}(t)  [\vecc{p}^{\epsilon}]_{j}(t)  - h^\epsilon(t, \vecc{x}^\epsilon) [\vecc{p}^{\epsilon}]_{i}   [\vecc{p}^{\epsilon}]_{j}   \nonumber \\ 
&\ \ \ \ - \int_0^t [\vecc{p}^{\epsilon}]_{i}(s) [\vecc{p}^{\epsilon}]_{j}(s) \left( \vecc{\nabla}_{\vecc{x}^\epsilon} h^\epsilon(s, \vecc{x}^\epsilon(s)) \cdot \frac{\vecc{p}^\epsilon(s)}{\epsilon} + \frac{\partial}{\partial s} h^\epsilon(s, \vecc{x}^\epsilon(s))  \right) ds. 
\end{align}

Using the notation $\mathbb{E}_1[\cdot] = \mathbb{E}[\cdot; \epsilon \leq \epsilon_1]$, we estimate, for $p \geq 1$, 
\begin{align}
&\mathbb{E}_1 \left[\sup_{t \in [0,T]} \left| \int_{0}^{t} h^\epsilon(s,\vecc{x}^{\epsilon}(s)) d( [\vecc{p}^{\epsilon}]_{i}(s)\cdot [\vecc{p}^{\epsilon}]_{j}(s)) \right|^{p}\right] \nonumber \\
&\leq 4^{p-1}\bigg( \mathbb{E}_1 \sup_{t \in [0,T]} \left|  h^\epsilon(t, \vecc{x}^\epsilon(t)) [\vecc{p}^{\epsilon}]_{i}(t)  [\vecc{p}^{\epsilon}]_{j}(t)  \right|^p \nonumber \\
&\ \ \ \ \  \ \ \ + \mathbb{E}_1 \sup_{t \in [0,T]} \left| h^\epsilon(t, \vecc{x}^\epsilon) [\vecc{p}^{\epsilon}]_{i}   [\vecc{p}^{\epsilon}]_{j}   \right|^p \nonumber \\
&\ \ \ \ \ \ \  \ + \mathbb{E}_1 \sup_{t \in [0,T]} \left| \int_0^t [\vecc{p}^{\epsilon}]_{i}(s) [\vecc{p}^{\epsilon}]_{j}(s)  \vecc{\nabla}_{\vecc{x}^\epsilon} h^\epsilon(s, \vecc{x}^\epsilon(s)) \cdot \frac{\vecc{p}^\epsilon(s)}{\epsilon}  ds \right|^p \nonumber \\
&\ \ \ \ \ \ \ \ + \mathbb{E}_1 \sup_{t \in [0,T]} \left| \int_0^t [\vecc{p}^{\epsilon}]_{i}(s) [\vecc{p}^{\epsilon}]_{j}(s) \frac{\partial}{\partial s} h^\epsilon(s, \vecc{x}^\epsilon(s))   ds \right|^p  \bigg) \\
&\leq C(p,T) \bigg[ \|h^\epsilon\|^p_{\infty} \left(\mathbb{E}_1 \sup_{t \in [0,T]} |\vecc{p}^\epsilon(t)|^{2p} +  \mathbb{E}_1 |\vecc{p}^\epsilon|^{2p} \right) \nonumber \\
&\ \ \ \ \ \ \ \ + \frac{1}{\epsilon^p} \mathbb{E}_1 \sup_{t \in [0,T]} \left| \int_0^t [\vecc{p}^{\epsilon}]_{i}(s) [\vecc{p}^{\epsilon}]_{j}(s)  [\vecc{\nabla}_{\vecc{x}^\epsilon} h^\epsilon]_k (s, \vecc{x}^\epsilon(s)) [\vecc{p}^\epsilon]_k(s)  ds \right|^p \nonumber \\
&\ \ \ \ \ \ \ \  + \left\|\frac{\partial}{\partial s} h^\epsilon\right\|_{\infty}^p \cdot \mathbb{E}_1 \sup_{t \in [0,T]} |\vecc{p}^\epsilon(t)|^{2p}   \bigg],
\end{align}
where $C(p,T) > 0$ is a constant, $\| g^\epsilon \|_{\infty} := \sup_{s \in [0,T], \vecc{y} \in \RR^{n_{1}}} |g^\epsilon(s, \vecc{y})|$, and we have used Einstein's summation over repeated indices convention.

Now, estimating as before, we obtain:
\begin{align}
&\mathbb{E}_1 \sup_{t \in [0,T]} \left| \int_0^t [\vecc{p}^{\epsilon}]_{i}(s) [\vecc{p}^{\epsilon}]_{j}(s)  [\vecc{\nabla}_{\vecc{x}^\epsilon} h^\epsilon]_k (s, \vecc{x}^\epsilon(s)) [\vecc{p}^\epsilon]_k(s)  ds \right|^p \nonumber \\
&\leq D(p,T) \| \vecc{\nabla}_{\vecc{x}^\epsilon} h^\epsilon\|_{\infty} \cdot \mathbb{E}_1 \sup_{t \in [0,T]} |\vecc{p}^\epsilon(t)|^{3p},
\end{align}
where $D(p,T)>0$ is a constant. 

By our assumptions, all the quantities in the form $\| \cdot \|_{\infty}$ are bounded and are $O(1)$ as $\epsilon \to 0$.  Therefore, collecting the above estimates, using Assumption \ref{a2_ch2}, and applying Proposition \ref{mom_bound}, we have, for $p \geq 1$, $T>0$, $i,j=1,\dots,n_2$,  
\begin{align}
&\mathbb{E}_1 \left[\sup_{t \in [0,T]} \left| \int_{0}^{t} h^\epsilon(s,\vecc{x}^{\epsilon}(s)) d( [\vecc{p}^{\epsilon}]_{i}(s)\cdot [\vecc{p}^{\epsilon}]_{j}(s)) \right|^{p}\right] = O(\epsilon^\beta),
\end{align}
for every $0 < \beta < p/2$. 

\end{proof}


Now we proceed to prove Theorem \ref{mainthm}.
Using the above moment estimates and the proof techniques in \cite{birrell2017small, 2017BirrellLatest}, we are going to first obtain the convergence of $\vecc{x}^\epsilon_{t}$ to $\vecc{X}_{t}$ in the limit as $\epsilon \to 0$ in the following sense:   for all finite $T>0$, $p \geq 1$,
\begin{equation} \label{nonst}
\mathbb{E}\left[ \sup_{t \in [0,T]} |\vecc{x}_t^\epsilon - \vecc{X}_t|^p; \epsilon \leq \epsilon_1 \right] \to 0,
\end{equation}
as $\epsilon \to 0$, where the $\epsilon_1$ is from Proposition \ref{mom_bound}. The main tools are well known ordinary and stochastic integral inequalities, as well as a Gronwall type argument. This result would then imply that for all finite $T>0$,  $\sup_{t \in [0,T]} |\vecc{x}_t^\epsilon - \vecc{X}_t| \to 0$ in probability, in the limit as $\epsilon \to 0$ (see Lemma 1 in \cite{LimWehr_Homog_NonMarkovian}).

\begin{proof} (Proof of Theorem \ref{mainthm})
Let $T>0$ and recall that $[\vecc{B}]_{i,j}$ denotes the $(i,j)$-entry of a matrix $\vecc{B}$.
First, we assume that $p>2$.

From \eqref{sde2}, we have, for every $\epsilon > 0$, $t \in [0,T]$,
\begin{align}
\vecc{v}^{\epsilon}(t) dt &= \epsilon \vecc{a}_{2}^{-1}(t,\vecc{x}^{\epsilon}(t),\epsilon) d\vecc{v}^{\epsilon}(t)  - \vecc{a}_{2}^{-1}(t,\vecc{x}^{\epsilon}(t),\epsilon) \vecc{b}_{2}(t,\vecc{x}^{\epsilon}(t),\epsilon) dt  \nonumber \\ 
&\ \ \ \ - \vecc{a}_{2}^{-1}(t,\vecc{x}^{\epsilon}(t),\epsilon) \vecc{\sigma}_2(t,\vecc{x}^{\epsilon}(t), \epsilon) d\vecc{W}^{(k_2)}(t).
\end{align}

Substituting this into \eqref{sde1}, we obtain:
\begin{align}
d\vecc{x}^{\epsilon}(t) &= \epsilon \vecc{a}_{1}(t,\vecc{x}^{\epsilon}(t), \epsilon)  \vecc{a}_{2}^{-1}(t,\vecc{x}^{\epsilon}(t),\epsilon) d\vecc{v}^{\epsilon}(t) \nonumber \\
&\ \ \  - \vecc{a}_{1}(t, \vecc{x}^{\epsilon}(t), \epsilon) \vecc{a}_{2}^{-1}(t, \vecc{x}^{\epsilon}(t),\epsilon) \vecc{b}_{2}(t, \vecc{x}^{\epsilon}(t),\epsilon) dt \nonumber \\
&\ \ \ - \vecc{a}_{1}(t, \vecc{x}^{\epsilon}(t), \epsilon) \vecc{a}_{2}^{-1}(t, \vecc{x}^{\epsilon}(t),\epsilon) \vecc{\sigma}_2(t, \vecc{x}^{\epsilon}(t), \epsilon) d\vecc{W}^{(k_2)}(t) \nonumber \\
&\ \ \  + \vecc{b}_{1}(t, \vecc{x}^{\epsilon}(t),\epsilon) dt + \vecc{\sigma}_{1}(t, \vecc{x}^{\epsilon}(t),\epsilon) d\vecc{W}^{(k_1)}(t).
\end{align}

In integral form, we have:
\begin{align}
\vecc{x}^{\epsilon}(t) &= \vecc{x}^\epsilon + \epsilon \int_0^t  \vecc{a}_{1}(s, \vecc{x}^{\epsilon}(s), \epsilon)  \vecc{a}_{2}^{-1}(s, \vecc{x}^{\epsilon}(s),\epsilon) d\vecc{v}^{\epsilon}(s) \nonumber \\
&\ \ \ \ + \int_0^t \{ \vecc{b}_{1}(s, \vecc{x}^{\epsilon}(s),\epsilon) -\vecc{a}_{1}(s, \vecc{x}^{\epsilon}(s), \epsilon) \vecc{a}_{2}^{-1}(s, \vecc{x}^{\epsilon}(s),\epsilon) \vecc{b}_{2}(s, \vecc{x}^{\epsilon}(s),\epsilon)  \} ds \nonumber \\
&\ \ \ \ - \int_0^t \vecc{a}_{1}(s, \vecc{x}^{\epsilon}(s), \epsilon) \vecc{a}_{2}^{-1}(s, \vecc{x}^{\epsilon}(s),\epsilon) \vecc{\sigma}_2(s, \vecc{x}^{\epsilon}(s), \epsilon) d\vecc{W}^{(k_2)}(s) \nonumber \\ 
&\ \ \ \ + \int_0^t \vecc{\sigma}_{1}(s, \vecc{x}^{\epsilon}(s),\epsilon) d\vecc{W}^{(k_1)}(s).
\end{align}

Its $i$th component, $[\vecc{x}^{\epsilon}]_{i}(t)$ ($i=1,2,\dots,n_1$) is (recall that we are employing Einstein's summation convention):
\begin{align}
[\vecc{x}^{\epsilon}]_i(t) &= [\vecc{x}^\epsilon]_i + \epsilon \int_0^t [\vecc{a}_{1}\vecc{a}_{2}^{-1}]_{i,j}(s, \vecc{x}^{\epsilon}(s), \epsilon) \cdot d[\vecc{v}^{\epsilon}]_j(s) \nonumber \\ 
&\ \ \ \ + \int_0^t \{ [\vecc{b}_{1}]_i(s, \vecc{x}^{\epsilon}(s),\epsilon) -[\vecc{a}_{1}\vecc{a}_{2}^{-1}\vecc{b}_{2}]_{i}(s, \vecc{x}^{\epsilon}(s), \epsilon)  \} ds \nonumber \\
&\ \ \ \ - \int_0^t [\vecc{a}_{1}\vecc{a}_{2}^{-1}\vecc{\sigma}_2]_{i,j}(s, \vecc{x}^{\epsilon}(s), \epsilon) \cdot d[\vecc{W}^{(k_2)}]_j(s) \nonumber \\
&\ \ \ \ + \int_0^t [\vecc{\sigma}_{1}]_{i,j}(s, \vecc{x}^{\epsilon}(s),\epsilon) \cdot d[\vecc{W}^{(k_1)}]_j(s).
\end{align}

Next, we perform integration by parts in the second term on the right hand side above:
\begin{align}
&\int_0^t [S^{\epsilon}]_i(s, \vecc{x}^\epsilon(s),\vecc{v}^\epsilon(s),\epsilon)ds :=\epsilon \int_0^t [\vecc{a}_{1}\vecc{a}_{2}^{-1}]_{i,j}(s, \vecc{x}^{\epsilon}(s), \epsilon) \cdot d[\vecc{v}^{\epsilon}]_j(s)  \\
&= \epsilon  [\vecc{a}_{1}\vecc{a}_{2}^{-1}]_{i,j}(t, \vecc{x}^{\epsilon}(t),\epsilon) \cdot [\vecc{v}^{\epsilon}]_j(t) - \epsilon  [\vecc{a}_{1}\vecc{a}_{2}^{-1}]_{i,j}(0, \vecc{x},\epsilon) \cdot [\vecc{v}^\epsilon]_j \nonumber \\
&\ \ \ \ \ - \int_0^t \frac{\partial}{\partial [\vecc{x}^{\epsilon}]_l(s)}\bigg(  [\vecc{a}_{1}\vecc{a}_{2}^{-1}]_{i,j}(s, \vecc{x}^{\epsilon}(s),\epsilon) \bigg) \cdot d[\vecc{x}^\epsilon]_l(s) \cdot \epsilon [\vecc{v}^{\epsilon}]_j(s) \nonumber \\
&\ \ \ \ \ -  \int_0^t  \frac{\partial}{\partial s}\left([\vecc{a}_1 \vecc{a}_2^{-1}]_{i,j}(s, \vecc{x}^\epsilon(s),\epsilon) \right) \cdot \epsilon [\vecc{v}^\epsilon]_j(s) ds.  \label{88}
\end{align}

Substituting the following expression for $d[\vecc{x}^\epsilon]_l(s)$:
\begin{align}
d[\vecc{x}^\epsilon]_l(s) &= [\vecc{a}_1]_{l,k}(s,\vecc{x}^\epsilon(s),\epsilon)[\vecc{v}^\epsilon]_k(s) ds + [\vecc{b}_1]_l(s,\vecc{x}^\epsilon(s),\epsilon)ds \nonumber \\
&\ \ \ \  + [\vecc{\sigma}_1]_{l,k}(s,\vecc{x}^\epsilon(s),\epsilon)d[\vecc{W}^{(k_1)}]_k(s)
\end{align}
into  \eqref{88}, we obtain:
\begin{align}
&\int_0^t [S^{\epsilon}]_i(s, \vecc{x}^\epsilon(s),\vecc{v}^\epsilon(s),\epsilon)ds \nonumber \\
&= \epsilon  [\vecc{a}_{1}\vecc{a}_{2}^{-1}]_{i,j}(t, \vecc{x}^{\epsilon}(t),\epsilon) \cdot [\vecc{v}^{\epsilon}]_j(t) - \epsilon  [\vecc{a}_{1}\vecc{a}_{2}^{-1}]_{i,j}(0, \vecc{x},\epsilon) \cdot [\vecc{v}^\epsilon]_j \nonumber \\
&  - \int_0^t \frac{\partial}{\partial [\vecc{x}^{\epsilon}]_l(s)}\bigg(  [\vecc{a}_{1}\vecc{a}_{2}^{-1}]_{i,j}(s, \vecc{x}^{\epsilon}(s),\epsilon) \bigg) \cdot [\vecc{b}_1]_l(s, \vecc{x}^\epsilon(s),\epsilon)  \cdot \epsilon [\vecc{v}^{\epsilon}]_j(s) ds \nonumber \\
& - \int_0^t \frac{\partial}{\partial [\vecc{x}^{\epsilon}]_l(s)}\bigg(  [\vecc{a}_{1}\vecc{a}_{2}^{-1}]_{i,j}(s, \vecc{x}^{\epsilon}(s),\epsilon) \bigg)  [\vecc{\sigma}_1]_{l,k}(s, \vecc{x}^\epsilon(s),\epsilon)  \epsilon [\vecc{v}^{\epsilon}]_j(s) d[\vecc{W}^{(k_1)}]_{k}(s) \nonumber \\
& - \int_0^t \frac{\partial}{\partial [\vecc{x}^{\epsilon}]_l(s)}\bigg(  [\vecc{a}_{1}\vecc{a}_{2}^{-1}]_{i,j}(s, \vecc{x}^{\epsilon}(s),\epsilon) \bigg)  [\vecc{a}_1]_{l,k}(s, \vecc{x}^\epsilon(s),\epsilon)  \epsilon [\vecc{v}^\epsilon]_k(s)  [\vecc{v}^{\epsilon}]_j(s) ds \nonumber \\
& -  \int_0^t  \frac{\partial}{\partial s}\left([\vecc{a}_1 \vecc{a}_2^{-1}]_{i,j}(s, \vecc{x}^\epsilon(s),\epsilon) \right) \cdot \epsilon [\vecc{v}^\epsilon]_j(s) ds. 
\end{align}

Next, we apply It\^o formula to $ \epsilon \vecc{v}^{\epsilon}(t) (\epsilon \vecc{v}^{\epsilon}(t))^{*}  \in \RR^{n_2\times n_2}$:
\begin{align}
&d[\epsilon \vecc{v}^{\epsilon}(t) (\epsilon \vecc{v}^{\epsilon}(t))^{*}] \nonumber \\
&= \epsilon d\vecc{v}^{\epsilon}(t) \cdot \epsilon (\vecc{v}^{\epsilon}(t))^* + \epsilon \vecc{v}^{\epsilon}(t) \cdot \epsilon d(\vecc{v}^{\epsilon}(t))^{*} + d[\epsilon \vecc{v}^{\epsilon}(t)] \cdot d[ (\epsilon \vecc{v}^{\epsilon}(t))^{*}] \\
&= \left[\vecc{a}_{2}(t, \vecc{x}^{\epsilon}(t),\epsilon) \vecc{v}^{\epsilon}(t) dt + \vecc{b}_{2}(t, \vecc{x}^{\epsilon}(t),\epsilon) dt + \vecc{\sigma}_2(t, \vecc{x}^{\epsilon}(t), \epsilon) d\vecc{W}^{(k_2)}(t) \right] \epsilon \vecc{v}^\epsilon(t)^{*} \nonumber \\
&\ \ + \epsilon \vecc{v}^\epsilon(t)\left[
\vecc{a}_{2}(t, \vecc{x}^{\epsilon}(t),\epsilon) \vecc{v}^{\epsilon}(t) dt + \vecc{b}_{2}(t, \vecc{x}^{\epsilon}(t),\epsilon) dt + \vecc{\sigma}_2(t, \vecc{x}^{\epsilon}(t), \epsilon) d\vecc{W}^{(k_2)}(t) \right]^{*}\nonumber \\
&\  \ + \vecc{\sigma}_2(t, \vecc{x}^{\epsilon}(t), \epsilon)\vecc{\sigma}_2^*(t, \vecc{x}^{\epsilon}(t), \epsilon) dt.
\end{align}

Denoting $\vecc{J}^\epsilon(t) := \epsilon \vecc{v}^\epsilon(t) (\vecc{v}^\epsilon(t))^{*}$,
we can rewrite the above as:
\begin{equation} \label{lyap_prelimit}
-\vecc{a}_{2}(t, \vecc{x}^{\epsilon}(t),\epsilon) \vecc{J}^\epsilon(t)dt - \vecc{J}^\epsilon(t) \vecc{a}_{2}^*(t, \vecc{x}^{\epsilon}(t),\epsilon)dt = \vecc{F}^{\epsilon}_1(t) dt + \vecc{F}^{\epsilon}_2(t) dt + \vecc{F}^{\epsilon}_3(t) dt,
\end{equation}
where
\begin{align}
\vecc{F}^{\epsilon}_1(t) dt &= -d[\epsilon \vecc{v}^{\epsilon}(t) (\epsilon \vecc{v}^{\epsilon}(t))^{*}], \\
\vecc{F}^{\epsilon}_2(t) dt &= (\vecc{b}_{2}(t, \vecc{x}^{\epsilon}(t),\epsilon) dt + \vecc{\sigma}_2(t, \vecc{x}^{\epsilon}(t), \epsilon)d\vecc{W}^{(k_2)}(t) )\epsilon (\vecc{v}^{\epsilon}(t))^{*} \nonumber \\
&\ \ \ \ \ + \epsilon \vecc{v}^{\epsilon}(t)(\vecc{b}_{2}(t, \vecc{x}^{\epsilon}(t),\epsilon) dt + \vecc{\sigma}_2(t,\vecc{x}^{\epsilon}(t), \epsilon) d\vecc{W}^{(k_2)}(t))^{*},\\
\vecc{F}^{\epsilon}_3(t) &= \vecc{\sigma}_2(t,\vecc{x}^{\epsilon}(t), \epsilon) \vecc{\sigma}_2(t,\vecc{x}^{\epsilon}(t), \epsilon)^{*}.
\end{align}

Since $-\vecc{a}_2(t, \vecc{x}^\epsilon(t),\epsilon)$ is positive stable uniformly (in $t$, $\vecc{x}^\epsilon$ and $\epsilon$) by Assumption \ref{a0_ch2}, the solution of the Lyapunov equation \eqref{lyap_prelimit} can be represented as:
\begin{equation}
\vecc{J}^\epsilon(t)  = \vecc{J}_1^\epsilon(t) + \vecc{J}_2^\epsilon(t) + \vecc{J}_3^\epsilon(t),
\end{equation}
where
\begin{align}
\vecc{J}_n^\epsilon(t) &= \int_0^\infty  e^{\vecc{a}_2(t, \vecc{x}^\epsilon(t),\epsilon)y}  \vecc{F}^{\epsilon}_n(t) e^{\vecc{a}^*_2(t, \vecc{x}^\epsilon(t),\epsilon)y} dy
\end{align}
for $n=1,2,3$.

Therefore, for $s \in [0,T]$,
\begin{align}
&\epsilon [\vecc{v}^\epsilon]_j(s) [\vecc{v}^\epsilon]_k(s)ds \nonumber \\
&=-\int_0^\infty  \bigg[e^{\vecc{a}_2(s, \vecc{x}^\epsilon(s),\epsilon)y}\bigg]_{j,p_1}  \cdot \bigg[ d[\epsilon \vecc{v}^{\epsilon}(s) (\epsilon \vecc{v}^{\epsilon}(s))^{*}] \bigg]_{p_1,p_2} \cdot \bigg[ e^{\vecc{a}^*_2(s, \vecc{x}^\epsilon(s),\epsilon)y}\bigg]_{p_2,k} dy \nonumber \\
&\ \ \ +  \int_0^\infty  \bigg[ e^{\vecc{a}_2(s, \vecc{x}^\epsilon(s),\epsilon)y} \bigg]_{j,p_1}  \cdot\bigg[ (\vecc{b}_{2}(s, \vecc{x}^{\epsilon}(s),\epsilon) ds \nonumber \\ 
&\hspace{1cm}  + \vecc{\sigma}_2(s, \vecc{x}^{\epsilon}(s), \epsilon)d\vecc{W}^{(k_2)}(s) )\epsilon (\vecc{v}^{\epsilon}(s))^{*} \bigg]_{p_1,p_2} \cdot \bigg[  e^{\vecc{a}^*_2(s, \vecc{x}^\epsilon(s),\epsilon)y} \bigg]_{p_2,k} dy  \nonumber \\
&\ \ \ + \int_0^\infty  \bigg[ e^{\vecc{a}_2(s, \vecc{x}^\epsilon(s),\epsilon)y} \bigg]_{j,p_1} \cdot \bigg[ \epsilon \vecc{v}^{\epsilon}(s)(\vecc{b}_{2}(s,\vecc{x}^{\epsilon}(s),\epsilon) ds \nonumber \\
&\hspace{1cm} + \vecc{\sigma}_2(s,\vecc{x}^{\epsilon}(s), \epsilon) d\vecc{W}^{(k_2)}(s))^{*} \bigg]_{p_1,p_2} \cdot \bigg[ e^{\vecc{a}^*_2(s,\vecc{x}^\epsilon(s),\epsilon)y} \bigg]_{p_2,k} dy \nonumber \\
&\ \ \ + \int_0^\infty \bigg[  e^{\vecc{a}_2(s, \vecc{x}^\epsilon(s),\epsilon)y} \bigg]_{j,p_1} \cdot \bigg[ \vecc{\sigma}_2(s, \vecc{x}^{\epsilon}(s), \epsilon) \vecc{\sigma}_2(s, \vecc{x}^{\epsilon}(s), \epsilon)^{*} ds \bigg]_{p_1,p_2} \nonumber \\ 
&\hspace{1cm} \cdot \bigg[ e^{\vecc{a}^*_2(s, \vecc{x}^\epsilon(s),\epsilon)y} \bigg]_{p_2,k} dy.
\end{align}

On the other hand, by \eqref{mainlimitingeqn}, 
\begin{align}
\vecc{X}(t) &= \vecc{x} + \int_0^t  [\vecc{B}_1(s, \vecc{X}(s))-\vecc{A}_1(s, \vecc{X}(s))\vecc{A}_2^{-1}(s, \vecc{X}(s))\vecc{B}_2(s, \vecc{X}(s))] ds\nonumber \\
&\ \ \ \ \  + \int_0^t \vecc{S}(s, \vecc{X}(s)) ds  + \int_0^t \vecc{\Sigma}_1(s, \vecc{X}(s)) d\vecc{W}^{(k_1)}(s) \nonumber \\ 
&\ \ \ \ \ - \int_0^t \vecc{A}_1(s, \vecc{X}(s)) \vecc{A}_2^{-1}(s, \vecc{X}(s))\vecc{\Sigma}_2(s, \vecc{X}(s)) d\vecc{W}^{(k_2)}(s).
\end{align}

We use again the notation  $\mathbb{E}_1[ \cdot ] := \mathbb{E}[\cdot; \epsilon \leq \epsilon_1]$, where $\epsilon_1 > 0$ is the random variable from Proposition \ref{mom_bound}. 
 
For any $p > 2$, $T>0$, $i=1,\dots,n_1$ (recall that $[\vecc{b}]_i$ denotes the $i$th component of vector $\vecc{b}$), we estimate: 
\begin{align}
&\mathbb{E}_1\left[ \sup_{t \in [0,T]} |[\vecc{x}^{\epsilon}(t) - \vecc{X}(t)]_i|^p \right] \nonumber \\
&\leq 6^{p-1}\bigg\{ \mathbb{E}_1\left[|\vecc{x}^\epsilon - \vecc{x}|^p\right] \nonumber \\
&\hspace{0.3cm}  +   \mathbb{E}_1\left[\sup_{t \in [0,T]}  \bigg| \int_0^t \bigg[\vecc{S}_\epsilon(s, \vecc{x}^\epsilon(s),\vecc{v}^\epsilon(s),\epsilon)- \vecc{S}(s, \vecc{X}(s))\bigg]_i  ds \bigg|^p \right]  \nonumber \\
&\hspace{0.3cm} + \mathbb{E}_1 \bigg[ \sup_{t \in [0,T]} \bigg( \int_0^t \bigg| \bigg[ \vecc{a}_{1}(s, \vecc{x}^{\epsilon}(s), \epsilon) \vecc{a}_{2}^{-1}(s, \vecc{x}^{\epsilon}(s),\epsilon) \vecc{b}_{2}(s, \vecc{x}^{\epsilon}(s),\epsilon) \nonumber \\
&\hspace{2cm} - \vecc{A}_1(s, \vecc{X}(s))\vecc{A}_2^{-1}(s,\vecc{X}(s))\vecc{B}_2(s,\vecc{X}(s))  \bigg]_i \bigg|   ds \bigg)^p \bigg]  \nonumber \\
&\hspace{0.3cm} + \mathbb{E}_1 \left[ \sup_{t \in [0,T]} \bigg( \int_0^t \bigg| \bigg[ \vecc{b}_{1}(s,\vecc{x}^{\epsilon}(s),\epsilon) - \vecc{B}_1(s,\vecc{X}(s)) \bigg]_i \bigg| ds \bigg)^p  \right]   \nonumber \\
&\hspace{0.3cm} + \mathbb{E}_1 \bigg[ \sup_{t \in [0,T]}  \bigg| \int_0^t \bigg[ \vecc{a}_{1}(s,\vecc{x}^{\epsilon}(s), \epsilon) \vecc{a}_{2}^{-1}(s,\vecc{x}^{\epsilon}(s),\epsilon) \vecc{\sigma}_2(s,\vecc{x}^{\epsilon}(s), \epsilon) \nonumber \\
&\hspace{1cm} - \vecc{A}_1(s,\vecc{X}(s)) \vecc{A}_2^{-1}(s,\vecc{X}(s))\vecc{\Sigma}_2(s,\vecc{X}(s)) \bigg]_{i,j} d[\vecc{W}^{(k_2)}]_j(s)  \bigg|^p \bigg]  \nonumber \\
&\hspace{0.3cm} + \mathbb{E}_1 \left[\sup_{t \in [0,T]}  \bigg| \int_0^t \bigg[\vecc{\sigma}_{1}(s,\vecc{x}^{\epsilon}(s),\epsilon) - \vecc{\Sigma}_1(s,\vecc{X}(s))  \bigg]_{i,j} d[\vecc{W}^{(k_1)}]_j(s) \bigg|^p \right]   \ \bigg\} \\
&=: 6^{p-1}\left(\sum_{k=0}^5 R_k \right).
\end{align}


By Assumption \ref{a2_ch2},  $R_0 = \mathbb{E}_1 \left[|\vecc{x}^\epsilon - \vecc{x}|^p\right] \leq  \mathbb{E}\left[|\vecc{x}^\epsilon - \vecc{x}|^p \right]  = O(\epsilon^{ p r_0})$ as $\epsilon \to 0$, where $r_0 > 1/2$ is a constant. 
We now estimate each of the $R_k$, $k=1,\dots,5$.


We have:
\begin{align}
R_3 &\leq  \mathbb{E}_1 \sup_{t \in [0,T]} \bigg( \int_0^t | \vecc{b}_{1}(s,\vecc{x}^{\epsilon}(s),\epsilon) - \vecc{B}_1(s,\vecc{X}(s))  | ds \bigg)^p  \\
&= \mathbb{E}_1  \sup_{t \in [0,T]} \bigg( \int_0^t |\vecc{b}_{1}(s,\vecc{x}^\epsilon(s),\epsilon) - \vecc{b}_1(s,\vecc{X}(s),\epsilon) + \vecc{b}_1(s,\vecc{X}(s),\epsilon)  \nonumber \\ 
&\ \ \ \ \  \  \ \   \ \ \ -   \vecc{B}_1(s,\vecc{X}(s)) | ds \bigg)^p  \nonumber \\
&\leq 2^{p-1} \bigg[  \mathbb{E}_1  \sup_{t \in [0,T]} \left(\int_0^t  |\vecc{b}_{1}(s,\vecc{x}^{\epsilon}(s),\epsilon) - \vecc{b}_1(s,\vecc{X}(s),\epsilon)| ds \right)^p \nonumber \\
&\ \ \  +  \mathbb{E}_1  \sup_{t \in [0,T]} \left(\int_0^t |\vecc{b}_1(s,\vecc{X}(s),\epsilon) -   \vecc{B}_1(s,\vecc{X}(s))| ds \right)^p  \bigg] \\
&\leq 2^{p-1}\left[ L^p(\epsilon) \mathbb{E}_1 \sup_{t \in [0,T]} \int_0^t |\vecc{x}^\epsilon(s)-\vecc{X}(s) |^p ds + T^p \beta_1(\epsilon)^p \mathbb{1}_{\{\vecc{b}_1 \neq \vecc{B}_1\}} \right]  \\
&\leq L_3(\epsilon,p,T) \int_0^T \mathbb{E}_1 \sup_{u \in [0,s]}  |\vecc{x}^\epsilon(u)-\vecc{X}(u) |^p ds + C_3(p,T) \beta_1(\epsilon)^p \mathbb{1}_{\{\vecc{b}_1 \neq \vecc{B}_1\}},
\end{align}
on the set $S_1 := \{\epsilon : \epsilon \leq \epsilon_1\}$, where $\mathbb{1}_{A}$ denotes the indicator function of a set $A$, $L_3(\epsilon,p,T) = O(1)$ as $\epsilon \to 0$ and $C_3(p,T)$ is a constant dependent on $p$ and $T$.   In the last two lines of the above estimate, we have used Assumption \ref{a1_ch2}, Assumption \ref{a5_ch2}, and the inequality:  \begin{equation}\mathbb{E}_1 \sup_{t\in [0,T]} \left( \int_0^t |\vecc{u}(s)| ds \right)^p \leq T^{p-1} \mathbb{E}_1 \int_0^T |\vecc{u}(s)|^p ds,  \end{equation}
where $\vecc{u}(s) \in \RR^{n_1}$ for $s \in [0,T]$ (recall that $L(\epsilon) = O(1)$ as $\epsilon \to 0$ by the  Assumption \ref{a1_ch2}).

Using again the above techniques, together with Lemma \ref{lipzlemma}, one obtains:
\begin{align}
R_2 &\leq L_2(\epsilon,p,T) \int_0^T \mathbb{E}_1 \sup_{u \in [0,s]}  |\vecc{x}^\epsilon(u)-\vecc{X}(u) |^p ds \nonumber \\
&\ \ \ \ + C_2(p,T)\left[\alpha_1(\epsilon)^p \mathbb{1}_{\{\vecc{a}_1 \neq \vecc{A}_1\}} +  \alpha_2(\epsilon)^p \mathbb{1}_{\{\vecc{a}_2 \neq \vecc{A}_2\}} + \beta_2(\epsilon)^p \mathbb{1}_{\{\vecc{b}_2 \neq \vecc{B}_2\}} \right],
\end{align}
on $S_1$, where $\alpha_1(\epsilon)$, $\alpha_2(\epsilon)$, $\beta_2(\epsilon)$ are from Assumption \ref{a1_ch2}, $L_2(\epsilon, p,T) = O(1)$ as $\epsilon \to 0$ and $C_2(p,T)$ is a constant.


To estimate  $R_5$, we use the Burkholder-Davis-Gundy inequality:
\begin{align}
R_5 &\leq C'_p \mathbb{E}_1  \bigg( \int_0^T \|\vecc{\sigma}_{1}(s,\vecc{x}^{\epsilon}(s),\epsilon) - \vecc{\Sigma}_1(s,\vecc{X}(s))\|_{F}^2 ds \bigg)^{p/2},
\end{align}
where $C'_p$ is a positive constant and $\|\cdot\|_F$ denotes the Frobenius norm. Using H\"older's inequality, Assumption \ref{a1_ch2}, Assumption \ref{a5_ch2}, and the above techniques, we obtain:
\begin{align}
R_5  &\leq C''_p \mathbb{E}_1  \bigg( \int_0^T \|\vecc{\sigma}_1(s,\vecc{x}^\epsilon(s),\epsilon)-\vecc{\sigma}_1(s,\vecc{X}(s),\epsilon)\|_{F}^2 ds  \bigg)^{p/2} \nonumber \\
&\ \ \ \   + C''_p \mathbb{E}_1 \bigg( \int_0^T \|\vecc{\sigma}_1(s,\vecc{X}(s),\epsilon) - \vecc{\Sigma}_1(s,\vecc{X}(s)) \|_{F}^2 ds  \bigg)^{p/2}    \\
&\leq C''_p T^{\frac{p}{2}-1} \int_0^T \mathbb{E}_1 \|\vecc{\sigma}_1(s,\vecc{x}^\epsilon(s),\epsilon)-\vecc{\sigma}_1(s,\vecc{X}(s),\epsilon)\|_{F}^p ds \nonumber \\
&\ \ \ \  + C'''_p |\gamma_1(\epsilon)|^p T^{\frac{p}{2}} \mathbb{1}_{\{\vecc{\sigma}_1 \neq \vecc{\Sigma}_1 \}} \\
&\leq L_5(\epsilon,p,T) \int_0^T \mathbb{E}_1 \sup_{u \in [0,s]}  |\vecc{x}^\epsilon(u)-\vecc{X}(u) |^p ds + C_5(p,T)\gamma_1(\epsilon)^p \mathbb{1}_{\{\vecc{\sigma}_1 \neq \vecc{\Sigma}_1\}},
\end{align}
on the set $S_1$, where $C_p''$ and $C_p'''$ are constants, $\gamma_1(\epsilon)$ is from Assumption \ref{a1_ch2}, $L_5(\epsilon,p,T) = O(1)$ as $\epsilon \to 0$ and $C_5(p,T)$ is a constant.

Similarly, using the above techniques and Lemma \ref{lipzlemma}, one can show:
\begin{align}
R_4 &\leq L_4(\epsilon,p,T) \int_0^T \mathbb{E}_1 \sup_{u \in [0,s]}  |\vecc{x}^\epsilon(u)-\vecc{X}(u) |^p ds \nonumber \\
&\ \ \ \ + C_4(p,T)\left[\alpha_1(\epsilon)^p \mathbb{1}_{\{\vecc{a}_1 \neq \vecc{A}_1\}} +  \alpha_2(\epsilon)^p \mathbb{1}_{\{\vecc{a}_2 \neq \vecc{A}_2\}} + \gamma_2(\epsilon)^p \mathbb{1}_{\{\vecc{\sigma}_2 \neq \vecc{\Sigma}_2\}} \right],
\end{align}
on $S_1$,  where $\gamma_2(\epsilon)$ is from Assumption \ref{a1_ch2}, $L_4(\epsilon,p,T) = O(1)$ as $\epsilon \to 0$ and $C_4(p,T)$ is a constant.

To obtain a bound for $R_1$, first we estimate:
\begin{align}
& \bigg| \int_0^t \bigg[\vecc{S}^\epsilon(s,\vecc{x}^\epsilon(s),\vecc{v}^\epsilon(s),\epsilon)- \vecc{S}(s,\vecc{X}(s))\bigg]_i  ds \bigg| \nonumber \\
&\leq \bigg|\epsilon  [\vecc{a}_{1}\vecc{a}_{2}^{-1}]_{i,j}(t, \vecc{x}^{\epsilon}(t),\epsilon) \cdot [\vecc{v}^{\epsilon}]_j(t) - \epsilon  [\vecc{a}_{1}\vecc{a}_{2}^{-1}]_{i,j}(0, \vecc{x},\epsilon) \cdot [\vecc{v}]_j  \bigg| \nonumber \\
&\ \ \ + \left|\int_0^t \frac{\partial}{\partial s} \left( [\vecc{a}_1 \vecc{a}_2^{-1}]_{i,j}(s, \vecc{x}^\epsilon(s), \epsilon) \right) \cdot \epsilon [\vecc{v}^\epsilon]_j(s) ds  \right| \nonumber \\
&\ \  \ + \int_0^t \bigg| \frac{\partial}{\partial [\vecc{x}^{\epsilon}]_l(s)}\bigg(  [\vecc{a}_{1}\vecc{a}_{2}^{-1}]_{i,j}(s, \vecc{x}^{\epsilon}(s),\epsilon) \bigg) \cdot [\vecc{b}_1]_l(s, \vecc{x}^\epsilon(s),\epsilon)  \cdot \epsilon [\vecc{v}^{\epsilon}]_j(s) \bigg| ds \nonumber \\
&\ \  \ + \bigg| \int_0^t  \frac{\partial}{\partial [\vecc{x}^{\epsilon}]_l(s)}\bigg(  [\vecc{a}_{1}\vecc{a}_{2}^{-1}]_{i,j}(s, \vecc{x}^{\epsilon}(s),\epsilon) \bigg) \cdot [\vecc{\sigma}_1]_{l,k}(s, \vecc{x}^\epsilon(s),\epsilon)  \nonumber \\
&\ \ \ \ \ \ \ \ \  \cdot \epsilon [\vecc{v}^{\epsilon}]_j(s) d[\vecc{W}^{(k_1)}]_{k}(s) \bigg| \nonumber \\
&\ \ \  + \bigg|  \int_0^t \frac{\partial}{\partial [\vecc{x}^{\epsilon}]_l(s)}\bigg(  [\vecc{a}_{1}\vecc{a}_{2}^{-1}]_{i,j}(s, \vecc{x}^{\epsilon}(s),\epsilon) \bigg) \cdot [\vecc{a}_1]_{l,k}(s, \vecc{x}^\epsilon(s),\epsilon) \cdot [\vecc{J}_1^\epsilon]_{j,k}(s) ds \bigg| \nonumber \\
&\ \ \  + \bigg|  \int_0^t \frac{\partial}{\partial [\vecc{x}^{\epsilon}]_l(s)}\bigg(  [\vecc{a}_{1}\vecc{a}_{2}^{-1}]_{i,j}(s, \vecc{x}^{\epsilon}(s),\epsilon) \bigg) \cdot [\vecc{a}_1]_{l,k}(s, \vecc{x}^\epsilon(s),\epsilon) \cdot [\vecc{J}_2^\epsilon]_{j,k}(s) ds \bigg| \nonumber \\
&\ \ \  + \bigg| \int_0^t  \frac{\partial}{\partial [\vecc{X}]_l(s)}\left([\vecc{A}_1\vecc{A}_2^{-1}]_{i,j}(s, \vecc{X}(s)) \right) \cdot [\vecc{A}_1]_{l,k}(s, \vecc{X}(s))\cdot [\vecc{J}]_{j,k}(s) \nonumber \\
&\ \ \ \ \ \ \ \  \ - \frac{\partial}{\partial [\vecc{x}^{\epsilon}]_l(s)}\bigg(  [\vecc{a}_{1}\vecc{a}_{2}^{-1}]_{i,j}(s, \vecc{x}^{\epsilon}(s),\epsilon) \bigg) \cdot [\vecc{a}_1]_{l,k}(s, \vecc{x}^\epsilon(s),\epsilon) \cdot [\vecc{J}_3^\epsilon]_{j,k}(s) ds \bigg| \\
&=: \sum_{k=0}^6 \Pi_k,
\end{align}
and so $R_1 \leq 6^{p-1} \sum_{k=0}^6 \left( \mathbb{E}_1 \sup_{t\in [0,T]} |\Pi_k|^p \right) =: 6^{p-1} \sum_{k=0}^6 M_k. $

It is straightforward to show, using the boundedness assumptions of the theorem, that for $k=0,1,2,3,5$:
\begin{equation}
M_k \leq C_k(p,T) \cdot \mathbb{E}_1 \sup_{t\in [0,T]} |\epsilon \vecc{v}^\epsilon(t)|^p,
\end{equation}
where the $C_k$ are positive constants.

Applying Proposition \ref{bound_on_integ_wrt_p_square}, we obtain:
\begin{equation} \label{M4bound}
M_4 := \mathbb{E}_1 \sup_{t \in [0,T]} |\Pi_4|^p \leq C_4(p,T) \epsilon^\beta,
\end{equation}
on $S_1$, for all $0 < \beta < p/2$, as $\epsilon \to 0$,
where  $C_4(p,T)$ is a positive constant.

We now estimate $M_6$:
\begin{align}
&M_6 \nonumber \\
&\leq \mathbb{E}_1 \sup_{t \in [0,T]}\bigg( \int_0^t \bigg| \frac{\partial}{\partial [\vecc{x}^{\epsilon}]_l(s)}\bigg(  [\vecc{a}_{1}\vecc{a}_{2}^{-1}]_{i,j}(s,\vecc{x}^{\epsilon}(s),\epsilon) \bigg) \cdot [\vecc{a}_1]_{l,k}(s,\vecc{x}^\epsilon(s),\epsilon)   \nonumber \\
&\ \  \ \  \ \  \ \cdot [\vecc{J}_3^\epsilon]_{j,k}(s) -
 \frac{\partial}{\partial [\vecc{X}]_l(s)}\left([\vecc{A}_1\vecc{A}_2^{-1}]_{i,j}(s,\vecc{X}(s)) \right) \cdot [\vecc{A}_1]_{l,k}(s,\vecc{X}(s)) \nonumber \\
 &\ \ \ \ \ \  \  \cdot [\vecc{J}]_{j,k}(s) \bigg| ds
\bigg)^p  \\
&\leq C(p) \mathbb{E}_1 \sup_{t \in [0,T]}\bigg( \int_0^t \bigg| \frac{\partial}{\partial [\vecc{x}^{\epsilon}]_l(s)}\bigg(  [\vecc{a}_{1}\vecc{a}_{2}^{-1}]_{i,j}(s,\vecc{x}^{\epsilon}(s),\epsilon) \bigg) \cdot [\vecc{a}_1]_{l,k}(s, \vecc{x}^\epsilon(s),\epsilon) \nonumber \\
&\ \ \ \ \ \  - \frac{\partial}{\partial [\vecc{X}]_l(s)}\bigg([\vecc{A}_1\vecc{A}_2^{-1}]_{i,j}(s,\vecc{X}(s)) \bigg) \cdot [\vecc{A}_1]_{l,k}(s, \vecc{X}(s)) \bigg|^p
\cdot |[\vecc{J}_3^\epsilon]_{j,k}(s)|^p ds \bigg) \nonumber \\
&\ \  + C(p) \mathbb{E}_1 \sup_{t\in [0,T]} \bigg(\int_0^t \bigg| \frac{\partial}{\partial [\vecc{X}]_l(s)}\bigg([\vecc{A}_1\vecc{A}_2^{-1}]_{i,j}(s, \vecc{X}(s)) \bigg) \cdot [\vecc{A}_1]_{l,k}(s, \vecc{X}(s)) \bigg|^p \nonumber \\
&\hspace{3.3cm} \cdot |[\vecc{J}_3^\epsilon-\vecc{J}]_{j,k}(s)|^p ds\bigg)  \\
&\leq  C(p) \mathbb{E}_1 \sup_{t \in [0,T]}\bigg( \int_0^t \bigg| \frac{\partial}{\partial [\vecc{x}^{\epsilon}]_l(s)}\bigg(  [\vecc{a}_{1}\vecc{a}_{2}^{-1}]_{i,j}(s, \vecc{x}^{\epsilon}(s),\epsilon) \bigg) \cdot [\vecc{a}_1]_{l,k}(s, \vecc{x}^\epsilon(s),\epsilon) \nonumber \\
&\ \ \ \ \ \ - \frac{\partial}{\partial [\vecc{X}]_l(s)}\bigg([\vecc{A}_1\vecc{A}_2^{-1}]_{i,j}(s, \vecc{X}(s)) \bigg) \cdot [\vecc{A}_1]_{l,k}(s, \vecc{X}(s)) \bigg|^p
\cdot |[\vecc{J}_3^\epsilon]_{j,k}(s)|^p ds \bigg) \nonumber \\
&\ \ \ + C(p) \mathbb{E}_1 \sup_{t \in [0,T]} \int_0^t \|\vecc{J}_3^\epsilon(s) - \vecc{J}(s) \|_F^p ds, \label{intermed}
\end{align}
where the $C(p)$ are constants may vary from one expression to another.

Note that in the above, $\vecc{J}_3^\epsilon(s)$ and $\vecc{J}(s)$ are solutions to the Lyapunov equation
\begin{equation}
\vecc{a}_2(s, \vecc{x}^\epsilon(s),\epsilon)\vecc{J}_3^\epsilon(s) +\vecc{J}_3^\epsilon(s) \vecc{a}_2^*(s, \vecc{x}^\epsilon(s),\epsilon) = -(\vecc{\sigma}_2 \vecc{\sigma}^*_2)(s, \vecc{x}^\epsilon(s),\epsilon)
\end{equation}
and
\begin{equation}
\vecc{A}_2(s, \vecc{X}(s)) \vecc{J}(s) +\vecc{J}(s) \vecc{A}_2^{*}(s, \vecc{X}(s)) = -(\vecc{\Sigma}_2  \vecc{\Sigma}_2^{*})(s, \vecc{X}(s)).
\end{equation}
respectively.

Let $\vecc{H}^\epsilon(s) := \vecc{J}_3^\epsilon(s) - \vecc{J}(s) $ and $\vecc{G}^\epsilon(s) := \vecc{a}_2(s, \vecc{x}^\epsilon(s),\epsilon)-\vecc{A}_2(s, \vecc{X}(s))$.   After some algebraic manipulations with the above pair of Lyapunov equations, we obtain another Lyapunov equation:
\begin{align}
&\vecc{A}_2(s, \vecc{X}(s)) \vecc{H}^\epsilon(s) + \vecc{H}^\epsilon(s) \vecc{A}_2^*(s, \vecc{X}(s)) \nonumber \\
&=  (\vecc{\Sigma}_2 \vecc{\Sigma}_2^*)(s, \vecc{X}(s)) - (\vecc{\sigma}_2 \vecc{\sigma}_2^*)(s, \vecc{x}^\epsilon(s),\epsilon) - \vecc{G}^\epsilon(s)\vecc{J}_3^\epsilon(s) - \vecc{J}_3^\epsilon(s) (\vecc{G}^\epsilon)^*(s).
\end{align}

By the last statement in Assumption \ref{a5_ch2}, $\vecc{A}_2$ is positive stable uniformly (in $\vecc{X}$ and $s$), therefore the above Lyapunov equation has a unique solution:
\begin{align}
&\vecc{H}^\epsilon(s) = \int_0^\infty e^{\vecc{A}_2(s, \vecc{X}(s)) y} \bigg( -(\vecc{\Sigma}_2 \vecc{\Sigma}_2^*)(s, \vecc{X}(s)) + (\vecc{\sigma}_2 \vecc{\sigma}_2^*)(s, \vecc{x}^\epsilon(s),\epsilon) \nonumber \\
&\hspace{3cm} + \vecc{G}^\epsilon(s)\vecc{J}_3^\epsilon(s) + \vecc{J}_3^\epsilon(s) (\vecc{G}^\epsilon)^*(s)\bigg) e^{\vecc{A}^*_2(s, \vecc{X}(s))y} dy. \label{hdiff}
\end{align}

Using \eqref{hdiff}, the assumptions of the theorem, and estimating as before, we obtain: 
\begin{align}
\mathbb{E}_1 \sup_{t \in [0,T]} \int_0^t  \|\vecc{J}^\epsilon_3(s) - \vecc{J}(s)\|^p_F ds  &\leq C(\epsilon, p, T) \int_0^T \mathbb{E}_1 \sup_{u \in [0,s]} |\vecc{x}^\epsilon(u) - \vecc{X}(u)|^p ds \nonumber \\
&\ \ \ \ \  + D(p,T)[\alpha_2(\epsilon)^p \mathbb{1}_{\vecc{a}_2 \neq \vecc{A}_2} + \gamma_2(\epsilon)^p \mathbb{1}_{\vecc{\sigma}_2 \neq \vecc{\Sigma}_2}]
\end{align}
on the set $S_1$, where $C(\epsilon, p, T) = O(1)$ as $\epsilon \to 0$ and $D(p,T)$ is a positive constant, $\alpha_2(\epsilon)$ and $\gamma_2(\epsilon)$ are from Assumption \ref{a5_ch2}.

Applying the above estimates, Lemma \ref{lipzlemma} and techniques used earlier,  one obtains from \eqref{intermed}:
\begin{align}
M_6 &\leq L_6(\epsilon,p,T) \int_0^T \mathbb{E}_1 \sup_{u \in [0,s]}  |\vecc{x}^\epsilon(u)-\vecc{X}(u) |^p ds \nonumber \\
&\ \ \ \ + C_6(p,T)\bigg[\alpha_1(\epsilon)^p \mathbb{1}_{\{\vecc{a}_1 \neq \vecc{A}_1\}} +  \alpha_2(\epsilon)^p \mathbb{1}_{\{\vecc{a}_2 \neq \vecc{A}_2\}} + \gamma_2(\epsilon)^p \mathbb{1}_{\{\vecc{\sigma}_2 \neq \vecc{\Sigma}_2\}} +   \nonumber \\
&\ \ \ \ \ \ \ \ \ \ +
\theta_1(\epsilon)^p \mathbb{1}_{\{(\vecc{a}_1)_{\vecc{x}} \neq (\vecc{A}_1)_{\vecc{x}}\}} +  \theta_2(\epsilon)^p \mathbb{1}_{\{(\vecc{a}_2)_{\vecc{x}} \neq (\vecc{A}_2)_{\vecc{x}}\}} \bigg],
\end{align}
on $S_1$,  where $L_6(\epsilon,p,T)=O(1)$ as $\epsilon \to 0$, $C_6(p,T)$ is a positive constant, and $\alpha_i(\epsilon)$, $\theta_i(\epsilon)$ ($i=1,2$) and $\gamma_2(\epsilon)$ are from Assumption \ref{a5_ch2}.

Collecting the above estimates for the $M_k$, we obtain: 
\begin{align}
R_1 &\leq C_1(p,T) \bigg( \mathbb{E}_1  \sup_{t \in [0,T]} |\epsilon \vecc{v}^\epsilon(t)|^p \nonumber \\
&\ \ \ \ \ \ + \alpha_1(\epsilon)^p \mathbb{1}_{\{\vecc{a}_1 \neq \vecc{A}_1\}} +  \alpha_2(\epsilon)^p \mathbb{1}_{\{\vecc{a}_2 \neq \vecc{A}_2\}} + \gamma_2(\epsilon)^p \mathbb{1}_{\{\vecc{\sigma}_2 \neq \vecc{\Sigma}_2\}} +   \nonumber \\
&\ \ \ \ \ \ +
\theta_1(\epsilon)^p \mathbb{1}_{\{(\vecc{a}_1)_{\vecc{x}} \neq (\vecc{A}_1)_{\vecc{x}}\}} +  \theta_2(\epsilon)^p \mathbb{1}_{\{(\vecc{a}_2)_{\vecc{x}} \neq (\vecc{A}_2)_{\vecc{x}}\}} \bigg) \nonumber \\
&\ \ \ \ \ \  + C_2(\epsilon, p,T) \int_0^T \mathbb{E}_1 \sup_{u \in [0,s]} |\vecc{x}^\epsilon(u) - \vecc{X}(u)|^p ds + C_3(p,T) M_4 
\end{align}
on $S_1$, where $C_1(p,T)$ and $C_3(p,T)$ are constants, $C_2(\epsilon, p, T) = O(1)$ as $\epsilon \to 0$, and $M_4$ satisfies the bound in \eqref{M4bound}.

Using all the estimates for the $R_i$, we have:
\begin{align}
&\mathbb{E}_1 \left[\sup_{t \in [0,T]} |\vecc{x}^{\epsilon}(t) - \vecc{X}(t)|^p \right]  = \mathbb{E}_1 \left[ \sup_{t \in [0,T]} \sum_{k=1}^{n_1} |[\vecc{x}^{\epsilon}- \vecc{X}]_k(t)|^p  \right]  \\  &\leq n_1 \max_{k=1,\dots,n_1} \left\{ \mathbb{E}_1 \sup_{t \in [0,T]}  |[\vecc{x}^{\epsilon}- \vecc{X}]_k(t)|^p \right\}  \\
&\leq L(\epsilon,p,T,n_1) \int_0^T \mathbb{E}_1 \sup_{u \in [0,s]}  |\vecc{x}^\epsilon(u)-\vecc{X}(u) |^p ds  \nonumber \\
&\ \ \  + C(p,T,n_1)
\bigg( \epsilon^{p r_0} +  \mathbb{E}_1 \sup_{t \in [0,T]} |\epsilon \vecc{v}^\epsilon(t)|^p + M_4 \nonumber \\
&\ \ \ \ \ \ + \alpha_1(\epsilon)^p \mathbb{1}_{\{\vecc{a}_1 \neq \vecc{A}_1\}} +  \alpha_2(\epsilon)^p \mathbb{1}_{\{\vecc{a}_2 \neq \vecc{A}_2\}} + \gamma_1(\epsilon)^p \mathbb{1}_{\{\vecc{\sigma}_1 \neq \vecc{\Sigma}_1\}} \nonumber \\
&\ \ \ \ \ \  +   \gamma_2(\epsilon)^p \mathbb{1}_{\{\vecc{\sigma}_2 \neq \vecc{\Sigma}_2\}}  + \beta_1(\epsilon)^p \mathbb{1}_{\{\vecc{b}_1 \neq \vecc{B}_1\}} +  \beta_2(\epsilon)^p \mathbb{1}_{\{\vecc{B}_2 \neq \vecc{B}_2\}} \nonumber \\
&\ \ \ \ \  \ +\theta_1(\epsilon)^p \mathbb{1}_{\{(\vecc{a}_1)_{\vecc{x}} \neq (\vecc{A}_1)_{\vecc{x}}\}} +  \theta_2(\epsilon)^p \mathbb{1}_{\{(\vecc{a}_2)_{\vecc{x}} \neq (\vecc{A}_2)_{\vecc{x}}\}} \bigg)  \\
&\leq L(\epsilon,p,T,n_1) \int_0^T \mathbb{E}_1 \sup_{u \in [0,s]}  |\vecc{x}^\epsilon(u)-\vecc{X}(u) |^p ds  \nonumber \\
&\ \ \  + C(p,T,n_1)\epsilon^{r},
\end{align}
on $S_1$, where $L(\epsilon, p, T,n_1)=O(1)$ as $\epsilon \to 0$, $r$ is the rate of convergence \eqref{rate_mainresult} in the statement of the theorem, $C(p,T,n_1)$ is a constant that changes from line to line, and we have applied Proposition \ref{mom_bound}, Lemma \ref{lipzlemma} and Assumption \ref{a5_ch2} to get the last expression in the above estimate.

Finally, applying the Gronwall lemma gives:
\begin{align}
&\mathbb{E}_1 \left[\sup_{t \in [0,T]} |\vecc{x}^{\epsilon}(t) - \vecc{X}(t)|^p \right]   \leq  \epsilon^{r} \cdot C(p,T,n_1)  e^{L(\epsilon, p, T,n_1) T}
\end{align}
on $S_1$. 

\eqref{mainconv}  then follows for the case $p > 2$. The result for $0<p\leq 2$ follows by an application of the H\"older's inequality: for  $0<p\leq 2$, taking $q > 2$ so that $p/q < 1$, we have
\begin{align}
\mathbb{E}_1 \left[\sup_{t\in [0,T]} |\vecc{x}^\epsilon(t)-\vecc{X}(t)|^p  \right]  &\leq \bigg[ \mathbb{E}_1  \bigg( \sup_{t\in [0,T]} |\vecc{x}^\epsilon(t)-\vecc{X}(t)|^p \bigg)^{q/p} \bigg]^{p/q} \\
&= O(\epsilon^\beta),
\end{align} for all $0 < \beta < p'$,
as $\epsilon \to 0$.
The last statement on convergence in probabiity in the theorem follows from Lemma 1 in \cite{LimWehr_Homog_NonMarkovian}. 
\end{proof}

\section{An Implementation of Algorithm \ref{alg} under Assumption \ref{special}} \label{implem_alg}

We describe how Algorithm \ref{alg} can be applied to a large class of GLEs, satisfying Assumption \ref{special}.  For $i=2,4$, one can write
\begin{equation}
\vecc{Q}_i(z) = z^{d_i}\vecc{I} + \vecc{a}_{i,d_i-1}z^{d_i-1}+\dots+\vecc{a}_{i,1} z + \vecc{a}_{i,0},
\end{equation}
where the $\vecc{a}_{i,k}$ are related to the $\vecc{\Gamma}_{i,k}$ as follows:
\begin{align}
\vecc{a}_{i,0} &= \prod_{k=1}^{d_i} \vecc{\Gamma}_{i,k}, \nonumber \\
\vecc{a}_{i,1} &= \sum_{k_1, \dots, k_{d_i-1}=1,\dots,d_i: k_1  > \dots > k_{d_i-1}} \vecc{\Gamma}_{i,k_{1}} \vecc{\Gamma}_{i,k_2} \cdots \vecc{\Gamma}_{i,k_{d_i-1}}, \nonumber \\
&\vdots \nonumber \\
\vecc{a}_{i,d_i-2} &= \sum_{k_1,k_2=1,\dots,d_i: k_1>k_2} \vecc{\Gamma}_{i,k_1} \vecc{\Gamma}_{i,k_2}, \nonumber \\
\vecc{a}_{i,d_i-1} &= \sum_{k=1}^N \vecc{\Gamma}_{i,k}.
\end{align}
Then it can be shown that $\vecc{\Phi}_i(z)$ admits the following (controllable) realization \cite{brockett2015finite}: $\vecc{\Phi}_i(z) = \vecc{H}_i(z\vecc{I} + \vecc{F}_i)^{-1}\vecc{G}_i$, with
\begin{equation}
\vecc{H}_i = [\vecc{0} \ \cdots \ \vecc{0} \ \ \vecc{B}_{l_i} \ \  \vecc{0} \ \cdots \ \vecc{0}] \in \RR^{p_i \times p_i d_i},
\end{equation}
where  $\vecc{B}_{l_i}$ is in the $l_i$th slot,
\begin{equation}
\vecc{F}_i =
\begin{bmatrix}
\vecc{0} & -\vecc{I} & \\
 & \vecc{0} & -\vecc{I} & \\
 & & \ddots & \ddots \\
 &  & & \vecc{0} & -\vecc{I} \\
\vecc{a}_{i,0} & \vecc{a}_{i,1} & \dots & \vecc{a}_{i,d_{i-2}} & \vecc{a}_{i,d_{i-1}}
\end{bmatrix} \in \RR^{p_i d_i \times p_i d_i},
\end{equation}
\begin{equation}
\vecc{G}_i = [\vecc{0} \ \cdots \ \vecc{0} \ \ \vecc{I}]^* \in \RR^{p_i d_i}.
\end{equation}
Then the realization of the memory function (for the case $i=2$) and noise process (for the case $i=4$)
can be obtained by taking $\vecc{\Gamma}_i = \vecc{F}_i$, $\vecc{C}_i = \vecc{H}_i$ and solving the following linear matrix inequality:
\begin{equation}
\vecc{F}_i\vecc{M}_i + \vecc{M}_i \vecc{F}_i^* =: \vecc{\Sigma}_i \vecc{\Sigma}_i^* \geq 0, \ \ \vecc{M}_i \vecc{H}_i^* = \vecc{G}_i
\end{equation}
for $\vecc{M}_i = \vecc{M}_i^*$ \cite{willems1980stochastic}.

The above realization gives us the desired spectral densities. Indeed, let us use the transformation of type \eqref{transf_realize} to diagonalize the $\vecc{M}_i$, i.e. $\vecc{M}_i' = \vecc{T}_i \vecc{M}_i \vecc{T}_i^{*} = \vecc{I}$, $\vecc{\Gamma}_{i}' = \vecc{T}_i \vecc{\Gamma}_{i} \vecc{T}_i^{-1}$, $\vecc{\Sigma}_i = \vecc{T}_i \vecc{\Sigma}_i$, $\vecc{C}'_i = \vecc{C}_i \vecc{T}_i^{-1}$.
In this case, for $i=4$ we have: $(\vecc{\xi}^i)'_t  = \vecc{C}_i' (\vecc{\beta}^i)'_t = \vecc{C}_i \vecc{\beta}^i_t = \vecc{\xi}^i_t$,
where $(\vecc{\beta}^i)'_t$ solves the SDE:
\begin{equation}
d(\vecc{\beta}^i)'_t = -\vecc{\Gamma}_i' (\vecc{\beta}^i)'_t dt + \vecc{\Sigma}_i' d\vecc{W}_t^{(q_4)},
\end{equation}
and one can compute the spectral density to be:
\begin{equation}
\vecc{\mathcal{S}}_i(\omega) = \vecc{\Phi}_i(-i\omega)\vecc{\Phi}^*_i(i\omega) =\vecc{B}_{l_i} \omega^{2 l_i} ((\omega^2\vecc{I}+\vecc{\Gamma}_{i,1})^2) \cdots (\omega^2\vecc{I}+\vecc{\Gamma}_{i,d_i})^2)  )^{-1} \vecc{B}_{l_i}^*.
\end{equation}
A similar discussion applies to  the realization of the memory function.

For $i=2,4$,
set $m=\epsilon m_0$, $\vecc{\Gamma}_{i,k} = \vecc{\gamma}_{i,k}/\epsilon$ for $k=l_i+1,\dots,d_i$ and rescale the $\vecc{B}_{l_i}$ with $\epsilon$ accordingly, so that the limit as $\epsilon \to 0$ of the rescaled spectral densities gives us the desired asymptotic behavior. The choice of which and how many of the $\vecc{\Gamma}_{i,k}$ to rescale as well as the smallness of $\epsilon$ (i.e. what determines the wide separation of time scales and their magnitude) depends on the physical system under study. The resulting family of  GLEs can then be cast in a form suitable for application of Theorem \ref{mainthm} and the homogenized SDE for the particle's position can be obtained, under appropriate assumptions on the coefficients of the GLE.  


\section{Another Example for Section \ref{sect_appl}} \label{anothereg}

Consider the case of $l_2=l_4=l=2$, $d_2=d_4=d=3$ in Assumption \ref{special} and specialize to one-dimensional models as before. In this case the covariance function has a stronger singularity near $t=0$ than in cases studied previously.  
The spectral density of the driving noise in the GLE is taken to be:
\begin{equation}
\mathcal{S}(\omega) = \frac{\Gamma_3^2 \beta^2 \omega^4}{(\omega^2+\Gamma_1^2)(\omega^2+\Gamma_2^2)(\omega^2+\Gamma_3^2)},
\end{equation}
in which case the memory kernel (and covariance function) are
\begin{align} \label{w5}
\kappa(t) &= \beta^2 (\Gamma_3 \Gamma_2 + \Gamma_3 \Gamma_1 + \Gamma_2 \Gamma_1) \bigg(\frac{\Gamma_3^4 e^{-\Gamma_3 |t|}}{2(\Gamma_3^2-\Gamma_2^2)(\Gamma_3^2-\Gamma_1^2)(\Gamma_2+\Gamma_1)} \nonumber \\
&\ \ \ \ \ - \frac{\Gamma_3^2 \Gamma_2^2 e^{-\Gamma_2 |t|}}{2(\Gamma_3^2-\Gamma_2^2)(\Gamma_2^2-\Gamma_1^2)(\Gamma_1+\Gamma_3)} + \frac{\Gamma_3^2 \Gamma_1^2 e^{-\Gamma_1 |t|}}{2(\Gamma_3^2-\Gamma_1^2)(\Gamma_2^2-\Gamma_1^2)(\Gamma_3+\Gamma_2)} \bigg),
\end{align}
where $0< \Gamma_1 < \Gamma_2 < \Gamma_3$ (see Figure \ref{fig1} for a plot of $\kappa(t)$). This gives a model for  hyper-diffusion of a particle in a heat bath \cite{siegle2010origin}.

Rescale the parameters by setting $m=m_0 \epsilon$, $\Gamma_3=\gamma_3/\epsilon$, where $m_0$ and $\gamma_3$ are positive constants, and study the limit $\epsilon \to 0$ of the resulting family of GLEs as before.
The resulting rescaled versions of $\kappa(t)$ and $\mathcal{S}(\omega)$ have the following asymptotic behavior as $\epsilon \to 0$:
\begin{align}
\kappa^\epsilon(t)  &\to \beta^2 \delta(t) + \frac{\beta^2}{2(\Gamma_1^2-\Gamma_2^2)}(\Gamma_2^3 e^{-\Gamma_2|t|} - \Gamma_1^3 e^{-\Gamma_1|t|}), \\
\mathcal{S}^\epsilon(\omega) &= \frac{\gamma_3^2 \beta^2 \omega^4}{(\omega^2 + \Gamma_1^2)(\omega_2+\Gamma_2^2)(\epsilon^2 \omega^2 + \gamma_3^2)} \to \frac{\beta^2 \omega^4}{(\omega^2+\Gamma_1^2)(\omega^2+\Gamma_2^2)},
\end{align}

We outline, omitting details, a convergence result similar to  Corollary \ref{w2case}, focusing on the particular case  $g = h$.  In this case, the particle's position, $x^\epsilon_t \in \RR$, converges, as $\epsilon \to 0$, to $X_t$, satisfying the following It\^o SDE system:
\begin{align}
dX_t &= \bigg[ \frac{2F_e}{\beta^2 g^2} + \frac{2}{\beta^2 g}\left(\Gamma_1 \Gamma_2 Z_t^0 + (\Gamma_1+\Gamma_2)Z_t^1 \right) \nonumber \\   
&\ \ \  \ - \frac{2\sigma}{\beta^2 g^2}\left(\Gamma_1 \Gamma_2 Y_t^0 + (\Gamma_1+\Gamma_2)Y_t^1 \bigg) \right] dt \nonumber \\
&\ \ \ \ + \bigg[\frac{2}{\beta^2}\frac{\partial}{\partial X}\left(\frac{1}{g^2} \right) \frac{\sigma^2}{g^2} - \frac{\partial}{\partial X}\left(\frac{1}{g} \right)\frac{4 \sigma^2}{g(g^2\beta^2+4\gamma_3 m_0)} \nonumber \\ 
&\ \ \ \  \ \  \ \  + \frac{\partial}{\partial X}\left(\frac{\sigma}{g^2} \right)\frac{4\sigma}{g^2\beta^2+4\gamma_3 m_0}  \bigg] dt + \frac{2\sigma}{\beta g^2} dW_t^{(1)}, \\
dZ_t^0 &= \left[ -\frac{F_e}{g(\Gamma_1+\Gamma_2)} - \frac{\Gamma_1 \Gamma_2}{\Gamma_1+\Gamma_2} Z_t^0 + \frac{\Gamma_1 \Gamma_2 \sigma}{g(\Gamma_1+\Gamma_2)} Y_t^0 + \frac{\sigma}{g} Y_t^1 \right] dt \nonumber \\
&\ \ \ \ - \frac{1}{\Gamma_1+\Gamma_2} \left[ \frac{\partial}{\partial X}\left(\frac{1}{g} \right)\frac{\sigma^2}{g^2} + \frac{\partial}{\partial X}\left(\frac{\sigma}{g} \right) \frac{2\beta^2 \sigma}{g^2\beta^2+4\gamma_3 m_0}\right]dt \nonumber  \\
&\ \ \ \ \ - \frac{\beta \sigma}{g(\Gamma_1+\Gamma_2)} dW_t^{(1)}, \\
dZ_t^1 &= \left[\frac{F_e}{g}-\frac{\Gamma_1 \Gamma_2 \sigma}{g} Y_t^0 - \frac{\sigma}{g}(\Gamma_1+\Gamma_2) Y_t^1 \right]dt + \frac{\sigma \beta}{g} dW_t^{(1)} \nonumber \\
&\ \ \ \ + \left[\frac{\partial}{\partial X}\left(\frac{1}{g} \right)\frac{\sigma^2}{g^2} + \frac{\partial}{\partial X}\left(\frac{\sigma}{g} \right) \frac{2\beta^2 \sigma}{g^2\beta^2+4\gamma_3 m_0}   \right]dt, \\
dY_t^0 &= Y_t^1 dt, \\
dY_t^1 &= -\Gamma_1 \Gamma_2 Y_t^0 dt - (\Gamma_1+\Gamma_2) Y_t^1 dt + \beta dW_t^{(1)}.
 \end{align}
 
Inspecting the above SDEs in detail, one can make similar remarks to those made in Section \ref{sect_appl}. In particular, taking $g$ to be proportional to $\sigma$ (so that a  fluctuation-dissipation relation holds) again allows us to reduce the number of effective SDEs. Physically, this means that homogenized GLEs for models of hyper-diffusion of particles in a non-equilibrium bath may be highly non-trivial  but they simplify when the fluctuation-dissipation relation is satisfied.


\subsection*{Acknowledgment}
S.H.Lim and J.Wehr were partially supported by the NSF grant DMS 1615045. S.H.Lim is grateful for the support provided by the Michael Tabor Fellowship from the Program in Applied Mathematics at the University of Arizona during the academic year 2017-2018.  M.L. acknowledges the Spanish Ministry MINECO (National Plan 15 Grant: FISICATEAMO No. FIS2016-79508-P, SEVERO OCHOA No. SEV-2015-0522, FPI), European Social Fund, Fundaci\'o Cellex, Generalitat de Catalunya (AGAUR Grant No. 2017 SGR 1341 and CERCA/Program), ERC AdG OSYRIS, EU FETPRO QUIC, and the National Science Centre, Poland-Symfonia Grant No. 2016/20/W/ST4/00314.

\bibliographystyle{plain}
\bibliography{ref1}

\begin{thebibliography}{10}

\bibitem{bao2005non}
Jing-Dong Bao, Peter H{\"a}nggi, and Yi-Zhong Zhuo.
\newblock Non-{M}arkovian {B}rownian dynamics and nonergodicity.
\newblock {\em Physical Review E}, 72(6):061107, 2005.

\bibitem{bao2005harmonic}
Jing-Dong Bao, Yan-Li Song, Qing Ji, and Yi-Zhong Zhuo.
\newblock Harmonic velocity noise: non-{M}arkovian features of noise-driven
  systems at long times.
\newblock {\em Physical Review E}, 72(1):011113, 2005.

\bibitem{bao2003ballistic}
Jing-Dong Bao and Yi-Zhong Zhuo.
\newblock Ballistic diffusion induced by a thermal broadband noise.
\newblock {\em Physical Review Letters}, 91(13):138104, 2003.

\bibitem{birrell2017small}
Jeremiah Birrell, Scott Hottovy, Giovanni Volpe, and Jan Wehr.
\newblock Small mass limit of a {L}angevin equation on a manifold.
\newblock In {\em Annales Henri Poincar{\'e}}, volume~18, pages 707--755.
  Springer, 2017.

\bibitem{birrell2017}
Jeremiah Birrell and Jan Wehr.
\newblock Homogenization of dissipative, noisy, {H}amiltonian dynamics.
\newblock {\em Stochastic Processes and their Applications}, 128(7):2367--2403,
  2018.

\bibitem{2017BirrellLatest}
Jeremiah Birrell and Jan Wehr.
\newblock A homogenization theorem for {L}angevin systems with an application
  to {H}amiltonian dynamics.
\newblock In {\em Sojourns in Probability Theory and Statistical Physics-I},
  pages 89--122. Springer, 2019.

\bibitem{bo2016multiple}
Stefano Bo and Antonio Celani.
\newblock Multiple-scale stochastic processes: decimation, averaging and
  beyond.
\newblock {\em Physics Reports}, 2016.

\bibitem{brockett2015finite}
Roger~W Brockett.
\newblock {\em Finite {D}imensional {L}inear {S}ystems}, volume~74.
\newblock SIAM, 2015.

\bibitem{chevyrev2016multiscale}
Ilya Chevyrev, Peter~K Friz, Alexey Korepanov, Ian Melbourne, and Huilin Zhang.
\newblock Multiscale systems, homogenization, and rough paths.
\newblock In {\em International Conference in Honor of the 75th Birthday of SRS
  Varadhan}, pages 17--48. Springer, 2016.

\bibitem{cordoba2012elimination}
Andr{\'e}s C{\'o}rdoba, Tsutomu Indei, and Jay~D Schieber.
\newblock Elimination of inertia from a generalized {L}angevin equation:
  applications to microbead rheology modeling and data analysis.
\newblock {\em Journal of Rheology}, 56(1):185--212, 2012.

\bibitem{Cui18}
Bingyu Cui and Alessio Zaccone.
\newblock Generalized {L}angevin equation and fluctuation-dissipation theorem
  for particle-bath systems in external oscillating fields.
\newblock {\em Phys. Rev. E}, 97:060102, Jun 2018.

\bibitem{da2014stochastic}
G.~Da~Prato and J.~Zabczyk.
\newblock {\em Stochastic Equations in Infinite Dimensions}.
\newblock Encyclopedia of Mathematics and its Applications. Cambridge
  University Press, 2014.

\bibitem{da1996ergodicity}
Giuseppe Da~Prato and Jerzy Zabczyk.
\newblock {\em Ergodicity for Infinite Dimensional Systems}, volume 229.
\newblock Cambridge University Press, 1996.

\bibitem{dabelow2019irreversibility}
Lennart Dabelow, Stefano Bo, and Ralf Eichhorn.
\newblock Irreversibility in active matter systems: fluctuation theorem and
  mutual information.
\newblock {\em Physical Review X}, 9(2):021009, 2019.

\bibitem{didier2019asymptotic}
Gustavo Didier and Hung Nguyen.
\newblock Asymptotic analysis of the mean squared displacement under fractional
  memory kernels.
\newblock {\em arXiv preprint arXiv:1901.03007}, 2019.

\bibitem{doob1953stochastic}
Joseph~L Doob.
\newblock {\em Stochastic Processes}, volume~7.
\newblock Wiley New York, 1953.

\bibitem{ermak1978brownian}
Donald~L Ermak and JA~McCammon.
\newblock Brownian dynamics with hydrodynamic interactions.
\newblock {\em The Journal of Chemical Physics}, 69(4):1352--1360, 1978.

\bibitem{feller-vol-2}
William Feller.
\newblock {\em An {I}ntroduction to {P}robability {T}heory and {I}ts
  {A}pplications. {V}ol. {II}.}
\newblock Second edition. John Wiley \& Sons Inc., New York, 1971.

\bibitem{froyland2016trajectory}
Gary Froyland, Georg~A Gottwald, and Andy Hammerlindl.
\newblock A trajectory-free framework for analysing multiscale systems.
\newblock {\em Physica D: Nonlinear Phenomena}, 328:34--43, 2016.

\bibitem{givon2004extracting}
Dror Givon, Raz Kupferman, and Andrew Stuart.
\newblock Extracting macroscopic dynamics: model problems and algorithms.
\newblock {\em Nonlinearity}, 17(6):R55, 2004.

\bibitem{glatt2018generalized}
Nathan Glatt-Holtz, David Herzog, Scott McKinley, and Hung Nguyen.
\newblock The {G}eneralized {L}angevin {E}quation with a power-law memory in a
  nonlinear potential well.
\newblock {\em arXiv preprint arXiv:1804.00202}, 2018.

\bibitem{gottwald2015stochastic}
Georg~A Gottwald, Daan~T Crommelin, and Christian~LE Franzke.
\newblock Stochastic climate theory.
\newblock Nonlinear and Stochastic Climate Dynamics. Cambridge University
  Press, 2015.

\bibitem{goychuk2012viscoelastic}
Igor Goychuk.
\newblock Viscoelastic subdiffusion: generalized {L}angevin equation approach.
\newblock {\em Advances in Chemical Physics}, 150:187, 2012.

\bibitem{grebenkov2013hydrodynamic}
Denis~S Grebenkov, Mahsa Vahabi, Elena Bertseva, L{\'a}szl{\'o} Forr{\'o}, and
  Sylvia Jeney.
\newblock Hydrodynamic and subdiffusive motion of tracers in a viscoelastic
  medium.
\newblock {\em Physical Review E}, 88(4):040701, 2013.

\bibitem{hall2016uncertainty}
Eric~J Hall, Markos~A Katsoulakis, and Luc Rey-Bellet.
\newblock Uncertainty quantification for generalized {L}angevin dynamics.
\newblock {\em The Journal of Chemical Physics}, 145(22):224108, 2016.

\bibitem{hartmann2011balanced}
Carsten Hartmann.
\newblock Balanced model reduction of partially observed {L}angevin equations:
  an averaging principle.
\newblock {\em Mathematical and Computer Modelling of Dynamical Systems},
  17(5):463--490, 2011.

\bibitem{herzog2016small}
David~P Herzog, Scott Hottovy, and Giovanni Volpe.
\newblock The small-mass limit for {L}angevin dynamics with unbounded
  coefficients and positive friction.
\newblock {\em Journal of Statistical Physics}, 163(3):659--673, 2016.

\bibitem{hottovy2015smoluchowski}
Scott Hottovy, Austin McDaniel, Giovanni Volpe, and Jan Wehr.
\newblock The {S}moluchowski-{K}ramers limit of stochastic differential
  equations with arbitrary state-dependent friction.
\newblock {\em Communications in Mathematical Physics}, 336(3):1259--1283,
  2015.

\bibitem{indei2012treating}
Tsutomu Indei, Jay~D Schieber, Andr{\'e}s C{\'o}rdoba, and Ekaterina Pilyugina.
\newblock Treating inertia in passive microbead rheology.
\newblock {\em Physical Review E}, 85(2):021504, 2012.

\bibitem{kabanov2013two}
Y.~Kabanov and S.~Pergamenshchikov.
\newblock {\em Two-Scale Stochastic Systems: Asymptotic Analysis and Control}.
\newblock Stochastic Modelling and Applied Probability. Springer Berlin
  Heidelberg, 2013.

\bibitem{karatzas2012Brownian}
Ioannis Karatzas and Steven Shreve.
\newblock {\em Brownian Motion and Stochastic Calculus}, volume 113.
\newblock Springer Science \& Business Media, 2012.

\bibitem{khalfin1958contribution}
LA~Khalfin.
\newblock Contribution to the decay theory of a quasi-stationary state.
\newblock {\em Sov. Phys. JETP}, 6:1053--1063, 1958.

\bibitem{has1966stochastic}
RZ~Khas$'$minskii.
\newblock On stochastic processes defined by differential equations with a
  small parameter.
\newblock {\em Theory of Probability \& Its Applications}, 11(2):211--228,
  1966.

\bibitem{kou2008stochastic}
Samuel~C Kou.
\newblock Stochastic modeling in nanoscale biophysics: subdiffusion within
  proteins.
\newblock {\em The Annals of Applied Statistics}, pages 501--535, 2008.

\bibitem{Kubo_fd}
R~Kubo.
\newblock The fluctuation-dissipation theorem.
\newblock {\em Reports on Progress in Physics}, 29(1):255, 1966.

\bibitem{Kupferman2004}
Raz Kupferman.
\newblock Fractional kinetics in {K}ac--{Z}wanzig heat bath models.
\newblock {\em Journal of Statistical Physics}, 114(1):291--326, 2004.

\bibitem{kurtz1973limit}
Thomas~G Kurtz.
\newblock A limit theorem for perturbed operator semigroups with applications
  to random evolutions.
\newblock {\em Journal of Functional Analysis}, 12(1):55--67, 1973.

\bibitem{lei2016data}
Huan Lei, Nathan~A Baker, and Xiantao Li.
\newblock Data-driven parameterization of the generalized {L}angevin equation.
\newblock {\em Proceedings of the National Academy of Sciences},
  113(50):14183--14188, 2016.

\bibitem{leimkuhler2018ergodic}
Benedict Leimkuhler and Matthias Sachs.
\newblock Ergodic properties of quasi-{M}arkovian generalized {L}angevin
  equations with configuration dependent noise and non-conservative force.
\newblock In {\em International workshop on Stochastic Dynamics out of
  Equilibrium}, pages 282--330. Springer, 2017.

\bibitem{lewenstein1993cooling}
Maciej Lewenstein and Luis Roso.
\newblock Cooling of atoms in colored vacua.
\newblock {\em Physical Review A}, 47(4):3385, 1993.

\bibitem{lewenstein2000quantum}
Maciej Lewenstein and K~Rza̧{\.z}ewski.
\newblock Quantum anti-{Z}eno effect.
\newblock {\em Physical Review A}, 61(2):022105, 2000.

\bibitem{LimWehr_Homog_NonMarkovian}
Soon~Hoe Lim and Jan Wehr.
\newblock Homogenization for a class of generalized {L}angevin equations with
  an application to thermophoresis.
\newblock {\em Journal of Statistical Physics}, 174(3):656--691, 2019.

\bibitem{lindquist2015linear}
A.~Lindquist and G.~Picci.
\newblock {\em Linear Stochastic Systems: A Geometric Approach to Modeling,
  Estimation and Identification}.
\newblock Series in Contemporary Mathematics. Springer Berlin Heidelberg, 2015.

\bibitem{lord2014introduction}
G.J. Lord, C.E. Powell, and T.~Shardlow.
\newblock {\em An Introduction to Computational Stochastic PDEs}.
\newblock Cambridge Texts in Applied Mathematics. Cambridge University Press,
  2014.

\bibitem{lysy2016model}
Martin Lysy, Natesh~S Pillai, David~B Hill, M~Gregory Forest, John~WR Mellnik,
  Paula~A Vasquez, and Scott~A McKinley.
\newblock Model comparison and assessment for single particle tracking in
  biological fluids.
\newblock {\em Journal of the American Statistical Association},
  111(516):1413--1426, 2016.

\bibitem{maes2013langevin}
Christian Maes and Simi~R Thomas.
\newblock From {L}angevin to generalized {L}angevin equations for the
  nonequilibrium {R}ouse model.
\newblock {\em Physical Review E}, 87(2):022145, 2013.

\bibitem{majda2001mathematical}
Andrew~J Majda, Ilya Timofeyev, and Eric Vanden~Eijnden.
\newblock A mathematical framework for stochastic climate models.
\newblock {\em Communications on Pure and Applied Mathematics}, 54(8):891--974,
  2001.

\bibitem{mckinley2018anomalous}
Scott~A McKinley and Hung~D Nguyen.
\newblock Anomalous diffusion and the generalized {L}angevin equation.
\newblock {\em SIAM Journal on Mathematical Analysis}, 50(5):5119--5160, 2018.

\bibitem{mckinley2009transient}
Scott~A McKinley, Lingxing Yao, and M~Gregory Forest.
\newblock Transient anomalous diffusion of tracer particles in soft matter.
\newblock {\em Journal of Rheology (1978-present)}, 53(6):1487--1506, 2009.

\bibitem{metzler2014anomalous}
Ralf Metzler, Jae-Hyung Jeon, Andrey~G Cherstvy, and Eli Barkai.
\newblock Anomalous diffusion models and their properties: non-stationarity,
  non-ergodicity, and ageing at the centenary of single particle tracking.
\newblock {\em Physical Chemistry Chemical Physics}, 16(44):24128--24164, 2014.

\bibitem{morgado2002relation}
Rafael Morgado, Fernando~A Oliveira, G~George Batrouni, and Alex Hansen.
\newblock Relation between anomalous and normal diffusion in systems with
  memory.
\newblock {\em Physical Review Letters}, 89(10):100601, 2002.

\bibitem{mori1965transport}
Hazime Mori.
\newblock Transport, collective motion, and {B}rownian motion.
\newblock {\em Progress of Theoretical Physics}, 33(3):423--455, 1965.

\bibitem{nelson1967dynamical}
Edward Nelson.
\newblock {\em Dynamical Theories of Brownian Motion}, volume~2.
\newblock Princeton University Press, 1967.

\bibitem{2018arXiv180409682N}
Hung~D Nguyen.
\newblock The small-mass limit and white-noise limit of an infinite dimensional
  generalized {L}angevin equation.
\newblock {\em Journal of Statistical Physics}, 173(2):411--437, 2018.

\bibitem{Ottobre}
M.~{Ottobre} and G.~A. {Pavliotis}.
\newblock {Asymptotic analysis for the generalized Langevin equation}.
\newblock {\em Nonlinearity}, 24:1629--1653, May 2011.

\bibitem{papanicolaou1976some}
George~C Papanicolaou.
\newblock Some probabilistic problems and methods in singular perturbations.
\newblock {\em Rocky Mountain Journal of Mathematics}, 6(4), 1976.

\bibitem{Pavliotis}
G.A. Pavliotis and A.M. Stuart.
\newblock {\em Multiscale Methods}, volume~53 of {\em Texts in Applied
  Mathematics}.
\newblock Springer, New York, 2008.

\bibitem{g2005analysis}
Grigorios~A Pavliotis and Andrew~M Stuart.
\newblock Analysis of white noise limits for stochastic systems with two fast
  relaxation times.
\newblock {\em Multiscale Modeling \& Simulation}, 4(1):1--35, 2005.

\bibitem{peres1980nonexponential}
Asher Peres.
\newblock Nonexponential decay law.
\newblock {\em Annals of Physics}, 129(1):33--46, 1980.

\bibitem{picci1992stochastic}
Giorgio Picci.
\newblock Stochastic model reduction by aggregation.
\newblock In {\em Systems, Models and Feedback: Theory and Applications}, pages
  169--177. Springer, 1992.

\bibitem{Picci2011_nicereview}
Giorgio Picci.
\newblock {\em Stochastic Noises, Observation, Identification and Realization
  with}, pages 1672--1688.
\newblock Springer New York, New York, NY, 2011.

\bibitem{reverey2015superdiffusion}
Julia~F Reverey, Jae-Hyung Jeon, Han Bao, Matthias Leippe, Ralf Metzler, and
  Christine Selhuber-Unkel.
\newblock Superdiffusion dominates intracellular particle motion in the
  supercrowded cytoplasm of pathogenic {A}canthamoeba {C}astellanii.
\newblock {\em Scientific Reports}, 5:11690, 2015.

\bibitem{rothe2006violation}
C~Rothe, SI~Hintschich, and AP~Monkman.
\newblock Violation of the exponential-decay law at long times.
\newblock {\em Physical Review Letters}, 96(16):163601, 2006.

\bibitem{safdari2017aging}
Hadiseh Safdari, Andrey~G Cherstvy, Aleksei~V Chechkin, Anna Bodrova, and Ralf
  Metzler.
\newblock Aging underdamped scaled {B}rownian motion: ensemble-and
  time-averaged particle displacements, nonergodicity, and the failure of the
  overdamping approximation.
\newblock {\em Physical Review E}, 95(1):012120, 2017.

\bibitem{sevilla2018non}
Francisco~J Sevilla.
\newblock The non-equilibrium nature of active motion.
\newblock In {\em Quantitative Models for Microscopic to Macroscopic Biological
  Macromolecules and Tissues}, pages 59--86. Springer, 2018.

\bibitem{siegle2010origin}
P~Siegle, I~Goychuk, and P~H{\"a}nggi.
\newblock Origin of hyperdiffusion in generalized {B}rownian motion.
\newblock {\em Physical Review Letters}, 105(10):100602, 2010.

\bibitem{siegle2011markovian}
P~Siegle, I~Goychuk, and P~H{\"a}nggi.
\newblock Markovian embedding of fractional superdiffusion.
\newblock {\em EPL (Europhysics Letters)}, 93(2):20002, 2011.

\bibitem{siegle2010markovian}
Peter Siegle, Igor Goychuk, Peter Talkner, and Peter H{\"a}nggi.
\newblock Markovian embedding of non-{M}arkovian superdiffusion.
\newblock {\em Physical Review E}, 81(1):011136, 2010.

\bibitem{2018superstatistical}
Jakub Slezak, Ralf Metzler, and Marcin Magdziarz.
\newblock Superstatistical generalised {L}angevin equation: non-{G}aussian
  viscoelastic anomalous diffusion.
\newblock {\em New Journal of Physics}, 20(2):023026, 2018.

\bibitem{tavora2016inevitable}
Marco T{\'a}vora, EJ~Torres-Herrera, and Lea~F Santos.
\newblock Inevitable power-law behavior of isolated many-body quantum systems
  and how it anticipates thermalization.
\newblock {\em Physical Review A}, 94(4):041603, 2016.

\bibitem{toda2012statistical}
M.~Toda, R.~Kubo, R.~Kubo, M.~Toda, N.~Saito, N.~Hashitsume, and N.~Hashitsume.
\newblock {\em Statistical Physics II: Nonequilibrium Statistical Mechanics}.
\newblock Springer Series in Solid-State Sciences. Springer Berlin Heidelberg,
  2012.

\bibitem{trentelman2002control}
Harry~L Trentelman, Anton~A Stoorvogel, and Malo Hautus.
\newblock {\em Control Theory for Linear Systems}.
\newblock 2002.

\bibitem{willems1980stochastic}
JC~Willems and JH~Van~Schuppen.
\newblock Stochastic systems and the problem of state space realization.
\newblock In {\em Geometrical Methods for the Theory of Linear Systems:
  Proceedings of a NATO Advanced Study Institute and AMS Summer Seminar in
  Applied Mathematics held at Harvard University, Cambridge, Mass., June
  18--29, 1979}, volume~62, page 283. Springer, 1980.

\bibitem{Wei2018}
Wei Zhong, Debabrata Panja, Gerard~T. Barkema, and Robin~C. Ball.
\newblock Generalized {L}angevin equation formulation for anomalous diffusion
  in the {I}sing model at the critical temperature.
\newblock {\em Phys. Rev. E}, 98:012124, Jul 2018.

\bibitem{Zwanzig1973}
Robert Zwanzig.
\newblock Nonlinear generalized {L}angevin equations.
\newblock {\em Journal of Statistical Physics}, 9(3):215--220, 1973.

\end{thebibliography}

\end{document}